\documentclass{article}

\usepackage{blindtext}
\usepackage[abbrvbib, preprint]{style}
\usepackage{lastpage}
\firstpageno{1}
\usepackage[letterpaper,margin=1in]{geometry}
\usepackage[T1]{fontenc}
\usepackage[utf8]{inputenc}

\usepackage{natbib}
\usepackage{caption}
\usepackage{subcaption}

\usepackage{algorithm}
\usepackage{algpseudocode}

\usepackage[percent]{overpic}

\usepackage{newfloat}
\usepackage{listings}
\DeclareCaptionStyle{ruled}{
  labelfont=normalfont,
  labelsep=colon,
  strut=off
}
\floatstyle{ruled}
\newfloat{listing}{tb}{lst}{}
\floatname{listing}{Listing}

\usepackage{booktabs}
\usepackage{tabularx}
\usepackage{array}

\usepackage{amsmath,amssymb,amsthm,mathtools}
\usepackage{multirow}
\usepackage{enumitem}
\usepackage{microtype}
\usepackage{mdframed}
\usepackage{placeins}

\usepackage[export]{adjustbox}
\usepackage[table]{xcolor}

\newtheoremstyle{boxedplain}
{4pt}{4pt}
{\itshape}
{}
{\bfseries}
{.}
{0.5em}
{\thmname{#1}~\thmnumber{#2}\thmnote{ (#3)}}

\theoremstyle{boxedplain}

\newmdtheoremenv[
linecolor=black!45,
linewidth=0.6pt,
backgroundcolor=white,
roundcorner=0pt,
innertopmargin=6pt,
innerbottommargin=6pt,
innerleftmargin=6pt,
innerrightmargin=6pt
]{assumption}{Assumption}

\newmdtheoremenv[
linecolor=black!45,
linewidth=0.6pt,
backgroundcolor=white,
roundcorner=0pt,
innertopmargin=6pt,
innerbottommargin=6pt,
innerleftmargin=6pt,
innerrightmargin=6pt
]{theorem}{Theorem}

\newmdtheoremenv[
linecolor=black!45,
linewidth=0.6pt,
backgroundcolor=white,
roundcorner=0pt,
innertopmargin=6pt,
innerbottommargin=6pt,
innerleftmargin=6pt,
innerrightmargin=6pt
]{proposition}{Proposition}

\newtheorem{lemma}{Lemma}

\definecolor{skyblue}{RGB}{135,206,235}

\newif\ifincludeappendix
\includeappendixtrue

\title{
SurfSpec: Enhancing Off-Target-Agnostic Specificity
by Bounding Pocket-Ligand Geometric Mismatch
}

\author{\name Minyeong Hwang\thanks{Equal Contribution.}\\
\addr Kim Jaechul Graduate School of AI, KAIST
\AND
\name Yoorim Gang\footnotemark[1]\\
\addr Interdisciplinary Program in Artificial Intelligence, Seoul National University
\AND
\name Ziseok Lee\\
\addr Department of Biomedical Sciences, Seoul National University
\AND
\name Wooyeol Lee\\
\addr Department of Biomedical Sciences, Seoul National University
\AND
\name Young Bin Park\\
\addr Calici
\AND
\name Jae-Mun Choi\\
\addr Calici
\AND
\name Kyungsu Kim\footnotemark[2]\\
\addr School of Transdisciplinary Innovations, Interdisciplinary Program in Artificial Intelligence, and Department of Biomedical Sciences, Seoul National University
\AND
\name Eunho Yang\thanks{Correspondence to: Eunho Yang \texttt{eunhoy@kaist.ac.kr}, Kyunsu Kim \texttt{kyskim@snu.ac.kr}.}\\
\addr Kim Jaechul Graduate School of AI, School of Computing, KAIST, AITRICS
}

\date{}

\newcommand{\Spec}{\mathrm{Spec}}
\newcommand{\dgm}{d_{\mathrm{gm}}}
\newcommand{\Tgt}{P_\mathrm{tgt}}
\newcommand{\Off}{\mathcal{O}}
\newcommand{\surf}{\mathcal{S}}

\newcommand{\xanc}{x_{\mathrm{label}}}

\begin{document}

\maketitle

\begin{abstract}
Lead optimization in structure-based drug design aims to improve target
binding while avoiding unintended interactions with off-target pockets.
However, existing affinity-driven methods do not explicitly control
specificity, whereas current specificity-aware approaches commonly
require prior knowledge of off-target structures.
We address off-target-agnostic specificity-aware lead optimization by
analyzing the geometric mismatch between a ligand and the target pocket.
We provide a conservative specificity lower bound for geometrically
separated off-targets without requiring access to off-target structures.
By metricizing pocket--ligand mismatch, the triangle inequality shows
that reducing target--ligand mismatch improves a conservative lower bound
on mismatch to a separated off-target class, which can be translated into
a specificity lower bound through an empirical geometry--affinity
calibration.
Motivated by this analysis, we introduce \textbf{SurfSpec}, an
off-target-agnostic lead optimization framework that iteratively grows
ligands toward under-occupied regions of the target pocket surface.
SurfSpec alternates between linker generation toward selected
target-surface patches, which provides geometric pseudo-labels, and
refinement under a pocket-conditioned ligand prior, which restores these
pseudo-labels into valid ligands.
On the CrossDocked2020 test set, SurfSpec reduces geometric mismatch and
outperforms evaluated off-target-agnostic lead optimization baselines in
empirical specificity, while maintaining competitive target-affinity
improvement.

\end{abstract}


\section{Introduction}
Lead optimization in structure-based drug design (SBDD) seeks to improve
an initial ligand for a target pocket.
Beyond increasing target-pocket affinity, an ideal optimized ligand
should also avoid strong interactions with undesirable off-target pockets,
thereby achieving ligand specificity.
Many existing lead optimization methods
\cite{pmdm,diffleop,delete,decompopt,ace}, however, primarily optimize
affinity to the target pocket.
As highlighted by recent work~\cite{sbe-diff}, high target affinity alone
does not guarantee specificity, because an optimized ligand may also bind
strongly to unintended pockets.
Such off-target interactions can reduce selectivity and cause adverse
effects~\cite{keiser}, motivating lead optimization methods that
explicitly account for ligand specificity.

Recent specificity-aware methods address this limitation using explicit
negative supervision, either through off-target pocket information or
activity labels for undesired targets.
SBE-Diff~\cite{sbe-diff}, for example, learns specific binding energy
from positive and in-batch negative pocket--ligand pairs, whereas
ActivityDiff~\cite{activitydiff} employs separately trained drug--target
activity classifiers for desired target activity and undesired off-target
activity.
Although these methods demonstrate the value of incorporating specificity
into molecular optimization, their applicability depends on the
availability and representativeness of negative target information.
In practice, potential off-target interactions form an open-ended set
that may include unknown, structurally unavailable, or out-of-distribution
targets.
Larger surrogate negative sets and stronger activity predictors may
improve coverage of observed interactions, but they cannot directly
account for off-targets that are absent during training or unavailable at
inference time.
This motivates an off-target-agnostic formulation of specificity-aware
lead optimization, in which only the target pocket is used during
optimization.

Instead of relying on off-target supervision, we seek a target-only
geometric principle that can be analyzed without accessing off-targets
during optimization.
We define geometric mismatch between a ligand and a pocket using the
Jensen--Shannon distance between surface-induced probability measures.
Because this mismatch is metricized, the triangle inequality provides a
conservative lower bound on ligand mismatch to geometrically separated
off-target pockets using only the target pocket and the current ligand.
Together with an empirical geometry--affinity calibration, this yields a
target-only lower-bound analysis for specificity over separated
off-target classes.

Motivated by this analysis, we propose \textbf{SurfSpec}, an
off-target-agnostic lead optimization framework that progressively
expands an initial ligand toward under-occupied regions of the target
pocket surface.
SurfSpec alternates between surface-directed pseudo-label generation and
molecular-prior-based recovery.
At each iteration, it identifies a nearby target-surface patch that is
insufficiently covered by the current ligand and constructs a pseudo-label
by extending the ligand toward the selected region.
A pretrained pocket-conditioned diffusion prior then recovers a valid
ligand that preserves the intended geometric modification.
Through this iterative procedure, SurfSpec improves coverage of
accessible target-pocket regions and encourages target-pocket surface
complementarity while maintaining consistency with the learned molecular
distribution.

To keep pseudo-label recovery simple and stable within the ligand-growth
pipeline, SurfSpec uses a task-specific low-noise recovery procedure. General inverse-problem solvers~\cite{dps,flowchef,fastdips,reddiff} can be applied to this task, but they often rely on either clean-state posterior approximations along the reverse trajectory~\cite{dps,flowchef} or multiple dependent optimization stages~\cite{fastdips,reddiff}.
The former can be unreliable in the high-noise regime, while the latter
increases the orchestration cost of refinement.
SurfSpec is motivated by the localized nature of the anchored recovery
task: the desired ligand is expected to remain close to the geometric
pseudo-label while lying on the highly structured pocket-conditioned
ligand distribution, making the corresponding anchored distribution
locally concentrated.
Based on this observation, SurfSpec performs gradient-flow mode recovery
for a localized pseudo-label-conditioned distribution at a single
intermediate noise level, re-noises around the recovered mode, and
completes sampling with a low-noise guided reverse SDE.
This design avoids high-noise anchor-conditioned estimation and reduces
dependence on repeated correction stages, serving as a practical recovery
component for surface-directed ligand growth rather than a new
general-purpose refinement principle.

\begin{figure}[t]
\centering
\includegraphics[
    width=0.8\columnwidth,
    trim={60 105 60 20},
    clip,
    keepaspectratio
]{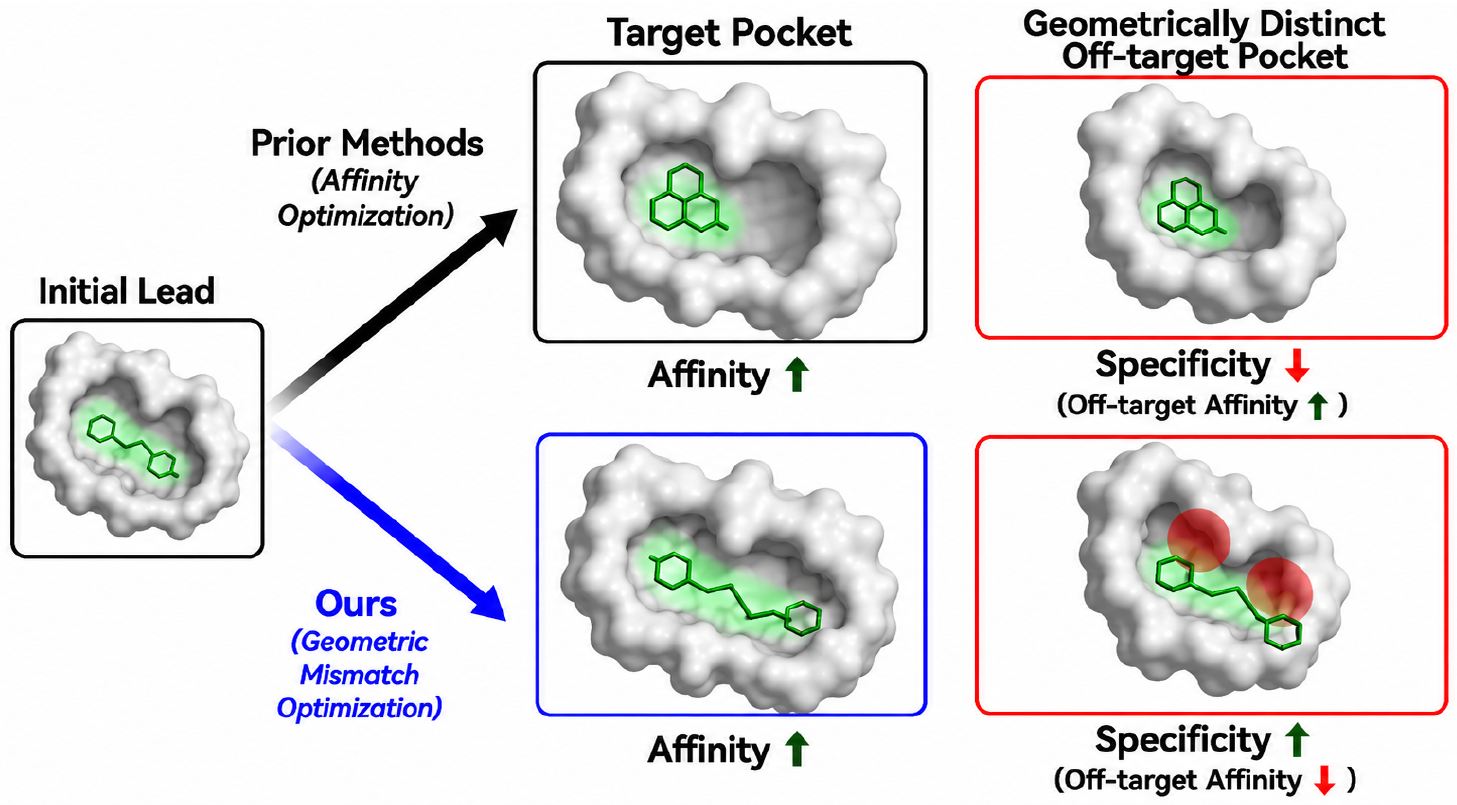}
\caption{
Conceptual comparison of off-target-agnostic lead optimization.
Existing target-only lead optimization methods can improve affinity to the native target pocket, but may also increase compatibility with off-target pockets, resulting in limited specificity improvement.
In contrast, SurfSpec optimizes the geometric fit to the native pocket, improving target affinity while reducing compatibility with geometrically distinct off-targets.
}
\label{fig:concept_specificity}
\end{figure}

On the CrossDocked2020 test set~\cite{crossdock2020}, SurfSpec achieves
higher empirical specificity than evaluated off-target-agnostic lead
optimization baselines while maintaining competitive target-pocket
affinity, as shown in Table~\ref{tab:main}.
It also outperforms a surrogate off-target-aware baseline guided by
randomly selected off-target sets, highlighting the risk of unreliable
specificity guidance from inaccurate off-target surrogates.
Further analyses show that SurfSpec improves pocket coverage and
geometric fit, and recovers surface-directed pseudo-labels with strong
validity and fidelity.

Our main contributions are as follows:
\begin{itemize}
\item We formulate off-target-agnostic specificity-aware lead optimization
through a target-only geometric mismatch metric, defined as the
Jensen--Shannon distance between probability measures induced by the
ligand surface and the ligand-oriented target-pocket surface.
Under a geometric separation condition on the off-target class, the
triangle inequality shows that reducing target--ligand mismatch improves
a conservative lower bound on ligand mismatch to separated off-targets,
which is translated into a specificity lower bound through an empirical
geometry--affinity calibration.

\item We introduce \textbf{SurfSpec}, an off-target-agnostic lead
optimization framework motivated by this target-only geometric principle.
SurfSpec iteratively grows ligands toward under-occupied regions of the
target pocket surface by alternating linker generation toward selected
surface patches and task-specific low-noise anchored recovery under a
pocket-conditioned ligand prior.

\item We evaluate SurfSpec against off-target-agnostic lead optimization
baselines using held-out off-target pockets that are never accessed
during optimization.
On CrossDocked2020, SurfSpec achieves higher empirical specificity while
maintaining competitive target-pocket docking performance, with improved
pocket coverage and geometric fit.
\end{itemize}

\section{Related Work}
\label{sec:related_work}

\paragraph{Lead Optimization and Molecular Specificity.}
Recent SBDD models support pocket-conditioned molecular generation and
lead optimization through autoregressive generation, equivariant
diffusion models, and controllable molecular optimization
\cite{pocket2mol,targetdiff,alidiff,diffsbdd,pmdm,decompopt,delete}.
Methods such as PMDM, DecompOpt, and Delete
\cite{pmdm,decompopt,delete} optimize an initial ligand or ligand
fragment for target-pocket affinity, drug-likeness, and validity, but do
not explicitly optimize molecular specificity.
More recent approaches, including SBE-Diff and ActivityDiff
\cite{sbe-diff,activitydiff}, incorporate negative pockets or activity
signals to improve specificity through off-target supervision.
SurfSpec instead considers a target-only setting that improves specificity
without requiring off-target supervision.

\paragraph{Geometric Complementarity.}
Classical docking methods \cite{docking1,docking2} exploit pocket--ligand
surface complementarity to identify favorable poses.
Protein-pocket-conditioned SBDD models
\cite{pocket2mol,targetdiff,diffsbdd,decompopt,pmdm,d3fg} use target
pocket geometry for molecular generation, elaboration, or optimization.
Other geometry-conditioned methods use reference molecular shapes,
multimodal ligand profiles, or learned geometric representations, as in
Diff-Shape \cite{diffshape}, ShEPhERD \cite{shepherd}, and GeoRCG
\cite{georcg}, while UniGuide \cite{uniguide} applies user-specified
geometric conditions through training-free inference-time guidance.
However, these approaches do not explicitly formulate pocket--ligand
surface complementarity as a metricized mismatch objective linked to
molecular specificity.
SurfSpec is oriented by this mismatch and connects geometric
complementarity to molecular specificity optimization.

\paragraph{Training-Free Sample Refinement.}
Existing training-free sample refinement methods can be grouped into three paradigms.
\emph{Reverse-path} methods, including SDEdit, DDNM, MCG, DPS, \(\Pi\)GDM, FAST-DIPS, FlowChef, and FlowDPS \cite{sdedit,ddnm,mcg,dps,gdm,fastdips,flowchef,flowdps}, initialize refinement from an intermediate noisy state and subsequently follow a reverse SDE/ODE trajectory.
SDEdit \cite{sdedit} constructs this state by directly perturbing an input anchor, whereas conditional reverse-path methods \cite{mcg,dps,gdm,flowdps} incorporate observations or constraints through guidance terms computed from denoised estimates.
\emph{Optimization-based} methods, such as DiffPIR, DDS, DAPS, DCDP, PnP-Flow, FLOWER, FlowLPS, and FAST-DIPS \cite{diffpir,dds,daps,dcdp,pnpflow,flower,flowlps,fastdips}, interleave prior-based denoising with explicit consistency-oriented optimization, projection, or proximal correction.
\emph{Variational posterior} methods, including RED-Diff, RSD, and FLAIR \cite{reddiff,rsd,flair}, instead optimize a parameterized variational distribution toward the target posterior.
These approaches provide general and flexible mechanisms for training-free
refinement across inverse problems and continuous data domains.
In contrast, SurfSpec adopts a simpler task-specific refinement strategy
tailored to geometrically constructed molecular pseudo-labels: it performs
anchored recovery at a low noise level under a pocket-conditioned molecular
prior, reducing reliance on high-noise clean-state posterior approximation
or multi-stage posterior optimization.

\section{Preliminaries}
\subsection{Problem Setup: Specific Lead Optimization without Off-Target Access}
\label{sec:prel_problem_setup}

We consider lead optimization for a target pocket \(\Tgt\) when off-target
information is unavailable.
A ligand is represented as \(L=(X_L,H_L)\), where
\(X_L\in\mathbb{R}^{3\times N_L}\) and
\(H_L\in\{0,1\}^{n_{\mathrm{types}}\times N_L}\) denote its coordinates
and atom-type features, respectively.
A pocket \(P=(X_P,H_P)\) is defined analogously.
For the methodology derivation in Section~\ref{sec:methodology}, we use
a rigid-body abstraction and identify ligands and pockets up to separate
rigid transformations:
\(L\sim\rho_g(L)\) and \(P\sim\rho_h(P)\) for
\(g,h\in\mathrm{SE}(3)\).

Given the target pocket \(\Tgt\) and a set of possible off-target pockets
\(\Off\), we define specificity as the affinity gap between the target
and the strongest off-target:
\begin{equation}
\Spec(L;\Tgt,\Off)
=
A_{\Tgt}(L)
-
\max_{O\in\Off}A_O(L),
\label{eq:spec}
\end{equation}
where \(A_P(L)\) denotes the affinity between ligand \(L\) and pocket
\(P\), with larger values indicating stronger binding.
In the methodology derivation, we instantiate \(A_P(L)\) under the
rigid-body abstraction as
\(\max_{g\in\mathrm{SE}(3)} a(\rho_g(P),L)\), where \(a\) is the
pose-level affinity.
In experiments, however, we evaluate affinity using the standard
AutoDock Vina docking protocol, allowing ligand torsional changes for
practical evaluation.
A ligand is therefore considered specific when it achieves strong
target-pocket affinity without exhibiting similarly strong affinity for
competing off-target pockets.

\paragraph{Specific Lead Optimization without Off-Target Access.}
Directly optimizing Eq.~\eqref{eq:spec} is generally impractical because
the true off-target set \(\Off\) is unavailable during lead optimization.
Existing methods therefore either optimize target-pocket affinity alone
or approximate the off-target set using a predefined collection of
pockets.
Methods that optimize target-pocket affinity alone often yield limited
specificity improvement, because increasing target-pocket affinity can
also increase off-target affinity~\cite{sbe-diff}.
Methods that rely on a predefined off-target set~\cite{sbe-diff,activitydiff}
explicitly account for specificity, but may fail to capture challenging
off-targets outside the constructed set.
We therefore seek a target-only optimization principle that improves
specificity without requiring access to off-target pockets, where the
guarantees apply to off-target pockets that are geometrically
distinguishable from the target pocket.

\begin{figure}[t]
\centering

\includegraphics[width=0.8\columnwidth]{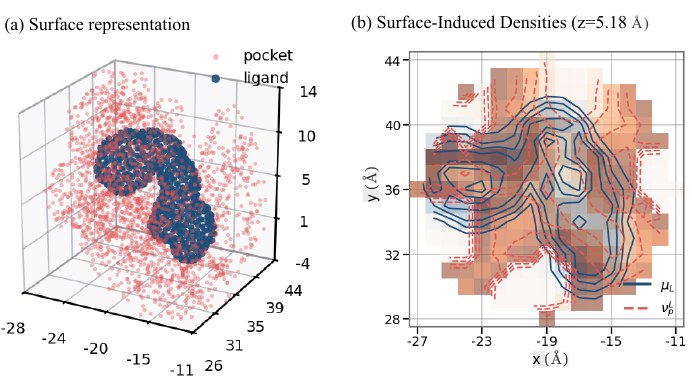}

\caption{
Surface and measure representations used for geometric mismatch.
(a) The ligand surface \(\surf_L\) and the ligand-oriented pocket surface
\(\surf_P^L\), obtained by cropping the aligned pocket surface around the
ligand.
(b) A two-dimensional slice of the signed-distance-based Boltzmann
densities induced by \(\surf_L\) and \(\surf_P^L\), which define the
probability measures \(\mu_L\) and \(\nu_P^L\), respectively.
The geometric mismatch \(\dgm(L,P)\) is computed as the Jensen--Shannon
distance between these induced probability measures.
}
\label{fig:surface_representation}
\end{figure}

\subsection{Jensen--Shannon Distance and Metric Structure}
\label{sec:prel_js}

The Jensen--Shannon distance measures the discrepancy between two
probability distributions through a symmetrized and smoothed KL
divergence.
For probability measures \(\mu\) and \(\nu\) defined on a common domain
\(\Omega\), let
\[
\bar\mu=\frac{1}{2}(\mu+\nu).
\]
The Jensen--Shannon distance is defined as
\begin{equation}
\begin{aligned}
d_{\mathrm{JS}}(\mu,\nu)
=
\Bigg[
\frac{1}{2}
\operatorname{KL}(\mu\|\bar\mu)
+
\frac{1}{2}
\operatorname{KL}(\nu\|\bar\mu)
\Bigg]^{1/2}.
\end{aligned}
\label{eq:js_distance}
\end{equation}
Unlike an arbitrary divergence, the square-root Jensen--Shannon
divergence is a metric on probability distributions and therefore
satisfies the triangle inequality:
\begin{equation}
d_{\mathrm{JS}}(\mu,\rho)
\le
d_{\mathrm{JS}}(\mu,\nu)
+
d_{\mathrm{JS}}(\nu,\rho).
\label{eq:js_triangle}
\end{equation}



\begin{figure}[t]
\centering
\setlength{\abovecaptionskip}{2pt}
\includegraphics[
    width=1.0\columnwidth,
    trim={0.02in 0.05in 0.02in 0.02in},
    clip
]{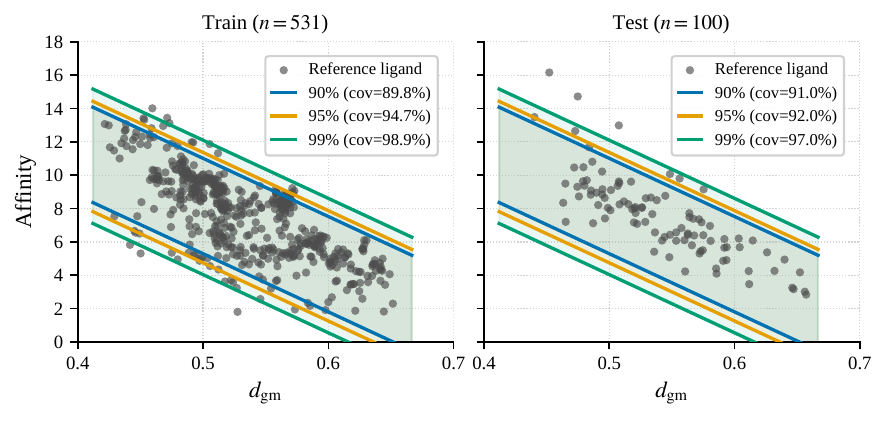}
\caption{
Geometry--affinity calibration on CrossDocked2020.
Gray points denote reference ligand--pocket pairs, where affinity is
measured as the negative AutoDock Vina score and \(d_{\mathrm{gm}}\) is
the proposed geometric mismatch.
Lines show the fitted linear geometry--affinity trend with central
90\%, 95\%, and 99\% empirical quantile envelopes constructed from
training residuals.
}
\label{fig:gm_affinity_quantile}
\end{figure}

\begin{figure*}[t]
\centering
\includegraphics[
    width=1.0\textwidth,
    trim={0 0 0 0},
    clip,
    keepaspectratio
]{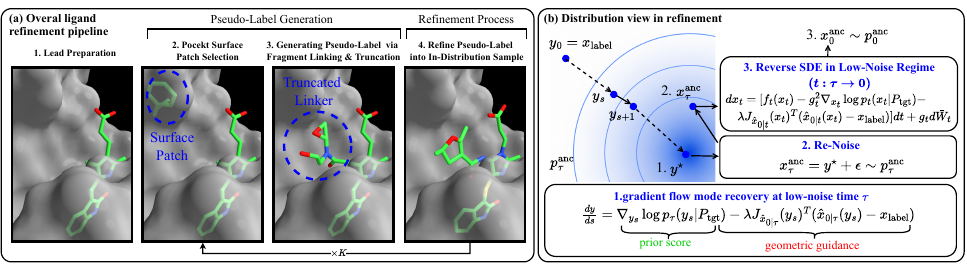}
\caption{
Overview of the ligand refinement pipeline.
The left panel illustrates iterative pocket-surface-targeted ligand growth, while the right panel presents the prior-guided refinement process that maps an out-of-distribution pseudo-label toward an in-distribution ligand sample.
}
\label{fig:ours_pipeline}
\end{figure*}


\section{Methodology}
\label{sec:methodology}
\subsection{Geometric Mismatch Bounds Specificity against Geometrically Distinct Off-Targets}
\label{sec:geo_mismatch}

We first analyze metricized geometric mismatch between the target pocket
and the ligand as a target-only principle for specificity-aware lead
optimization.
The key intuition is that a ligand with small mismatch to the target
pocket is encouraged to fit the native pocket, while its compatibility
with geometrically distinct off-target pockets can be conservatively
bounded through the metric structure.
This objective is not intended to replace affinity optimization, but
rather to optimize a geometric component that contributes to molecular
recognition and favorable binding.
Motivated by classical geometric docking and shape-complementarity
studies~\cite{docking1,docking2}, we formalize this principle by defining
a computable mismatch metric, empirically validating its relationship
with binding affinity, and deriving a specificity lower-bound analysis.

\paragraph{Definition of geometric mismatch.}
We first express each pocket in the ligand-centered affinity-optimal
frame.
For ligand \(L\) and pocket \(P\), define
\(P^L=\rho_{g^\star(L,P)}(P)\), where
\(g^\star(L,P)\in\arg\max_{g\in\mathrm{SE}(3)}a(\rho_g(P),L)\).
Let \(A_L=\{a_i\}_{i=1}^{N_L}\) denote the heavy-atom coordinates of
\(L\), and define the ligand-local domain
\(\Omega_L=\{x\in\mathbb{R}^3:\min_i\|x-a_i\|_2\le 8\,\text{\AA}\}\).
We define the ligand surface \(\surf_L\) as the van der Waals surface of
the ligand heavy atoms.
For the aligned pocket \(P^L\), we define the ligand-oriented pocket
surface \(\surf_P^L\) as the van der Waals surface of \(P^L\) restricted
to surface points within \(8\,\text{\AA}\) of any ligand heavy atom.

To compare these surfaces robustly, we convert \(\surf_L\) and
\(\surf_P^L\) into probability measures on the common domain
\(\Omega_L\).
Let \(\mu_L\) and \(\nu_P^L\) denote the probability measures induced by
\(\surf_L\) and \(\surf_P^L\), respectively, through the
signed-distance-based Boltzmann densities defined in
Appendix~\ref{app:surface_probability_field}.
We define geometric mismatch as the Jensen--Shannon (JS) distance between the
two measures:
\begin{equation}
\begin{aligned}
\dgm(L,P)
&=
d_{\mathrm{JS}}(\mu_L,\nu_P^L)
\\
&=
\left[
\frac{1}{2}\operatorname{KL}(\mu_L\|\bar\mu)
+
\frac{1}{2}\operatorname{KL}(\nu_P^L\|\bar\mu)
\right]^{1/2},
\end{aligned}
\label{eq:js_geometric_mismatch}
\end{equation}
where \(\bar\mu=(\mu_L+\nu_P^L)/2\).
A smaller \(\dgm(L,P)\) indicates stronger agreement between the ligand
surface field and the ligand-oriented pocket surface field. We visualize the surfaces and probablity measures in Figure~\ref{fig:surface_representation}.  

\paragraph{Correlation with affinity.}
We formalize the empirical relationship between geometric mismatch and
binding affinity through a quantile-calibrated geometry--affinity
envelope.
This formulation accounts for residual affinity variation caused by
non-geometric factors, such as electrostatics, hydrogen bonding, and
desolvation, while retaining the principle that geometric mismatch
constrains the attainable affinity range with high probability.

\begin{assumption}[Quantile-calibrated geometry--affinity envelope]
\label{assump:calibrated}
Let \(\mathcal{D}_{\mathrm{CD}}\) denote the empirical ligand--pocket
distribution induced by CrossDocked2020.
For a tail probability \(\eta\in(0,1)\), there exist monotone decreasing
functions \(f_-^{\eta}\) and \(f_+^{\eta}\) such that, for
\((L,P)\sim\mathcal{D}_{\mathrm{CD}}\),
\begin{equation}
\hspace*{-0.45em}
\begin{aligned}
\Pr\!\Big[
& f_-^{\eta}(\dgm(L,P))
\le
A_P(L)
\le
f_+^{\eta}(\dgm(L,P))
\Big]
\ge
1-\eta .
\end{aligned}
\label{eq:quantile_calibration}
\end{equation}
\end{assumption}

Figure~\ref{fig:gm_affinity_quantile} shows this validation.
We fit a linear trend \(\hat{a}(d)=\beta_0+\beta_1 d\) using the
random pairs from the CrossDocked2020 training split
\((n=531)\), and define the central \((1-\eta)\)-quantile envelope by
\(f_-^\eta(d)=\hat{a}(d)+q_{\eta/2}\) and
\(f_+^\eta(d)=\hat{a}(d)+q_{1-\eta/2}\), where \(q_\alpha\) is the
empirical \(\alpha\)-quantile of the training residuals.
The fitted 90\%, 95\%, and 99\% envelopes cover 89.8\%, 94.7\%, and
98.9\% of the training samples, respectively, and cover 91.0\%, 92.0\%,
and 97.0\% of the 100 CrossDocked2020 test samples.

\paragraph{Specificity certificate for geometrically separated off-targets.}
We analyze off-target pockets that remain geometrically distinguishable
from the target pocket under ligand-oriented surface representations.
Let \(\mathcal{L}_{\mathrm{opt}}\) denote the set of ligands reachable by
optimization.
For \(L\in\mathcal{L}_{\mathrm{opt}}\) and \(O\in\Off\), let
\(m_L(O)=d_{\mathrm{JS}}(\nu_{\Tgt}^{L},\nu_O^{L})\) denote the
ligand-oriented geometric separation between the target and off-target
pocket surfaces, where \(\nu_{\Tgt}^{L}\) and \(\nu_O^{L}\) are the
ligand-oriented pocket measures induced by the corresponding cropped
pocket surfaces.
We call \(\Off\) a \(\delta\)-separated off-target class if
\[
\inf_{L\in\mathcal{L}_{\mathrm{opt}}}
\inf_{O\in\Off}
m_L(O)
\ge
\delta .
\]
This condition requires every considered off-target pocket to remain at
least \(\delta\) away from the target pocket in the ligand-oriented metric
throughout optimization.

\begin{theorem}[Specificity certificate for separated off-targets]
\label{thm:spec_bound}
Let \(L\in\mathcal{L}_{\mathrm{opt}}\) and
\(\epsilon=\dgm(L,\Tgt)\).
Suppose that \(\Off\) is a \(\delta\)-separated off-target class.
Under Assumption~\ref{assump:calibrated}, with probability at least
\(1-(|\Off|+1)\eta\),
\begin{equation}
\Spec(L;\Tgt,\Off)
\ge
f_-^\eta(\epsilon)
-
f_+^\eta\!\left([\delta-\epsilon]_+\right).
\label{eq:spec_bound_ligand_invariant}
\end{equation}
Consequently, this high-probability specificity lower bound weakly
increases as the target mismatch \(\epsilon\) decreases.
\end{theorem}

The theorem formalizes the target-only role of geometric mismatch.
Although off-target pockets are not accessed during optimization, their
geometric separation from the target pocket provides a specificity
certificate: as the ligand better matches the target surface, the
certified lower bound on its separation from geometrically distinct
off-targets weakly increases.
The proof applies the triangle inequality of \(d_\mathrm{JS}\) to the
ligand-induced measure, the target-pocket measure, and each off-target
pocket measure, yielding
\(\dgm(L,O)\ge[\delta-\epsilon]_+\).
The quantile-calibrated geometry--affinity envelope then converts this
geometric certificate into an affinity-gap certificate, where the target
affinity is lower bounded by \(f_-^\eta(\epsilon)\) and each off-target
affinity is upper bounded by
\(f_+^\eta([\delta-\epsilon]_+)\).
A union bound over the target and off-target pockets gives the stated
high-probability lower bound.
The complete proof is provided in Appendix~\ref{app:proof_specificity}.

\subsection{Ligand Growth via Geometric Pseudo-Labeling and Refinement}
\label{sec:growth_refinement}

Motivated by the target-only specificity analysis, we design SurfSpec as
an iterative ligand-growth pipeline that improves target-pocket surface
fit without accessing off-target pockets during optimization.
SurfSpec uses geometric mismatch as a practical design principle and
operationalizes it through surface-directed ligand growth toward
under-occupied regions of the target pocket.

As illustrated in Figure~\ref{fig:ours_pipeline}, SurfSpec alternates
between two stages.
First, it generates a geometric pseudo-label by extending the current
ligand toward a nearby under-occupied target-pocket surface region.
Second, it refines this pseudo-label into a valid ligand using a
pretrained diffusion prior over the pocket-conditioned ligand
distribution.
Repeating these stages encourages coverage of accessible target-pocket
regions while maintaining consistency with the learned ligand
distribution.
The full procedure is described in Algorithm~\ref{alg:full_ligand_growth}.

\begin{table*}[t]
    \centering
    \small
    \setlength{\tabcolsep}{2.2pt}
    \renewcommand{\arraystretch}{1.1}
    \resizebox{\textwidth}{!}{%
    \begin{tabular}{lccccccccccccccccc}
    \toprule
    \multirow{2}{*}{Method}
    & \multicolumn{1}{c}{Off-Target}
    & \multicolumn{5}{c}{$\widehat{\mathrm{Spec}}$ $\uparrow$}
    & \multicolumn{3}{c}{Geometric Mismatch $\downarrow$}
    & \multicolumn{3}{c}{Occupancy $\uparrow$}
    & \multicolumn{3}{c}{Vina Dock $\downarrow$}
    & \multirow{2}{*}{\shortstack{Clash\\Rate $\downarrow$}}
    & \multirow{2}{*}{\#\,Atoms} \\
    \cmidrule(lr){3-7}
    \cmidrule(lr){8-10}
    \cmidrule(lr){11-13}
    \cmidrule(lr){14-16}
    & Agnostic
    & Avg. & Std. & $>{0.2}$ & $>{0.4}$ & $>{0.6}$
    & Avg. & Std. & Med.
    & Avg. & Q1 & Q3
    & Avg. & Q1 & Q3
    &  &  \\
    \midrule
    Initial Lead
    & O
    & -0.83 & 1.39 & \underline{0.15} & 0.13 & 0.11
    & 0.540 & 0.051 & 0.536
    & 0.26 & 0.22 & 0.30
    & -7.17 & -8.67 & -5.73
    & \textbf{0.00} & 22.75 \\
    \hline
    Delete
    & O
    & -0.84 & 1.35 & \underline{0.15} & 0.12 & 0.11
    & 0.539 & 0.052 & 0.535
    & 0.26 & 0.22 & 0.30
    & -7.23 & -8.65 & -5.73
    & \textbf{0.00} & 22.85 \\
    PMDM
    & O
    & -1.16 & 2.48 & 0.12 & 0.09 & 0.08
    & 0.501 & 0.056 & 0.489
    & \underline{0.33} & 0.27 & 0.38
    & -6.47 & -8.87 & -5.10
    & 0.16 & 31.00 \\
    DecompOpt
    & O
    & -0.97 & 1.79 & 0.12 & 0.09 & 0.07
    & 0.542 & 0.053 & 0.536
    & 0.25 & 0.22 & 0.29
    & -6.08 & -8.26 & -4.49
    & 0.02 & 22.75 \\
    DiffShape
    & O
    & -0.97 & 1.33 & 0.13 & 0.09 & 0.06
    & 0.585 & 0.115 & 0.555
    & 0.14 & 0.00 & 0.25
    & -6.90 & -8.33 & -5.20
    & 0.11 & 21.91 \\
    ActivityDiff
    & X
    & -0.89 & 1.42 & \underline{0.15} & 0.13 & 0.12
    & 0.538 & 0.054 & 0.529
    & 0.27 & 0.23 & 0.31
    & -7.07 & -8.35 & -5.39
    & 0.11 & 22.89 \\
    \hline
    PMDM--MS2
    & O
    & -1.15 & 2.65 & 0.12 & 0.11 & 0.11
    & 0.499 & 0.063 & 0.482
    & 0.31 & 0.26 & 0.37
    & \textbf{-8.95} & -10.04 & \textbf{-7.01}
    & 0.25 & 33.42 \\
    PMDM--MS3
    & O
    & -2.11 & 2.66 & 0.07 & 0.05 & 0.05
    & 0.496 & 0.069 & 0.480
    & 0.31 & 0.26 & 0.37
    & -8.38 & \textbf{-10.51} & -6.43
    & 0.31 & 35.83 \\
    PMDM--MS+ES
    & O
    & -0.90 & 1.98 & \underline{0.15} & \underline{0.15} & \underline{0.14}
    & 0.500 & 0.057 & 0.489
    & \underline{0.33} & \underline{0.28} & 0.38
    & -8.61 & -10.14 & \underline{-6.56}
    & 0.16 & 31.37 \\
    DecompOpt--MS2
    & O
    & -1.00 & 1.26 & 0.11 & 0.09 & 0.07
    & 0.542 & 0.053 & 0.538
    & 0.25 & 0.22 & 0.28
    & -7.08 & -8.79 & -5.62
    & 0.03 & 22.75 \\
    DecompOpt--MS3
    & O
    & -0.91 & 1.41 & 0.12 & 0.10 & 0.09
    & 0.541 & 0.053 & 0.537
    & 0.25 & 0.21 & 0.28
    & -7.00 & -8.57 & -5.72
    & 0.04 & 22.76 \\
    DecompOpt--MS+ES
    & O
    & \underline{-0.80} & 1.39 & 0.14 & 0.11 & 0.09
    & 0.541 & 0.052 & 0.534
    & 0.25 & 0.22 & 0.28
    & -7.40 & -8.91 & -5.95
    & 0.03 & 22.75 \\
    \hline
    DiffSBDD--SizeExt(+20)
    & O
    & -1.71 & 1.84 & 0.05 & 0.05 & 0.05
    & \underline{0.480} & 0.037 & 0.476
    & 0.28 & 0.19 & 0.38
    & -7.57 & -8.25 & -6.37
    & 0.45 & 35.88 \\
    DiffSBDD--SizeExt(+30)
    & O
    & -1.81 & 1.97 & 0.03 & 0.03 & 0.02
    & \textbf{0.468} & 0.046 & \textbf{0.461}
    & 0.32 & 0.23 & \textbf{0.42}
    & -7.14 & -8.40 & -5.44
    & 0.62 & 41.58 \\
    \hline
    \rowcolor{blue!8}
    Ours
    & O
    & \textbf{-0.71} & 1.84 & \textbf{0.18} & \textbf{0.16} & \textbf{0.15}
    & \underline{0.480} & 0.050 & \underline{0.467}
    & \textbf{0.36} & \textbf{0.30} & \underline{0.41}
    & \underline{-8.63} & \underline{-10.43} & -6.31
    & \textbf{0.00} & 34.99 \\
    \bottomrule
    \end{tabular}%
    }
    \caption{
    Off-target-agnostic lead optimization analysis.
    Bold indicates the unique best; underline marks all second-best values when the best is unique.
    }
    \label{tab:main}
    \end{table*}

\paragraph{Geometric pseudo-label generation.}
We generate a pseudo-label by exploiting a pretrained linker generation model \(p_{\mathrm{linker}}\). At iteration \(k\), we select the closest unoccupied surface patch \(r^{(k)}\) whose distance from the current ligand \(x^{(k)}\) falls within \(4\)--\(10\) \AA. The detailed filtering rule for unoccupied surface patches is provided in Appendix \ref{app:surface_patch_selection}. Given \(r^{(k)}\), we sample a linker that connects the current ligand and the selected surface patch. The sampled linker is truncated to exclude pocket-clashing atoms falling within the \(2\) \AA{} clashing threshold, forming a geometric pseudo-label \(\xanc^{(k)}\). The resulting \(\xanc^{(k)}\) serves as a directed anchor for reducing the target pocket mismatch.

\paragraph{Pseudo-label refinement.}
Given the pseudo-label \(\xanc\), the refinement task is to recover a valid
ligand with the same number of atoms that remains close to \(\xanc\) while
staying likely under the pocket-conditioned ligand prior.
Formally, we target the anchored clean density
\begin{equation}
p_0^{\mathrm{anc}}(x\mid\Tgt,\xanc)
\propto
p_0(x\mid\Tgt)
\exp\!\left(
-\frac{\lambda}{2}
\|x-\xanc\|^2
\right),
\label{eq:anchored_density}
\end{equation}
where \(p_0(x\mid\Tgt)\) is the pretrained ligand prior and
\(\lambda>0\) controls anchor fidelity.
Let \(p_t^{\mathrm{anc}}\) denote the distribution obtained by diffusing
\(p_0^{\mathrm{anc}}\) to time \(t\).

General-purpose inverse-problem solvers discussed in
Section~\ref{sec:related_work} can be applied to this refinement task, but
their design choices are not tailored to geometrically constructed
molecular pseudo-labels.
Optimization-based and variational posterior methods often introduce
multiple coupled correction or optimization stages, which require careful
coordination of interaction hyperparameters in this molecular setting.
Reverse-path methods avoid such multi-stage optimization, but typically
depend on either denoised clean-state approximations, such as the DPS
plug-in approximation
\(
p(x_0\mid y,\Tgt,\tau)
\approx
\delta_{\hat{x}_{0\mid\tau}(y)}(x_0),
\)
or direct noising of the pseudo-label as in SDEdit.
These approximations can become unreliable when applied at high noise
levels or when the geometrically constructed pseudo-label lies outside
the typical support of the pretrained molecular prior.

We therefore use a simpler task-specific low-noise recovery procedure
that avoids multi-stage posterior optimization and avoids applying
clean-state posterior approximation in the high-noise regime.
The procedure is motivated by the local concentration of the anchored
distribution under the combined constraints of pocket-conditioned
validity and similarity to the pseudo-label.
In Appendix~\ref{app:proof_refinement}, we formalize this intuition by
assuming that the anchored clean distribution is concentrated in a small
ball around its mean.
Under this condition, at a low-noise time \(\tau\) where the local DPS
score approximation in Assumption~\ref{assump:local_dps_score_approx}
holds, Proposition~\ref{prop:mode_recovery} shows that the anchored score
flow converges to a neighborhood of the mode of
\(p_\tau^{\mathrm{anc}}\), with residual error controlled by the
implemented score approximation error.
We therefore estimate this low-noise anchored point using the anchored
score flow
\begin{equation}
\begin{aligned}
\frac{d y_s}{d s}
={}&
\nabla_{y_s}\log p_\tau(y_s\mid\Tgt)
\\
&-
\lambda
J_{\hat{x}_{0\mid\tau}}(y_s)^\top
\left(
\hat{x}_{0\mid\tau}(y_s)-\xanc
\right).
\end{aligned}
\label{eq:low_noise_mode_recovery}
\end{equation}

Let \(y^\star\) denote the recovered point.
Proposition~\ref{prop:renoising_recovered_mode} further shows that
re-noising around this recovered point,
\begin{equation}
x_\tau^{\mathrm{anc}}
=
y^\star+\sigma_\tau\varepsilon,
\qquad
\varepsilon\sim\mathcal{N}(0,I),
\label{eq:renoise_anchor}
\end{equation}
approximates sampling from the low-noise anchored marginal
\(p_\tau^{\mathrm{anc}}\).
Starting from \(x_\tau^{\mathrm{anc}}\), we complete sampling with a
low-noise guided reverse SDE:
\begin{equation}
\begin{aligned}
d x_t
={}&
\Big[
f_t(x_t)
-
g_t^2
\{
\nabla_{x_t}\log p_t(x_t\mid\Tgt)
\\
&\qquad
-
\lambda
J_{\hat{x}_{0\mid t}}(x_t)^\top
\left(
\hat{x}_{0\mid t}(x_t)-\xanc
\right)
\}
\Big]dt
+
g_t\,d\bar{W}_t ,
\end{aligned}
\label{eq:normal_reverse_sde}
\end{equation}
where \(f_t\) and \(g_t\) denote the drift and diffusion coefficients of
the forward SDE defining the pretrained pocket-conditioned ligand prior
\(p_t(\cdot\mid\Tgt)\), and \(\bar{W}_t\) is the reverse-time Brownian
motion. This low-noise reverse process approximately samples from the anchored
clean distribution \(p_0^{\mathrm{anc}}\).

The refined ligand is passed to the next growth iteration, and the
pipeline alternates between geometric expansion and refinement until the
maximum number of iterations is reached or the target-pocket Vina score
after local optimization increases relative to the previous iteration.

\begin{table*}[t]
\centering
\small
\renewcommand{\arraystretch}{0.90}
\setlength{\abovecaptionskip}{3pt}
\setlength{\belowcaptionskip}{0pt}
\setlength{\tabcolsep}{3.2pt}
\begin{adjustbox}{max width=\textwidth}
\begin{tabular}{lccccccccc|cccccc}
\toprule
\multirow{3}{*}{Method} & \multicolumn{9}{c|}{\textbf{Prior Consistency}} & \multicolumn{6}{c}{\textbf{Pseudo-Label Faithfulness}} \\
\cmidrule(lr){2-10}\cmidrule(lr){11-16}
 & \multicolumn{3}{c}{Pocket Clash $\downarrow$} & \multicolumn{2}{c}{Target Affinity $\uparrow$} & \multirow{2}{*}{Validity $\uparrow$} & \multicolumn{3}{c|}{Bond MMD $\downarrow$} & \multicolumn{3}{c}{RMSD $\downarrow$} & \multicolumn{3}{c}{Topo. Sim. $\uparrow$} \\
\cmidrule(lr){2-4}\cmidrule(lr){5-6}\cmidrule(lr){8-10}\cmidrule(lr){11-13}\cmidrule(lr){14-16}
 & \%<2\AA & Local Rep. & Global Rep. & Local & Global &  & C-C & C-N & \multicolumn{1}{c|}{C-O} & Avg. & Q1 & Q3 & Avg. & Q1 & Q3 \\
\midrule
Pseudo label & 0.00 & 1.46 & 7.61 & 0.23 & 6.75 & 0.84 & 0.06 & 0.01 & 0.09 & 0.00 & 0.00 & 0.00 & 1.00 & 1.00 & 1.00 \\
\hline
\addlinespace[2pt]
DCDP & \underline{0.02} & 0.83 & 5.38 & \underline{0.43} & \underline{7.69} & \textbf{1.00} & \textbf{0.08} & \textbf{0.03} & \underline{0.10} & \underline{0.62} & \underline{0.55} & \underline{0.67} & 0.48 & 0.35 & 0.62 \\
DPS & 0.03 & \textbf{0.20} & \underline{4.88} & 0.39 & 6.49 & \textbf{1.00} & 0.11 & 0.04 & 0.11 & 6.17 & 5.44 & 6.79 & 0.35 & 0.25 & 0.46 \\
SDEdit & 0.03 & 1.29 & 6.31 & 0.28 & 7.13 & \textbf{1.00} & 0.09 & \textbf{0.03} & 0.13 & \textbf{0.31} & \textbf{0.28} & \textbf{0.35} & 0.53 & 0.36 & 0.68 \\
Red-Diff & 0.51 & 5.04 & 42.39 & 0.14 & 1.67 & \textbf{1.00} & 0.12 & 0.05 & 0.16 & 1.47 & 0.97 & 2.00 & \textbf{0.95} & \textbf{0.99} & \textbf{1.00} \\
RSD & 0.77 & 9.54 & 106.75 & 0.12 & 0.69 & \textbf{1.00} & 0.11 & 0.04 & 0.15 & 2.77 & 1.98 & 3.62 & \underline{0.93} & \underline{0.91} & \textbf{1.00} \\
\rowcolor{blue!8}
Ours & \textbf{0.01} & \underline{0.33} & \textbf{4.40} & \textbf{0.46} & \textbf{8.00} & \textbf{1.00} & \textbf{0.08} & \textbf{0.03} & \textbf{0.08} & 0.78 & 0.66 & 0.90 & 0.48 & 0.34 & 0.64 \\
\bottomrule
\end{tabular}
\end{adjustbox}
\caption{Refinement results. Bold indicates best; underline indicates second-best among methods (only when best is unique).}
\label{tab:refinement}
\end{table*}

\section{Experiments}
\label{sec:experiments}

\subsection{Experimental Setup}
\label{sec:experimental_setup}

\paragraph{Benchmark and Implementation.}
We evaluate SurfSpec on the 100 test complexes from
CrossDocked2020~\cite{crossdock2020}.
For each complex, the reference ligand paired with the target pocket is
used as the initial lead, and each method generates one optimized
molecule.
SurfSpec runs at most \(K=3\) ligand-growth iterations, using
DiffLinker~\cite{e3quivariant} as the linker generation model
\(p_{\mathrm{linker}}\) for surface-directed pseudo-label generation and
DiffSBDD~\cite{diffsbdd} as the pretrained pocket-conditioned ligand
prior \(p(\cdot\mid P_{\mathrm{tgt}})\) for refinement.
The low-noise recovery time is fixed to \(\tau=0.15\).
Off-target pockets are never accessed by off-target-agnostic methods
during optimization; for evaluation, each target is paired with an
off-target set \(\Off_{\mathrm{eval}}\) of 10 randomly sampled pockets
from the remaining CrossDocked2020 test pockets.
For refinement validation, we keep the pseudo-label and pretrained prior
fixed, replace only the pseudo-label refinement step with existing inverse-problem
solvers, and set the outer ligand-growth budget to \(K=1\) so that each
method performs a single optimization step for each CrossDocked2020 test
complex.

\paragraph{Baselines and Metrics.}
We compare SurfSpec with representative lead optimization and editing
baselines, including Delete, PMDM, DecompOpt, DiffShape, ActivityDiff,
multi-step PMDM/DecompOpt variants, and DiffSBDD-based size-extension
baselines~\cite{delete,pmdm,decompopt,diffshape,activitydiff}.
We evaluate empirical specificity using target and off-target AutoDock
Vina scores, and additionally report target-pocket Vina score,
target-pocket geometric mismatch, pocket occupancy, clash rate, and
ligand size.
For refinement validation, we report metrics for prior consistency and
pseudo-label faithfulness.
Detailed baseline construction, hyperparameters, and metric definitions
are provided in Appendices~\ref{app:refinement_details}
and~\ref{app:eval_metrics}

\subsection{Performance on Specificity Enhancement}
\label{subsec:exp2}

Table~\ref{tab:main} reports lead optimization performance on
CrossDocked2020.
SurfSpec achieves the best empirical specificity among all evaluated
methods, including the highest average score and the highest thresholded
success rates.
It also achieves the best pocket occupancy and a low target-pocket
geometric mismatch, indicating that target-surface-directed growth
effectively increases coverage of the target-pocket region.
At the same time, SurfSpec maintains a strong target-pocket Vina docking
score and zero clash rate.
These results support that SurfSpec improves the empirical
target-over-off-target preference while preserving favorable target-pocket
docking behavior, without accessing off-targets during
optimization.

The augmented PMDM and DecompOpt variants with repeated optimization or
early stopping do not match SurfSpec in empirical specificity, suggesting
that the improvement is not explained solely by repeated baseline
application or by the stopping criterion.
The DiffSBDD-based size-extension baselines achieve low geometric
mismatch, but they also produce high clash rates and weak empirical
specificity.
This indicates that reducing geometric mismatch through naive size growth
is insufficient and can lead to invalid pocket occupation.
In contrast, SurfSpec improves target-pocket occupancy and empirical
specificity while maintaining zero clash rate, suggesting that
surface-directed growth provides a more controlled way to exploit target
geometry.

\paragraph{Non-vacuity of the specificity certificate.}
In Appendix~\ref{app:certificate_nonvacuity}, we evaluate the certificate
margin
\(M_L(O)=d_{\mathrm{JS}}(\nu_{\Tgt}^{L},\nu_O^{L})-\dgm(L,\Tgt)\), where
\(d_{\mathrm{JS}}(\nu_{\Tgt}^{L},\nu_O^{L})\) is the ligand-oriented
target--off-target pocket separation.
When \(M_L(O)>0\), the triangle-inequality lower bound on
ligand--off-target mismatch is nonzero.
Empirically, \(65.96\%\) of the randomly sampled target--off-target pairs
in the CrossDocked2020 test set have positive margin, indicating that the
certificate is non-vacuous for a substantial fraction of evaluated pairs
while remaining conservative.

\subsection{Validation of Pseudo-Label Refinement}
\label{subsec:refinement_validation}

Table~\ref{tab:refinement} evaluates refinement quality under a shared
pseudo-label setting. SurfSpec achieves strong prior consistency, including the lowest clash
rate among refinement methods, lowest global repulsion, best
target-affinity diagnostics, perfect valence validity, and best or
tied-best bond MMD values.
At the same time, it preserves reasonable pseudo-label faithfulness, with
RMSD substantially lower than DPS, Red-Diff, and RSD and topological
similarity comparable to DCDP.
These results support the proposed low-noise anchored recovery as a
simple task-specific refinement method for recovering geometrically
constructed molecular pseudo-labels under a pocket-conditioned ligand
prior.

\section{Conclusion}
\label{sec:conclusion}

We studied off-target-agnostic specificity-aware lead optimization, where
only the target pocket is available during optimization.
We introduced a target-only geometric analysis based on metricized
pocket--ligand surface mismatch, showing that reducing target mismatch
weakly improves a conservative specificity lower bound for geometrically
separated off-targets.
Motivated by this analysis, we proposed \textbf{SurfSpec}, a
surface-directed ligand-growth framework that expands ligands toward
under-occupied target-pocket regions.
On CrossDocked2020, SurfSpec achieves higher empirical specificity,
improves target-pocket occupancy, and maintains competitive docking
performance with low geometric mismatch and zero clash rate.
Additional analyses are provided in Appendix~\ref{app:additional_results}.

\section{Acknowledgement}
This work was partly supported by the KHIDI grant funded by the Korean government (MOHW) [No.RS-2025-02307233, No.RS-2026-25613012], the NRF or IITP grants funded by the Korean government (MSIT) [No.05-26-04-0094, No.RS-2026-25472075, No.RS-2025-02305581, No.RS-2025-25442338, and No.RS-2021-II211343], the ITIP grant funded by the Korean government (MOTIR) [No.RS-2026-25549946], the Research grant from SNU, Strategic Hub grant for International Research Collaboration of SNU, and Institute for Information \& communications Technology Planning \& Evaluation(IITP) grant funded by the Korea government(MSIT)
(RS-2019-II190075, Artificial Intelligence Graduate School
Program(KAIST)). Kyungsu Kim is affiliated with the School of Transdisciplinary Innovations, Department of Biomedical Science, Interdisciplinary Program in Artificial Intelligence (IPAI), Medical Research Center, and AI Institute at SNU.

\bibliography{bibliograpy}

\begin{thebibliography}{39}
\providecommand{\natexlab}[1]{#1}
\providecommand{\url}[1]{\texttt{#1}}
\expandafter\ifx\csname urlstyle\endcsname\relax
  \providecommand{\doi}[1]{doi: #1}\else
  \providecommand{\doi}{doi: \begingroup \urlstyle{rm}\Url}\fi

\bibitem[Adams et~al.(2025)Adams, Abeywardane, Fromer, and Coley]{shepherd}
K.~Adams, K.~Abeywardane, J.~Fromer, and C.~W. Coley.
\newblock {ShEPhERD}: Diffusing shape, electrostatics, and pharmacophores for bioisosteric drug design.
\newblock In \emph{The Thirteenth International Conference on Learning Representations}, 2025.
\newblock URL \url{https://openreview.net/forum?id=KSLkFYHlYg}.

\bibitem[Ayadi et~al.(2024)Ayadi, Hetzel, Sommer, Theis, and G\"{u}nnemann]{uniguide}
S.~Ayadi, L.~Hetzel, J.~Sommer, F.~Theis, and S.~G\"{u}nnemann.
\newblock Unified guidance for geometry-conditioned molecular generation.
\newblock In A.~Globerson, L.~Mackey, D.~Belgrave, A.~Fan, U.~Paquet, J.~Tomczak, and C.~Zhang, editors, \emph{Advances in Neural Information Processing Systems}, volume~37, pages 138891--138924. Curran Associates, Inc., 2024.
\newblock \doi{10.52202/079017-4407}.
\newblock URL \url{https://proceedings.neurips.cc/paper_files/paper/2024/file/faa6276ea12d7afeb3e42b210c86f688-Paper-Conference.pdf}.

\bibitem[Chen et~al.(2025)Chen, Zhang, Jiang, Zhao, Zhang, Chen, Liu, Su, Wu, Wang, Qu, Ye, Chai, Wang, Wang, An, Wu, Yang, Chen, Xie, Lin, Li, Hsieh, Huang, Kang, Hou, and Pan]{delete}
S.~Chen, O.~Zhang, C.~Jiang, H.~Zhao, X.~Zhang, M.~Chen, Y.~Liu, Q.~Su, Z.~Wu, X.~Wang, W.~Qu, Y.~Ye, X.~Chai, N.~Wang, T.~Wang, Y.~An, G.~Wu, Q.~Yang, J.~Chen, W.~Xie, H.~Lin, D.~Li, C.-Y. Hsieh, Y.~Huang, Y.~Kang, T.~Hou, and P.~Pan.
\newblock Deep lead optimization enveloped in protein pocket and its application in designing potent and selective ligands targeting ltk protein.
\newblock \emph{Nature Machine Intelligence}, 7\penalty0 (3):\penalty0 448--458, 2025.
\newblock \doi{10.1038/s42256-025-00997-w}.
\newblock URL \url{https://doi.org/10.1038/s42256-025-00997-w}.

\bibitem[Chung et~al.(2022)Chung, Sim, Ryu, and Ye]{mcg}
H.~Chung, B.~Sim, D.~Ryu, and J.~C. Ye.
\newblock Improving diffusion models for inverse problems using manifold constraints.
\newblock In S.~Koyejo, S.~Mohamed, A.~Agarwal, D.~Belgrave, K.~Cho, and A.~Oh, editors, \emph{Advances in Neural Information Processing Systems}, volume~35, pages 25683--25696. Curran Associates, Inc., 2022.
\newblock URL \url{https://proceedings.neurips.cc/paper_files/paper/2022/file/a48e5877c7bf86a513950ab23b360498-Paper-Conference.pdf}.

\bibitem[Chung et~al.(2023)Chung, Kim, McCann, Klasky, and Ye]{dps}
H.~Chung, J.~Kim, M.~T. McCann, M.~L. Klasky, and J.~C. Ye.
\newblock Diffusion posterior sampling for general noisy inverse problems.
\newblock In \emph{International Conference on Learning Representations}, 2023.
\newblock URL \url{https://openreview.net/forum?id=OnD9zGAGT0k}.

\bibitem[Erbach et~al.(2025)Erbach, Narnhofer, Dombos, Schiele, Lenssen, and Schindler]{flair}
J.~Erbach, D.~Narnhofer, A.~Dombos, B.~Schiele, J.~E. Lenssen, and K.~Schindler.
\newblock Solving inverse problems with flair.
\newblock In D.~Belgrave, C.~Zhang, H.~Lin, R.~Pascanu, P.~Koniusz, M.~Ghassemi, and N.~Chen, editors, \emph{Advances in Neural Information Processing Systems}, volume~38, pages 136667--136704. Curran Associates, Inc., 2025.
\newblock URL \url{https://proceedings.neurips.cc/paper_files/paper/2025/file/c7ae6e9659f0c99582c2e8214ba0b413-Paper-Conference.pdf}.

\bibitem[Francoeur et~al.(2020)Francoeur, Masuda, Sunseri, Jia, Iovanisci, Snyder, and Koes]{crossdock2020}
P.~G. Francoeur, T.~Masuda, J.~Sunseri, A.~Jia, R.~B. Iovanisci, I.~Snyder, and D.~R. Koes.
\newblock Three-dimensional convolutional neural networks and a cross-docked data set for structure-based drug design.
\newblock \emph{J. Chem. Inf. Model.}, 60\penalty0 (9):\penalty0 4200--4215, 2020.
\newblock \doi{10.1021/acs.jcim.0c00411}.
\newblock URL \url{https://doi.org/10.1021/acs.jcim.0c00411}.

\bibitem[Gao et~al.(2024)Gao, Ren, Ni, Huang, Qiang, Ma, Ma, and Lan]{sbe-diff}
B.~Gao, M.~Ren, Y.~Ni, Y.~Huang, B.~Qiang, Z.-M. Ma, W.-Y. Ma, and Y.~Lan.
\newblock Rethinking specificity in {SBDD}: Leveraging delta score and energy-guided diffusion.
\newblock In \emph{Proceedings of the 41st International Conference on Machine Learning}, volume 235 of \emph{Proceedings of Machine Learning Research}, pages 14811--14825. PMLR, 2024.
\newblock URL \url{https://proceedings.mlr.press/v235/gao24k.html}.

\bibitem[Gu et~al.(2024)Gu, Xu, Powers, Nie, Geffner, Kreis, Leskovec, Vahdat, and Ermon]{alidiff}
S.~Gu, M.~Xu, A.~Powers, W.~Nie, T.~Geffner, K.~Kreis, J.~Leskovec, A.~Vahdat, and S.~Ermon.
\newblock Aligning target-aware molecule diffusion models with exact energy optimization.
\newblock In A.~Globerson, L.~Mackey, D.~Belgrave, A.~Fan, U.~Paquet, J.~Tomczak, and C.~Zhang, editors, \emph{Advances in Neural Information Processing Systems}, volume~37, pages 44040--44063. Curran Associates, Inc., 2024.
\newblock \doi{10.52202/079017-1398}.
\newblock URL \url{https://proceedings.neurips.cc/paper_files/paper/2024/file/4ddfe69f164eae70abc86f0f9cbed7e8-Paper-Conference.pdf}.

\bibitem[Guan et~al.(2023)Guan, Qian, Peng, Su, Peng, and Ma]{targetdiff}
J.~Guan, W.~W. Qian, X.~Peng, Y.~Su, J.~Peng, and J.~Ma.
\newblock 3d equivariant diffusion for target-aware molecule generation and affinity prediction.
\newblock In \emph{International Conference on Learning Representations}, 2023.

\bibitem[Huang et~al.(2024)Huang, Xu, Yu, Zhao, Chen, Han, Xie, Li, Zhong, Wong, and Zhang]{pmdm}
L.~Huang, T.~Xu, Y.~Yu, P.~Zhao, X.~Chen, J.~Han, Z.~Xie, H.~Li, W.~Zhong, K.-C. Wong, and H.~Zhang.
\newblock A dual diffusion model enables 3d molecule generation and lead optimization based on target pockets.
\newblock \emph{Nature Communications}, 15\penalty0 (1):\penalty0 2657, 2024.
\newblock \doi{10.1038/s41467-024-46569-1}.
\newblock URL \url{https://doi.org/10.1038/s41467-024-46569-1}.

\bibitem[Igashov et~al.(2024)Igashov, Stärk, Vignac, Schneuing, Satorras, Frossard, Welling, Bronstein, and Correia]{e3quivariant}
I.~Igashov, H.~Stärk, C.~Vignac, A.~Schneuing, V.~G. Satorras, P.~Frossard, M.~Welling, M.~Bronstein, and B.~Correia.
\newblock Equivariant 3d-conditional diffusion model for molecular linker design.
\newblock \emph{Nature Machine Intelligence}, 6\penalty0 (4):\penalty0 417--427, 2024.
\newblock \doi{10.1038/s42256-024-00815-9}.
\newblock URL \url{https://doi.org/10.1038/s42256-024-00815-9}.

\bibitem[Jia et~al.(2025)Jia, Huang, Wang, Garcia-Cardona, Bertozzi, and Wang]{pnpflow}
F.~Jia, Y.~Huang, S.-H. Wang, C.~Garcia-Cardona, A.~L. Bertozzi, and B.~Wang.
\newblock Plug-and-play image restoration with flow matching: A continuous viewpoint, 2025.
\newblock URL \url{https://arxiv.org/abs/2512.04283}.

\bibitem[Keiser et~al.(2009)Keiser, Setola, Irwin, Laggner, Abbas, Hufeisen, Jensen, Kuijer, Matos, Tran, Whaley, Glennon, Hert, Thomas, Edwards, Shoichet, and Roth]{keiser}
M.~J. Keiser, V.~Setola, J.~J. Irwin, C.~Laggner, A.~I. Abbas, S.~J. Hufeisen, N.~H. Jensen, M.~B. Kuijer, R.~C. Matos, T.~B. Tran, R.~Whaley, R.~A. Glennon, J.~Hert, K.~L.~H. Thomas, D.~D. Edwards, B.~K. Shoichet, and B.~L. Roth.
\newblock Predicting new molecular targets for known drugs.
\newblock \emph{Nature}, 462\penalty0 (7270):\penalty0 175--181, 2009.
\newblock \doi{10.1038/nature08506}.
\newblock URL \url{https://doi.org/10.1038/nature08506}.

\bibitem[Kim et~al.(2025)Kim, Kim, and Ye]{flowdps}
J.~Kim, B.~S. Kim, and J.~C. Ye.
\newblock Flowdps : Flow-driven posterior sampling for inverse problems.
\newblock In \emph{Proceedings of the IEEE/CVF International Conference on Computer Vision (ICCV)}, pages 12328--12337, 2025.

\bibitem[Kim et~al.(2026)Kim, Shin, and Lim]{fastdips}
M.~Kim, S.~Shin, and H.~Lim.
\newblock {FAST-DIPS}: Adjoint-free analytic steps and hard-constrained likelihood correction for diffusion-prior inverse problems.
\newblock In \emph{The Fourteenth International Conference on Learning Representations}, 2026.
\newblock URL \url{https://openreview.net/forum?id=voMeZVAkKL}.

\bibitem[Kuntz et~al.(1982)Kuntz, Blaney, Oatley, Langridge, and Ferrin]{docking1}
I.~D. Kuntz, J.~M. Blaney, S.~J. Oatley, R.~Langridge, and T.~E. Ferrin.
\newblock A geometric approach to macromolecule-ligand interactions.
\newblock \emph{Journal of Molecular Biology}, 161\penalty0 (2):\penalty0 269--288, 1982.
\newblock \doi{https://doi.org/10.1016/0022-2836(82)90153-X}.
\newblock URL \url{https://www.sciencedirect.com/science/article/pii/002228368290153X}.

\bibitem[Lee et~al.(2026)Lee, Hwang, Lee, Jo, Ko, Park, Choi, Yang, and Kim]{ace}
Z.~Lee, M.~Hwang, W.~Lee, S.~Jo, J.~Ko, Y.~B. Park, J.-M. Choi, E.~Yang, and K.~Kim.
\newblock On the collapse of generative paths: A criterion and correction for diffusion steering.
\newblock In \emph{Proceedings of the 43rd International Conference on Machine Learning}, 2026.
\newblock URL \url{https://openreview.net/forum?id=emv2qsi3TG}.

\bibitem[Li et~al.(2026)Li, Wu, Cao, Lin, Zhong, Chen, Lu, Tang, Lei, Ran, Yang, Xu, Shang, and Chen]{diffshape}
B.~Li, X.~Wu, Y.~Cao, J.~Lin, J.~Zhong, H.~Chen, Y.~Lu, M.~Tang, J.~Lei, T.~Ran, Y.~Yang, M.~Xu, J.~Shang, and H.~Chen.
\newblock De novo molecular design via shape-constrained diffusion models.
\newblock \emph{J. Chem. Inf. Model.}, 66\penalty0 (8):\penalty0 4592--4606, 2026.
\newblock \doi{10.1021/acs.jcim.6c00044}.
\newblock URL \url{https://doi.org/10.1021/acs.jcim.6c00044}.

\bibitem[Li et~al.(2024)Li, Kwon, Liang, Alkhouri, Ravishankar, and Qu]{dcdp}
X.~Li, S.~M. Kwon, S.~Liang, I.~R. Alkhouri, S.~Ravishankar, and Q.~Qu.
\newblock {DCDP}: Decoupled data consistency with diffusion purification for image restoration, 2024.
\newblock URL \url{https://arxiv.org/abs/2403.06054}.

\bibitem[Li et~al.(2025)Li, Zhou, Wang, Peng, and Zhang]{georcg}
Z.~Li, C.~Zhou, X.~Wang, X.~Peng, and M.~Zhang.
\newblock Geometric representation condition improves equivariant molecule generation.
\newblock In A.~Singh, M.~Fazel, D.~Hsu, S.~Lacoste-Julien, F.~Berkenkamp, T.~Maharaj, K.~Wagstaff, and J.~Zhu, editors, \emph{Proceedings of the 42nd International Conference on Machine Learning}, volume 267 of \emph{Proceedings of Machine Learning Research}, pages 36921--36953. PMLR, 2025.
\newblock URL \url{https://proceedings.mlr.press/v267/li25dz.html}.

\bibitem[Lin et~al.(2023)Lin, Huang, Zhang, Liu, Wu, Li, Chen, and Li]{d3fg}
H.~Lin, Y.~Huang, O.~Zhang, Y.~Liu, L.~Wu, S.~Li, Z.~Chen, and S.~Z. Li.
\newblock Functional-group-based diffusion for pocket-specific molecule generation and elaboration.
\newblock In A.~Oh, T.~Naumann, A.~Globerson, K.~Saenko, M.~Hardt, and S.~Levine, editors, \emph{Advances in Neural Information Processing Systems}, volume~36, pages 34603--34626. Curran Associates, Inc., 2023.
\newblock URL \url{https://proceedings.neurips.cc/paper_files/paper/2023/file/6cdd4ce9330025967dd1ed0bed3010f5-Paper-Conference.pdf}.

\bibitem[Mardani et~al.(2024)Mardani, Song, Kautz, and Vahdat]{reddiff}
M.~Mardani, J.~Song, J.~Kautz, and A.~Vahdat.
\newblock A variational perspective on solving inverse problems with diffusion models.
\newblock In B.~Kim, Y.~Yue, S.~Chaudhuri, K.~Fragkiadaki, M.~Khan, and Y.~Sun, editors, \emph{International Conference on Learning Representations}, volume 2024, pages 28027--28053, 2024.
\newblock URL \url{https://proceedings.iclr.cc/paper_files/paper/2024/file/7711026951f3dc3afb2be427bf98bc73-Paper-Conference.pdf}.

\bibitem[Meng et~al.(2022)Meng, He, Song, Song, Wu, Zhu, and Ermon]{sdedit}
C.~Meng, Y.~He, Y.~Song, J.~Song, J.~Wu, J.-Y. Zhu, and S.~Ermon.
\newblock {SDEdit}: Guided image synthesis and editing with stochastic differential equations.
\newblock In \emph{International Conference on Learning Representations}, 2022.
\newblock URL \url{https://openreview.net/forum?id=aBsCjcPu_tE}.

\bibitem[Park and Ye(2025)]{flowlps}
J.~Park and J.~C. Ye.
\newblock {FlowLPS}: Langevin-proximal sampling for flow-based inverse problem solvers, 2025.
\newblock URL \url{https://arxiv.org/abs/2512.07150}.

\bibitem[Patel et~al.(2025)Patel, Wen, Metaxas, and Yang]{flowchef}
M.~Patel, S.~Wen, D.~N. Metaxas, and Y.~Yang.
\newblock Flowchef: Steering of rectified flow models for controlled generations.
\newblock In \emph{Proceedings of the IEEE/CVF International Conference on Computer Vision (ICCV)}, pages 15308--15318, 2025.

\bibitem[Peng et~al.(2022)Peng, Luo, Guan, Xie, Peng, and Ma]{pocket2mol}
X.~Peng, S.~Luo, J.~Guan, Q.~Xie, J.~Peng, and J.~Ma.
\newblock {P}ocket2{M}ol: Efficient molecular sampling based on 3{D} protein pockets.
\newblock In K.~Chaudhuri, S.~Jegelka, L.~Song, C.~Szepesvari, G.~Niu, and S.~Sabato, editors, \emph{Proceedings of the 39th International Conference on Machine Learning}, volume 162 of \emph{Proceedings of Machine Learning Research}, pages 17644--17655. PMLR, 2022.
\newblock URL \url{https://proceedings.mlr.press/v162/peng22b.html}.

\bibitem[Pourya et~al.(2026)Pourya, El~Rawas, and Unser]{flower}
M.~Pourya, B.~El~Rawas, and M.~Unser.
\newblock {FLOWER}: A flow-matching solver for inverse problems.
\newblock In \emph{The Fourteenth International Conference on Learning Representations}, 2026.
\newblock URL \url{https://openreview.net/forum?id=QGd34p02mI}.

\bibitem[Qiao et~al.(2025)Qiao, Chen, Xie, Huang, Zhang, Deng, Rao, Deng, Meng, Wang, Xu, Chen, Xie, Zheng, Yang, Li, and Lei]{diffleop}
A.~Qiao, Y.~Chen, J.~Xie, W.~Huang, H.~Zhang, Q.~Deng, J.~Rao, J.~Deng, F.~Meng, Z.~Wang, M.~Xu, H.~Chen, J.~Xie, S.~Zheng, Y.~Yang, G.-B. Li, and J.~Lei.
\newblock A 3d pocket-aware lead optimization model with knowledge guidance and its application for discovery of new glutaminyl cyclase inhibitors.
\newblock \emph{Briefings in Bioinformatics}, 26\penalty0 (4):\penalty0 bbaf345, 2025.
\newblock \doi{10.1093/bib/bbaf345}.
\newblock URL \url{https://doi.org/10.1093/bib/bbaf345}.

\bibitem[Schneuing et~al.(2024)Schneuing, Harris, Du, Didi, Jamasb, Igashov, Du, Gomes, Blundell, Lio, Welling, Bronstein, and Correia]{diffsbdd}
A.~Schneuing, C.~Harris, Y.~Du, K.~Didi, A.~Jamasb, I.~Igashov, W.~Du, C.~Gomes, T.~L. Blundell, P.~Lio, M.~Welling, M.~Bronstein, and B.~Correia.
\newblock Structure-based drug design with equivariant diffusion models.
\newblock \emph{Nature Computational Science}, 4\penalty0 (12):\penalty0 899--909, 2024.
\newblock \doi{10.1038/s43588-024-00737-x}.

\bibitem[Shoichet and Kuntz(1993)]{docking2}
B.~K. Shoichet and I.~D. Kuntz.
\newblock Matching chemistry and shape in molecular docking.
\newblock \emph{Protein Engineering}, 6\penalty0 (7):\penalty0 723--732, 1993.
\newblock \doi{10.1093/protein/6.7.723}.
\newblock URL \url{https://doi.org/10.1093/protein/6.7.723}.

\bibitem[Song et~al.(2023)Song, Vahdat, Mardani, and Kautz]{gdm}
J.~Song, A.~Vahdat, M.~Mardani, and J.~Kautz.
\newblock Pseudoinverse-guided diffusion models for inverse problems.
\newblock In \emph{International Conference on Learning Representations}, 2023.
\newblock URL \url{https://openreview.net/forum?id=9_gsMA8MRKQ}.

\bibitem[Vargas et~al.(2023)Vargas, Grathwohl, and Doucet]{dds}
F.~Vargas, W.~S. Grathwohl, and A.~Doucet.
\newblock Denoising diffusion samplers.
\newblock In \emph{The Eleventh International Conference on Learning Representations}, 2023.
\newblock URL \url{https://openreview.net/forum?id=8pvnfTAbu1f}.

\bibitem[Wang et~al.(2023)Wang, Yu, and Zhang]{ddnm}
Y.~Wang, J.~Yu, and J.~Zhang.
\newblock Zero-shot image restoration using denoising diffusion null-space model.
\newblock In \emph{International Conference on Learning Representations}, 2023.
\newblock URL \url{https://openreview.net/forum?id=mRieQgMtNTQ}.

\bibitem[Zhang et~al.(2025)Zhang, Chu, Berner, Meng, Anandkumar, and Song]{daps}
B.~Zhang, W.~Chu, J.~Berner, C.~Meng, A.~Anandkumar, and Y.~Song.
\newblock Improving diffusion inverse problem solving with decoupled noise annealing.
\newblock In \emph{Proceedings of the IEEE/CVF Conference on Computer Vision and Pattern Recognition (CVPR)}, pages 20895--20905, 2025.

\bibitem[Zhou et~al.(2025)Zhou, Zhu, Tang, and Li]{activitydiff}
R.~Zhou, H.~Zhu, J.~Tang, and M.~Li.
\newblock {ActivityDiff}: A diffusion model with positive and negative activity guidance for de novo drug design, 2025.
\newblock URL \url{https://arxiv.org/abs/2508.06364}.

\bibitem[Zhou et~al.(2024)Zhou, Cheng, Yang, Bao, Wang, and Gu]{decompopt}
X.~Zhou, X.~Cheng, Y.~Yang, Y.~Bao, L.~Wang, and Q.~Gu.
\newblock {DecompOpt}: Controllable and decomposed diffusion models for structure-based molecular optimization.
\newblock In \emph{International Conference on Learning Representations}, 2024.
\newblock URL \url{https://openreview.net/forum?id=Y3BbxvAQS9}.

\bibitem[Zhu et~al.(2023)Zhu, Zhang, Liang, Cao, Wen, Timofte, and Van~Gool]{diffpir}
Y.~Zhu, K.~Zhang, J.~Liang, J.~Cao, B.~Wen, R.~Timofte, and L.~Van~Gool.
\newblock Denoising diffusion models for plug-and-play image restoration.
\newblock In \emph{Proceedings of the IEEE/CVF Conference on Computer Vision and Pattern Recognition (CVPR) Workshops}, pages 1219--1229, 2023.

\bibitem[Zilberstein et~al.(2025)Zilberstein, Mardani, and Segarra]{rsd}
N.~Zilberstein, M.~Mardani, and S.~Segarra.
\newblock Repulsive latent score distillation for solving inverse problems.
\newblock In Y.~Yue, A.~Garg, N.~Peng, F.~Sha, and R.~Yu, editors, \emph{International Conference on Learning Representations}, volume 2025, pages 60181--60224, 2025.
\newblock URL \url{https://proceedings.iclr.cc/paper_files/paper/2025/file/96d328a1f6d8396d8c8a62f2beee252a-Paper-Conference.pdf}.

\end{thebibliography}

\appendix
\clearpage
\begin{table*}[!t]

\centering
\caption{Comparison of related molecular design and sample-refinement methods.}
\label{tab:related_method_comparison}

\footnotesize

\colorlet{oursrow}{blue!8}

\newcommand{\markcell}[1]{%
  \raisebox{0pt}[1.1em][0.35em]{%
    \makebox[1.45em][c]{#1}%
  }%
}
\newcommand{\yes}{\markcell{$\checkmark$}}
\newcommand{\no}{\markcell{$\times$}}
\newcommand{\partialyes}{\markcell{\raisebox{-0.03em}{$\triangle$}}}

\newcommand{\groupcell}[1]{%
  \rotatebox[origin=c]{90}{%
    \fontsize{6.8}{7.3}\selectfont\bfseries
    \shortstack[c]{#1}%
  }%
}

\newcommand{\methodhead}[1]{%
  \parbox[c][3.35em][c]{\linewidth}{%
    \raggedright
    \fontsize{7.8}{8.6}\selectfont\bfseries
    #1%
  }%
}

\newcommand{\caphead}[1]{%
  \parbox[c][3.35em][c]{\linewidth}{%
    \centering
    \fontsize{6.8}{7.5}\selectfont\bfseries
    #1%
  }%
}

\newcolumntype{G}{>{\centering\arraybackslash}m{0.032\textwidth}}
\newcolumntype{F}{>{\raggedright\arraybackslash\bfseries}m{0.250\textwidth}}
\newcolumntype{R}{>{\raggedright\arraybackslash\fontsize{7.4}{8.3}\selectfont}m{0.210\textwidth}}
\newcolumntype{C}{>{\centering\arraybackslash}m{0.074\textwidth}}

\colorlet{subrulegray}{gray!25}
\setlength{\cmidrulewidth}{0.08pt}
\setlength{\tabcolsep}{2pt}
\renewcommand{\arraystretch}{1.12}

\begin{tabular*}{\textwidth}{
  @{\extracolsep{\fill}}
  G
  F
  R
  C C C C C C
  @{}
}
\toprule
&
\methodhead{Method family}
&
\methodhead{Representative methods}
&
\caphead{No objective-\\specific\\retraining}
&
\caphead{Inference-time\\sample control}
&
\caphead{Explicit\\geometric\\control}
&
\caphead{Pocket--ligand\\surface\\complementarity}
&
\caphead{Target-only\\specificity}
&
\caphead{OOD pseudo-label\\refinement}
\\

\midrule

\cellcolor{white}\multirow[c]{8}{*}{\groupcell{Lead\\Optimization}}
& \cellcolor{oursrow}\textbf{Off-target-agnostic specificity-aware optimization}
& \cellcolor{oursrow}\textbf{SurfSpec (Ours)}
& \cellcolor{oursrow}\yes
& \cellcolor{oursrow}\yes
& \cellcolor{oursrow}\yes
& \cellcolor{oursrow}\yes
& \cellcolor{oursrow}\yes
& \cellcolor{oursrow}\yes \\

\cmidrule(lr){2-9}
& \multirow[c]{3}{*}{Off-target-guided specificity}
& SBE-Diff, ActivityDiff
& \no & \yes & \no & \no & \no & \no \\

&
& MolSculptor
& \partialyes & \yes & \no & \no & \no & \no \\

\cmidrule(lr){2-9}
& \multirow[c]{4}{*}{Target-only optimization}
& DecompOpt
& \yes & \yes & \no & \no & \no & \no \\

&
& Delete
& \no & \no & \no & \no & \no & \no \\

&
& PMDM
& \partialyes & \partialyes & \no & \no & \no & \no \\

&
& Diffleop, MolJO
&  \no & \yes & \partialyes & \no & \no & \no \\


\midrule

\multirow[c]{4}{*}{\groupcell{Sample refinement}}
& Reverse-path refinement
& SDEdit, DDNM, MCG, DPS, $\Pi$GDM, FlowChef, FlowDPS
& \yes & \yes & \no & \no & \no & \no \\
\cmidrule(lr){2-9}
& Optimization-based refinement
& DiffPIR, DDS, DAPS, DCDP, PnP-Flow, FLOWER, FlowLPS
& \yes & \yes & \no & \no & \no & \no \\
\cmidrule(lr){2-9}
& Variational refinement
& RED-Diff, RSD, FLAIR
& \yes & \yes & \no & \no & \no & \no \\
\cmidrule(lr){2-9}
& Constrained sample refinement
& FAST-DIPS
& \yes & \yes & \no & \no & \no & \no \\

\midrule

\multirow[c]{4}{*}{\groupcell{Geometric control}}
& Shape-conditioned generation
& Diff-Shape
& \no & \no & \yes & \no & \no & \no \\
\cmidrule(lr){2-9}
& Multimodal geometric generation
& ShEPhERD
& \no & \no & \yes & \no & \no & \no \\
\cmidrule(lr){2-9}
& Representation-conditioned generation
& GeoRCG
& \no & \no & \yes & \no & \no & \no \\
\cmidrule(lr){2-9}
& Training-free geometric guidance
& UniGuide
& \yes & \yes & \yes & \no & \no & \no \\

\midrule

\multirow[c]{5}{*}{\groupcell{Molecular\\design}}
& Pocket-conditioned generation
& Pocket2Mol, TargetDiff, AliDiff
& \no & \no & \no & \no & \no & \no \\
\cmidrule(lr){2-9}
& \multirow[c]{2}{*}{Multi-task SBDD generation}

& DiffSBDD
& \yes & \yes & \partialyes & \no & \no & \no \\

&
& DiffGui, PocketXMol, 
& \no & \yes & \partialyes & \no & \no & \no \\

\cmidrule(lr){2-9}
& Molecular elaboration
& D3FG
& \no & \no & \partialyes & \no & \no & \no \\

\bottomrule
\end{tabular*}
\end{table*}

\FloatBarrier

\section{Literature Context of the Work}

Table~\ref{tab:related_method_comparison} summarizes related methods across lead optimization, sample refinement, geometric control, and structure-based molecular design. SurfSpec is distinguished by jointly supporting inference-time geometric control, pocket--ligand surface complementarity, target-only specificity, and OOD pseudo-label refinement without objective-specific retraining.

\section{Proofs and Additional Theoretical Details}
\label{app:proof}

\subsection{Surface-Oriented Probability Fields}
\label{app:surface_probability_field}

We construct the surface-induced probability measures used in
Eq.~\eqref{eq:js_geometric_mismatch} from signed-distance functions.
Let \(A_L=\{a_i\}_{i=1}^{N_L}\) denote the heavy-atom coordinates of the
ligand \(L\).
As in the main text, we define the continuous ligand-local domain
\begin{equation}
\Omega_L
=
\left\{
x\in\mathbb{R}^3:
\min_i \|x-a_i\|_2 \le 8\,\text{\AA}
\right\}.
\label{eq:ligand_local_domain_appendix}
\end{equation}
Let \(\surf_L\) denote the van der Waals surface of the ligand heavy
atoms, and let \(\surf_P^L\) denote the ligand-oriented pocket surface,
obtained by aligning the pocket to the ligand frame and restricting the
pocket van der Waals surface to points within \(8\,\text{\AA}\) of any
ligand heavy atom.

We define signed-distance functions on \(\Omega_L\).
For the ligand, let \(B_L\) be the ligand van der Waals solid, whose
boundary is \(\surf_L\).
The signed distance to the ligand surface is
\begin{equation}
\phi_L(x)
=
\begin{cases}
\phantom{-}\inf_{z\in\surf_L}\|x-z\|_2,
& x\notin B_L,\\
-\inf_{z\in\surf_L}\|x-z\|_2,
& x\in B_L.
\end{cases}
\label{eq:ligand_signed_distance}
\end{equation}
For the pocket, let \(B_P^L\) be the aligned protein van der Waals solid.
We define
\begin{equation}
\phi_P(x)
=
\begin{cases}
\phantom{-}\inf_{z\in\surf_P^L}\|x-z\|_2,
& x\notin B_P^L,\\
-\inf_{z\in\surf_P^L}\|x-z\|_2,
& x\in B_P^L.
\end{cases}
\label{eq:pocket_signed_distance}
\end{equation}
Thus, \(\phi_P(x)<0\) indicates that \(x\) lies inside the protein
van der Waals volume, corresponding to a clash region.

We define the ligand and pocket surface energies as
\begin{equation}
\begin{aligned}
E_L(x)
&=
\frac{\phi_L(x)^2}{2\sigma_L^2},
\\
E_P(x)
&=
\frac{\phi_P(x)^2}{2\sigma_P^2}
+
\lambda_{\mathrm{clash}}
[-\phi_P(x)]_+,
\end{aligned}
\label{eq:surface_field_energy_appendix}
\end{equation}
where \([u]_+=\max(u,0)\).
The clash term discourages probability mass inside the protein side of
the pocket surface.

We convert these energies into normalized Boltzmann densities on the
common continuous domain \(\Omega_L\):
\begin{align}
q_L(x)
=
\frac{
\exp(-\beta E_L(x))
}{
\int_{\Omega_L}\exp(-\beta E_L(x'))\,dx'
},\\
\qquad
q_P(x)
=
\frac{
\exp(-\beta E_P(x))
}{
\int_{\Omega_L}\exp(-\beta E_P(x'))\,dx'
}.
\label{eq:surface_probability_densities_appendix}
\end{align}
We set \(\beta=2\) in all experiments.
The corresponding probability measures are
\begin{equation}
d\mu_L(x)=q_L(x)\,dx,
\qquad
d\nu_P^L(x)=q_P(x)\,dx.
\label{eq:surface_probability_measures_appendix}
\end{equation}
The geometric mismatch in Eq.~\eqref{eq:js_geometric_mismatch} is then
computed as the Jensen--Shannon distance between
\(\mu_L\) and \(\nu_P^L\), which are defined on the same ligand-local
domain.

In implementation, the integrals in
Eq.~\eqref{eq:surface_probability_densities_appendix} are approximated
by evaluating the unnormalized Boltzmann densities on a ligand-centered
voxel grid over \(\Omega_L\) and normalizing the resulting discrete
weights.

\subsection{Proof of the Specificity Certificate}
\label{app:proof_specificity}

We first restate the quantile-calibrated geometry--affinity assumption
used in the main text. This assumption provides an empirical bridge from
geometric mismatch to affinity, while allowing residual variation due to
non-geometric factors such as electrostatics, hydrogen bonding, and
desolvation.

\addtocounter{assumption}{-1}
\begin{assumption}[Quantile-calibrated geometry--affinity envelope]
\label{assump:calibrated_appendix}
Let \(\mathcal{D}_{\mathrm{CD}}\) denote the empirical ligand--pocket
distribution induced by CrossDocked2020. For a tail probability
\(\eta\in(0,1)\), there exist monotone decreasing functions
\(f_-^{\eta}\) and \(f_+^{\eta}\) such that, for
\((L,P)\sim\mathcal{D}_{\mathrm{CD}}\),
\begin{equation}
\Pr\!\left[
f_-^{\eta}(\dgm(L,P))
\le
A_P(L)
\le
f_+^{\eta}(\dgm(L,P))
\right]
\ge
1-\eta .
\label{eq:percentile_calibration_appendix}
\end{equation}
Here, \(A_P(L)\) is an affinity-like score for ligand \(L\) against pocket
\(P\), where larger values indicate stronger binding.
\end{assumption}

This assumption does not require geometry to fully determine affinity.
Rather, it states that geometric mismatch constrains the attainable affinity
range with high probability. The percentile envelope absorbs affinity
variation caused by non-geometric interactions.

We next formalize the geometric part of the certificate. Let
\(\mathcal{L}_{\mathrm{opt}}\) denote the set of ligands reachable by
optimization. For \(L\in\mathcal{L}_{\mathrm{opt}}\) and \(O\in\Off\), let
\[
m_L(O)=d_\mathrm{JS}(\nu_{\Tgt}^{L},\nu_O^{L})
\]
denote the ligand-oriented geometric separation between the target and
off-target pocket surfaces, where \(\nu_{\Tgt}^{L}\) and \(\nu_O^{L}\) are
the corresponding ligand-oriented pocket surface measures. We say that
\(\Off\) is a \(\delta\)-separated off-target class if
\[
\inf_{L\in\mathcal{L}_{\mathrm{opt}}}
\inf_{O\in\Off}
m_L(O)
\ge
\delta .
\]
This means that every considered off-target pocket remains at least
\(\delta\) away from the target pocket under the ligand-oriented surface
metric throughout optimization.

\begin{lemma}[Triangle lower bound for off-target mismatch]
\label{lem:triangle_offtarget_mismatch}
Let \(L\in\mathcal{L}_{\mathrm{opt}}\), and let
\[
\epsilon=\dgm(L,\Tgt)
\]
be the target mismatch. Then, for each off-target pocket \(O\in\Off\),
\begin{equation}
\dgm(L,O)
\ge
[m_L(O)-\epsilon]_+ .
\label{eq:offtarget_mismatch_lower_bound}
\end{equation}
In particular, if \(\Off\) is a \(\delta\)-separated off-target class, then
\begin{equation}
\dgm(L,O)
\ge
[\delta-\epsilon]_+
\qquad
\text{for all } O\in\Off .
\label{eq:offtarget_mismatch_delta_lower_bound}
\end{equation}
\end{lemma}

\begin{proof}
By definition of the ligand-oriented mismatch,
\[
\dgm(L,\Tgt)=d_\mathrm{JS}(\mu_L,\nu_{\Tgt}^{L})=\epsilon,
\]
and
\[
\dgm(L,O)=d_\mathrm{JS}(\mu_L,\nu_O^{L}).
\]
Since \(d_\mathrm{JS}\) is a metric on probability measures, it satisfies the triangle
inequality:
\[
d_\mathrm{JS}(\nu_{\Tgt}^{L},\nu_O^{L})
\le
d_\mathrm{JS}(\nu_{\Tgt}^{L},\mu_L)
+
d_\mathrm{JS}(\mu_L,\nu_O^{L}).
\]
Therefore,
\[
m_L(O)
\le
\epsilon+\dgm(L,O).
\]
Rearranging gives
\[
\dgm(L,O)\ge m_L(O)-\epsilon.
\]
Since \(\dgm(L,O)\ge 0\), we obtain
\[
\dgm(L,O)
\ge
[m_L(O)-\epsilon]_+ .
\]

If \(\Off\) is \(\delta\)-separated, then \(m_L(O)\ge\delta\) for every
\(L\in\mathcal{L}_{\mathrm{opt}}\) and \(O\in\Off\). Hence
\[
[m_L(O)-\epsilon]_+
\ge
[\delta-\epsilon]_+,
\]
which yields
\[
\dgm(L,O)
\ge
[\delta-\epsilon]_+ .
\]
\end{proof}

\addtocounter{theorem}{-1}
\begin{theorem}[Specificity certificate for separated off-targets]
\label{thm:spec_bound_appendix}
Let \(L\in\mathcal{L}_{\mathrm{opt}}\), and let
\(\epsilon=\dgm(L,\Tgt)\). Suppose that \(\Off\) is a
\(\delta\)-separated off-target class. Under
Assumption~\ref{assump:calibrated_appendix}, with probability at least
\(1-(|\Off|+1)\eta\),
\begin{equation}
\Spec(L;\Tgt,\Off)
\ge
f_-^\eta(\epsilon)
-
f_+^\eta\!\left([\delta-\epsilon]_+\right).
\label{eq:spec_bound_ligand_invariant_appendix}
\end{equation}
Consequently, reducing the target mismatch \(\epsilon\) weakly improves
this high-probability specificity lower bound.
\end{theorem}

\begin{proof}
The proof has two steps. First, we obtain a geometric lower bound on the
mismatch between the ligand and each off-target pocket. Second, we convert
this geometric lower bound into an affinity-level specificity certificate
using the percentile-calibrated geometry--affinity envelope.

By Lemma~\ref{lem:triangle_offtarget_mismatch}, for every off-target pocket
\(O\in\Off\),
\begin{equation}
\dgm(L,O)
\ge
[\delta-\epsilon]_+ .
\label{eq:proof_offtarget_mismatch_delta}
\end{equation}

Now consider the calibrated affinity events. For the target pocket, the
lower envelope in Assumption~\ref{assump:calibrated_appendix} gives, with
probability at least \(1-\eta\),
\[
A_{\Tgt}(L)
\ge
f_-^\eta(\dgm(L,\Tgt))
=
f_-^\eta(\epsilon).
\]
For each off-target pocket \(O\in\Off\), the upper envelope gives, with
probability at least \(1-\eta\),
\[
A_O(L)
\le
f_+^\eta(\dgm(L,O)).
\]
Because \(f_+^\eta\) is monotone decreasing and
\(\dgm(L,O)\ge[\delta-\epsilon]_+\), we have
\[
f_+^\eta(\dgm(L,O))
\le
f_+^\eta\!\left([\delta-\epsilon]_+\right).
\]
Therefore, for every \(O\in\Off\),
\[
A_O(L)
\le
f_+^\eta\!\left([\delta-\epsilon]_+\right).
\]
Taking the maximum over off-target pockets yields
\[
\max_{O\in\Off} A_O(L)
\le
f_+^\eta\!\left([\delta-\epsilon]_+\right).
\]

Combining the target lower bound and the off-target upper bound gives
\[
\begin{aligned}
\Spec(L;\Tgt,\Off)
&=
A_{\Tgt}(L)
-
\max_{O\in\Off} A_O(L) \\
&\ge
f_-^\eta(\epsilon)
-
f_+^\eta\!\left([\delta-\epsilon]_+\right).
\end{aligned}
\]

It remains to account for the probability of the calibrated events. The
target lower-envelope event fails with probability at most \(\eta\). For
each of the \(|\Off|\) off-target pockets, the upper-envelope event fails
with probability at most \(\eta\). By the union bound, all calibrated events
hold simultaneously with probability at least
\[
1-(|\Off|+1)\eta .
\]
Thus Eq.~\eqref{eq:spec_bound_ligand_invariant_appendix} holds with the
stated probability.

Finally, define the certified lower bound as
\[
B(\epsilon)
=
f_-^\eta(\epsilon)
-
f_+^\eta\!\left([\delta-\epsilon]_+\right).
\]
Since \(f_-^\eta\) is monotone decreasing, reducing \(\epsilon\) weakly
increases \(f_-^\eta(\epsilon)\). Also, reducing \(\epsilon\) weakly
increases \([\delta-\epsilon]_+\), and since \(f_+^\eta\) is monotone
decreasing, this weakly decreases
\(f_+^\eta([\delta-\epsilon]_+)\). Therefore, reducing the target mismatch
\(\epsilon\) weakly improves the high-probability lower bound
\(B(\epsilon)\).
\end{proof}

\subsection{Low-Noise Analysis of Anchored Refinement}
\label{app:proof_refinement}

We analyze the low-noise anchored recovery step for a pseudo-label
\(a=\xanc\).
Let
\[
V_a(x)
=
\frac{\lambda}{2}\|x-a\|_2^2
\]
denote the anchor potential.
The anchored clean distribution is defined as
\begin{equation}
p_0^{\mathrm{anc}}(x\mid\Tgt,a)
=
\frac{1}{Z_a}
p_0(x\mid\Tgt)\exp(-V_a(x)),
\label{eq:anchored_clean_distribution_appendix}
\end{equation}
where \(Z_a\) is the normalizing constant.
Let \(\mu_0^{\mathrm{anc}}\) and
\(\Sigma_0^{\mathrm{anc}}\) denote the mean and covariance of
\(p_0^{\mathrm{anc}}(\cdot\mid\Tgt,a)\), respectively.

For diffusion time \(\tau\), consider the Gaussian forward perturbation
\begin{equation}
Y_\tau
=
\alpha_\tau X_0+\sigma_\tau Z,
\qquad
Z\sim\mathcal{N}(0,I_d),
\label{eq:anchored_forward_process_appendix}
\end{equation}
where \(X_0\sim p_0^{\mathrm{anc}}(\cdot\mid\Tgt,a)\) and \(Z\) is
independent of \(X_0\).
We denote the distribution of \(Y_\tau\) by \(p_\tau^{\mathrm{anc}}\).

\begin{assumption}[Anchored concentration]
\label{assump:anchored_concentration}
The anchored clean distribution is concentrated in a small ball around
its mean:
\begin{equation}
\operatorname{supp}
\left(
p_0^{\mathrm{anc}}(\cdot\mid\Tgt,a)
\right)
\subset
B(\mu_0^{\mathrm{anc}},r_0).
\label{eq:anchored_concentration_assumption}
\end{equation}
\end{assumption}

This assumption formalizes the intuition that imposing both
pocket-conditioned validity and similarity to the pseudo-label leaves a
locally concentrated set of feasible clean ligands.
It also implies the covariance bound
\begin{equation}
\operatorname{tr}(\Sigma_0^{\mathrm{anc}})
=
\mathbb{E}
\left[
\|X_0-\mu_0^{\mathrm{anc}}\|_2^2
\right]
\le
r_0^2 .
\label{eq:anchored_covariance_from_support}
\end{equation}
This covariance bound is a consequence of
Assumption~\ref{assump:anchored_concentration}, not an additional
assumption.

\begin{proposition}[Mode recovery under anchored concentration]
\label{prop:mode_recovery}
Suppose Assumption~\ref{assump:anchored_concentration} holds.
Let \(y_\tau^\star\) be an interior mode of
\(p_\tau^{\mathrm{anc}}\).
Then
\begin{equation}
\left\|
y_\tau^\star
-
\alpha_\tau\mu_0^{\mathrm{anc}}
\right\|_2
\le
|\alpha_\tau|r_0 .
\label{eq:mode_to_forward_mean_gap}
\end{equation}
Consequently, if the low-noise anchored score flow recovers the mode
\(y_\tau^\star\), then its convergence point lies within
\(|\alpha_\tau|r_0\) of the forward mean
\(\alpha_\tau\mu_0^{\mathrm{anc}}\).
\end{proposition}

\begin{proof}
For the Gaussian forward perturbation in
Eq.~\eqref{eq:anchored_forward_process_appendix}, the score identity is
\begin{equation}
\nabla_y\log p_\tau^{\mathrm{anc}}(y)
=
-\frac{1}{\sigma_\tau^2}
\left(
y
-
\alpha_\tau
\mathbb{E}
[
X_0\mid Y_\tau=y,\Tgt,a
]
\right).
\label{eq:anchored_score_identity}
\end{equation}
Since \(y_\tau^\star\) is an interior mode,
\(\nabla_y\log p_\tau^{\mathrm{anc}}(y_\tau^\star)=0\).
Therefore,
\[
y_\tau^\star
=
\alpha_\tau
\mathbb{E}
[
X_0\mid Y_\tau=y_\tau^\star,\Tgt,a
].
\]
Conditioning on \(Y_\tau=y_\tau^\star\) only reweights the clean
distribution and does not enlarge its support.
Hence the posterior distribution of \(X_0\mid Y_\tau=y_\tau^\star\)
remains supported in
\(B(\mu_0^{\mathrm{anc}},r_0)\).
Since this ball is convex, the posterior mean also lies in the same ball:
\[
\left\|
\mathbb{E}
[
X_0\mid Y_\tau=y_\tau^\star,\Tgt,a
]
-
\mu_0^{\mathrm{anc}}
\right\|_2
\le
r_0 .
\]
Multiplying by \(|\alpha_\tau|\) proves
Eq.~\eqref{eq:mode_to_forward_mean_gap}.
\end{proof}

\begin{assumption}[Local DPS score approximation]
\label{assump:local_dps_score_approx}
Let
\[
\widetilde{s}^{\mathrm{anc}}_\tau(y)
=
\nabla_y\log p_\tau(y\mid\Tgt)
-
\lambda
J_{\hat{x}_{0\mid\tau}}(y)^\top
\left(
\hat{x}_{0\mid\tau}(y)-\xanc
\right)
\]
denote the DPS-style anchored score estimator used in
Eq.~\eqref{eq:low_noise_mode_recovery}.
On the region \(B\) visited by the low-noise anchored score flow, we
assume that this estimator is locally close to the exact anchored score:
\begin{equation}
\sup_{y\in B}
\left\|
\widetilde{s}^{\mathrm{anc}}_\tau(y)
-
\nabla_y\log p_\tau^{\mathrm{anc}}(y)
\right\|_2
\le
\xi_\tau .
\label{eq:local_dps_score_approx}
\end{equation}
\end{assumption}

\begin{proposition}[Mode recovery under anchored concentration]
\label{prop:mode_recovery}
Suppose Assumption~\ref{assump:anchored_concentration} holds and define
\[
\kappa_\tau
=
\frac{\alpha_\tau^2 r_0^2}{\sigma_\tau^2}.
\]
If \(\kappa_\tau<1\), then
\(\log p_\tau^{\mathrm{anc}}\) is strongly concave with constant
\[
m_\tau
=
\frac{1-\kappa_\tau}{\sigma_\tau^2}.
\]
Consequently, \(p_\tau^{\mathrm{anc}}\) has a unique mode
\(y_\tau^\star\), and
\begin{equation}
\left\|
y_\tau^\star
-
\alpha_\tau\mu_0^{\mathrm{anc}}
\right\|_2
\le
|\alpha_\tau|r_0.
\label{eq:mode_to_forward_mean_gap}
\end{equation}

Assume further that the approximate flow
\[
\frac{d y_s}{ds}
=
\widetilde{s}^{\mathrm{anc}}_\tau(y_s)
\]
remains in a region \(B\), and that the implemented score satisfies
\[
\sup_{y\in B}
\left\|
\widetilde{s}^{\mathrm{anc}}_\tau(y)
-
\nabla_y\log p_\tau^{\mathrm{anc}}(y)
\right\|_2
\le
\xi_\tau .
\]
Then
\begin{equation}
\|y_s-y_\tau^\star\|_2
\le
e^{-m_\tau s}\|y_0-y_\tau^\star\|_2
+
\frac{\xi_\tau}{m_\tau}
\left(
1-e^{-m_\tau s}
\right).
\label{eq:approx_mode_recovery_bound}
\end{equation}
\end{proposition}

\begin{proof}
For a Gaussian forward perturbation, the Hessian identity gives
\[
\nabla_y^2\log p_\tau^{\mathrm{anc}}(y)
=
-\frac{1}{\sigma_\tau^2}I_d
+
\frac{\alpha_\tau^2}{\sigma_\tau^4}
\operatorname{Cov}
\left(
X_0\mid Y_\tau=y,\Tgt,a
\right).
\]
Under Assumption~\ref{assump:anchored_concentration}, the conditional
support of \(X_0\mid Y_\tau=y\) remains inside
\(B(\mu_0^{\mathrm{anc}},r_0)\).
Therefore,
\[
\left\|
\operatorname{Cov}
\left(
X_0\mid Y_\tau=y,\Tgt,a
\right)
\right\|_{\mathrm{op}}
\le
r_0^2 .
\]
Hence the largest eigenvalue of
\(\nabla_y^2\log p_\tau^{\mathrm{anc}}(y)\) is at most
\[
-\frac{1}{\sigma_\tau^2}
+
\frac{\alpha_\tau^2 r_0^2}{\sigma_\tau^4}
=
-\frac{1-\kappa_\tau}{\sigma_\tau^2}
=
-m_\tau,
\]
so \(\log p_\tau^{\mathrm{anc}}\) is \(m_\tau\)-strongly concave.

The score identity for Gaussian perturbations is
\[
\nabla_y\log p_\tau^{\mathrm{anc}}(y)
=
-\frac{1}{\sigma_\tau^2}
\left(
y-\alpha_\tau
\mathbb{E}
[
X_0\mid Y_\tau=y,\Tgt,a
]
\right).
\]
At the mode \(y_\tau^\star\), the score vanishes, so
\[
y_\tau^\star
=
\alpha_\tau
\mathbb{E}
[
X_0\mid Y_\tau=y_\tau^\star,\Tgt,a
].
\]
The posterior mean lies in the same convex ball
\(B(\mu_0^{\mathrm{anc}},r_0)\), which gives
Eq.~\eqref{eq:mode_to_forward_mean_gap}.

Let
\[
s_\tau^{\mathrm{anc}}(y)
=
\nabla_y\log p_\tau^{\mathrm{anc}}(y),
\qquad
\widetilde{s}^{\mathrm{anc}}_\tau(y)
=
s_\tau^{\mathrm{anc}}(y)+\Delta_\tau(y),
\]
with \(\|\Delta_\tau(y)\|_2\le\xi_\tau\) on \(B\).
Since \(s_\tau^{\mathrm{anc}}(y_\tau^\star)=0\), strong concavity implies
\[
\left\langle
y_s-y_\tau^\star,
s_\tau^{\mathrm{anc}}(y_s)
-
s_\tau^{\mathrm{anc}}(y_\tau^\star)
\right\rangle
\le
-m_\tau
\|y_s-y_\tau^\star\|_2^2.
\]
Therefore,
\[
\frac{d}{ds}
\|y_s-y_\tau^\star\|_2
\le
-m_\tau\|y_s-y_\tau^\star\|_2+\xi_\tau.
\]
Gronwall's inequality gives
Eq.~\eqref{eq:approx_mode_recovery_bound}.
\end{proof}

\begin{proposition}[Accuracy of the low-noise Gaussian surrogate]
\label{prop:gaussian_surrogate_accuracy}
Define the Gaussian surrogate
\begin{equation}
\widetilde p_\tau^{\mathrm{anc}}
=
\mathcal{N}
\left(
\alpha_\tau\mu_0^{\mathrm{anc}},
\sigma_\tau^2 I_d
\right),
\label{eq:gaussian_surrogate_distribution}
\end{equation}
and the residual-dispersion ratio
\begin{equation}
\rho_\tau
=
\frac{
\alpha_\tau^2
\operatorname{tr}(\Sigma_0^{\mathrm{anc}})
}{
\sigma_\tau^2
}.
\label{eq:residual_dispersion_ratio}
\end{equation}
Then
\begin{equation}
D_{\mathrm{KL}}
\left(
p_\tau^{\mathrm{anc}}
\middle\|
\widetilde p_\tau^{\mathrm{anc}}
\right)
\le
\frac{\rho_\tau}{2},
\label{eq:gaussian_surrogate_kl}
\end{equation}
\begin{equation}
W_2^2
\left(
p_\tau^{\mathrm{anc}},
\widetilde p_\tau^{\mathrm{anc}}
\right)
\le
\alpha_\tau^2
\operatorname{tr}(\Sigma_0^{\mathrm{anc}}),
\label{eq:gaussian_surrogate_w2}
\end{equation}
and
\begin{equation}
\operatorname{TV}
\left(
p_\tau^{\mathrm{anc}},
\widetilde p_\tau^{\mathrm{anc}}
\right)
\le
\frac{1}{2}\sqrt{\rho_\tau}.
\label{eq:gaussian_surrogate_tv}
\end{equation}
Under Assumption~\ref{assump:anchored_concentration},
\begin{equation}
\rho_\tau
\le
\frac{\alpha_\tau^2 r_0^2}{\sigma_\tau^2}.
\label{eq:rho_tau_r0_bound}
\end{equation}
Thus, when \(|\alpha_\tau|r_0\ll\sigma_\tau\),
\(p_\tau^{\mathrm{anc}}\) is well approximated by the Gaussian surrogate
centered at \(\alpha_\tau\mu_0^{\mathrm{anc}}\).
\end{proposition}

\begin{proof}
Conditioned on \(X_0=x\), the forward distribution is
\[
Y_\tau\mid X_0=x
\sim
\mathcal{N}(\alpha_\tau x,\sigma_\tau^2I_d).
\]
By convexity of KL divergence in its first argument,
\[
\begin{aligned}
&D_{\mathrm{KL}}
\left(
p_\tau^{\mathrm{anc}}
\middle\|
\widetilde p_\tau^{\mathrm{anc}}
\right)
\\
&\le
\mathbb{E}_{X_0}
\left[
D_{\mathrm{KL}}
\left(
\mathcal{N}(\alpha_\tau X_0,\sigma_\tau^2I_d)
\middle\|
\mathcal{N}(\alpha_\tau\mu_0^{\mathrm{anc}},\sigma_\tau^2I_d)
\right)
\right]
\\
&=
\frac{\alpha_\tau^2}{2\sigma_\tau^2}
\mathbb{E}
\left[
\|X_0-\mu_0^{\mathrm{anc}}\|_2^2
\right]
=
\frac{\rho_\tau}{2}.
\end{aligned}
\]
For the Wasserstein bound, couple
\[
Y_\tau
=
\alpha_\tau X_0+\sigma_\tau Z,
\qquad
\widetilde Y_\tau
=
\alpha_\tau\mu_0^{\mathrm{anc}}+\sigma_\tau Z
\]
using the same Gaussian noise \(Z\).
Then
\[
\mathbb{E}
\left[
\|Y_\tau-\widetilde Y_\tau\|_2^2
\right]
=
\alpha_\tau^2
\operatorname{tr}(\Sigma_0^{\mathrm{anc}}),
\]
which proves Eq.~\eqref{eq:gaussian_surrogate_w2}.
The total-variation bound follows from Pinsker's inequality.
Finally, Eq.~\eqref{eq:rho_tau_r0_bound} follows from
Eq.~\eqref{eq:anchored_covariance_from_support}.
\end{proof}

\begin{proposition}[Re-noising around the recovered point]
\label{prop:renoising_recovered_mode}
Suppose Assumption~\ref{assump:anchored_concentration} holds, and suppose
the conditions of Proposition~\ref{prop:mode_recovery} hold.
Let \(y_s\) be the point obtained by the approximate anchored score flow
at flow time \(s\), and define
\[
\widehat X_\tau
=
y_s+\sigma_\tau\varepsilon,
\qquad
\varepsilon\sim\mathcal{N}(0,I_d).
\]
Equivalently,
\[
\widehat X_\tau
\sim
\mathcal{N}(y_s,\sigma_\tau^2I_d).
\]
Then
\begin{equation}
\begin{aligned}
&W_2
\left(
p_\tau^{\mathrm{anc}},
\mathcal{N}(y_s,\sigma_\tau^2I_d)
\right)
\\
&\le
2|\alpha_\tau|r_0
+
e^{-m_\tau s}\|y_0-y_\tau^\star\|_2
+
\frac{\xi_\tau}{m_\tau}
\left(
1-e^{-m_\tau s}
\right).
\end{aligned}
\label{eq:renoising_recovered_point_w2_bound}
\end{equation}
In particular, if the recovered point is the exact mode
\(y_s=y_\tau^\star\), then
\begin{equation}
W_2
\left(
p_\tau^{\mathrm{anc}},
\mathcal{N}(y_\tau^\star,\sigma_\tau^2I_d)
\right)
\le
2|\alpha_\tau|r_0 .
\label{eq:renoising_mode_w2_bound}
\end{equation}
\end{proposition}

\begin{proof}
By the triangle inequality,
\[
\begin{aligned}
&W_2
\left(
p_\tau^{\mathrm{anc}},
\mathcal{N}(y_s,\sigma_\tau^2I_d)
\right)
\\
&\le
W_2
\left(
p_\tau^{\mathrm{anc}},
\widetilde p_\tau^{\mathrm{anc}}
\right)
+
W_2
\left(
\widetilde p_\tau^{\mathrm{anc}},
\mathcal{N}(y_s,\sigma_\tau^2I_d)
\right).
\end{aligned}
\]
The first term is bounded by
\[
W_2
\left(
p_\tau^{\mathrm{anc}},
\widetilde p_\tau^{\mathrm{anc}}
\right)
\le
|\alpha_\tau|
\sqrt{\operatorname{tr}(\Sigma_0^{\mathrm{anc}})}
\le
|\alpha_\tau|r_0,
\]
where the first inequality follows from
Proposition~\ref{prop:gaussian_surrogate_accuracy}, and the second
follows from Assumption~\ref{assump:anchored_concentration}.

For the second term, recall that
\[
\widetilde p_\tau^{\mathrm{anc}}
=
\mathcal{N}
\left(
\alpha_\tau\mu_0^{\mathrm{anc}},
\sigma_\tau^2I_d
\right).
\]
Since the two Gaussian distributions have the same covariance,
\[
W_2
\left(
\widetilde p_\tau^{\mathrm{anc}},
\mathcal{N}(y_s,\sigma_\tau^2I_d)
\right)
=
\|y_s-\alpha_\tau\mu_0^{\mathrm{anc}}\|_2.
\]
By the triangle inequality,
\[
\|y_s-\alpha_\tau\mu_0^{\mathrm{anc}}\|_2
\le
\|y_s-y_\tau^\star\|_2
+
\|y_\tau^\star-\alpha_\tau\mu_0^{\mathrm{anc}}\|_2.
\]
Proposition~\ref{prop:mode_recovery} gives
\[
\|y_s-y_\tau^\star\|_2
\le
e^{-m_\tau s}\|y_0-y_\tau^\star\|_2
+
\frac{\xi_\tau}{m_\tau}
\left(
1-e^{-m_\tau s}
\right),
\]
and also
\[
\|y_\tau^\star-\alpha_\tau\mu_0^{\mathrm{anc}}\|_2
\le
|\alpha_\tau|r_0.
\]
Combining these bounds proves
Eq.~\eqref{eq:renoising_recovered_point_w2_bound}.
If \(y_s=y_\tau^\star\), the flow-error term vanishes, yielding
Eq.~\eqref{eq:renoising_mode_w2_bound}.
\end{proof}

\section{Experimental Details}
\label{app:exp_details}
\subsection{Implementation Details for Ligand Growth and Refinement}
\label{app:refinement_details}

\paragraph{Benchmark protocol.}
All ligand-growth experiments are conducted on the 100 test complexes
from CrossDocked2020~\cite{crossdock2020}.
For each complex, the reference ligand paired with the target pocket is
used as the initial lead, and each method generates one optimized
molecule.
For each target pocket, we construct the evaluated off-target set
\(\Off_{\mathrm{eval}}\) by randomly sampling 10 pockets from the
remaining 99 CrossDocked2020 test pockets.
For off-target-agnostic methods, these off-target pockets are never
accessed during optimization and are used only for evaluation.
For ActivityDiff, which requires negative targets, we randomly assign 10
pockets as surrogate off-targets, reflecting the assumed inaccessibility
of the true off-target set.

\paragraph{Docking timeout.} Each AutoDock Vina docking run is limited to 300 seconds.
Docking attempts that exceed this budget are treated as docking failures
and excluded from the corresponding Vina-score aggregation.
For empirical specificity metrics, a valid target-pocket docking score is
required to compute the target-over-off-target gap; therefore, samples
whose target-pocket docking fails are excluded from the specificity
aggregation.

\paragraph{Shared outer ligand-growth loop.}
Unless otherwise stated, all methods use at most \(K=3\) outer ligand-growth iterations.
At each iteration, we select the closest feasible unoccupied surface patch within \(4\text{--}10\,\text{\AA}\), generate a linker of size 10 using DiffLinker, and refine the resulting pseudo-label using the specified refinement method.
The procedure terminates early if no feasible surface patch exists or if the target-pocket Vina score after local optimization worsens relative to the previous iteration.

\paragraph{Baseline variants.}
We evaluate PMDM and DecompOpt both in their standard forms and in augmented variants.
\(\mathrm{MS2}\) and \(\mathrm{MS3}\) denote two-step and three-step applications of the corresponding baseline.
\(\mathrm{MS+ES}\) denotes the multi-step variant with early stopping, where the stopping rule uses the same target-pocket Vina-score criterion as SurfSpec.
These variants test whether SurfSpec's performance can be explained by repeated application of existing lead-optimization baselines or by early stopping alone.

\paragraph{Naive size extension.}
We include \(\mathrm{DiffSBDD\text{-}SizeExt}\) as a naive size-extension baseline because DiffSBDD is the pocket-conditioned diffusion prior used inside SurfSpec.
This baseline increases the generated molecule size using DiffSBDD without selecting under-occupied target-surface regions.
It therefore tests whether improved geometric mismatch can be obtained simply by adding more atoms, rather than through controlled surface-directed growth.

\paragraph{Surrogate off-target-aware baseline.}
ActivityDiff is evaluated as a surrogate off-target-aware specificity baseline.
Because the true off-target set is unavailable during optimization, we construct its negative guidance using randomly selected negative targets.
This setting evaluates whether approximate off-target information provides reliable specificity guidance.

\paragraph{Ours.}
SurfSpec uses the same DiffLinker proposal mechanism as the refinement baselines, but applies a low-noise anchored recovery procedure.
After DiffLinker proposes a candidate pseudo-label, we first remove atoms whose distance to any pocket atom is less than \(2\,\text{\AA}\).
This clash-truncation step produces the anchor used for refinement.
We then run the low-noise mode-recovery update at \(\tau=0.15\) for 100 steps with step size \(\gamma_{\mathrm{proj}}=0.3\).
The anchor-guidance coefficient is set to \(0.001\).
The anchor target is the truncated pseudo-label, and the guidance score is evaluated at the current denoising state.
After mode recovery, we perform 50 reverse-SDE steps from \(t=\tau\) to
\(t=0\) to obtain the refined ligand. In our experiments, this simpler SDEdit-style reverse process was sufficient for refinement, so we omit additional guidance along the reverse trajectory.
Finally, the ligand is locally optimized using AutoDock Vina before proceeding to the next outer iteration. The complete procedure is described in Algorithm~\ref{alg:full_ligand_growth}, and Figure~\ref{fig:ligand_growth.pdf} illustrates the ligand-growth steps.


\begin{algorithm}[t]
\caption{Full Ligand Growth Pipeline with Geometric Pseudo-Label Generation and Low-Noise Refinement}
\label{alg:full_ligand_growth}
\begin{algorithmic}[1]
\Require native target pocket \(\Tgt\), initial ligand \(x^{(0)}\), pretrained linker generation model \(p_{\mathrm{linker}}\), pretrained pocket-conditioned ligand prior \(p(x\mid\Tgt)\), maximum number of growth iterations \(K\), low noise level \(\tau\), anchor strength \(\lambda\)
\Ensure optimized ligand \(x^{\mathrm{final}}\)

\For{\(k=0,\ldots,K-1\)}

    \State Construct the current ligand-oriented pocket surface for \(x^{(k)}\)
    \State Identify unoccupied pocket surface surface patches around the current ligand
    \State Select the closest unoccupied surface patch \(r^{(k)}\) whose distance from \(x^{(k)}\) lies within \(4\)--\(10\) \AA

    \If{no feasible surface patch \(r^{(k)}\) exists}
        \State \textbf{break}
    \EndIf

    \State Sample a linker from the pretrained linker generation model:
    \[
    z^{(k)} \sim p_{\mathrm{linker}}\!\left(z\mid x^{(k)}, r^{(k)}, \Tgt\right)
    \]

    \State Attach the sampled linker to the current ligand
    \State Remove linker atoms that clash with the pocket under the \(2\) \AA{} clash threshold
    \State Obtain the geometric pseudo-label \(\xanc^{(k)}\)

    \State Define the anchored density around the geometric pseudo-label:
    \[
    p^{\mathrm{anc}}(x\mid \Tgt,\xanc^{(k)})
    \propto
    p(x\mid\Tgt)
    \exp\!\left(
    -\frac{\lambda}{2}
    \|x-\xanc^{(k)}\|^2
    \right)
    \]

    \State Recover the low-noise mode \(y^{\star(k)}\) of \(p_\tau^{\mathrm{anc}}\) by solving the gradient flow:
    \begin{align*}
    \frac{d y_s}{ds}
    =
    &\nabla_{y_s}\log p_\tau(y_s\mid\Tgt)\nonumber\\
    &-
    \lambda
    J_{\hat{x}_{0|\tau}}(y_s)^\top
    \left(
    \hat{x}_{0|\tau}(y_s)-\xanc^{(k)}
    \right)
    \end{align*}

    \State Re-noise the recovered low-noise mode:
    \[
    x_\tau^{\mathrm{anc},(k)}
    =
    y^{\star(k)}
    +
    \sigma_\tau \varepsilon,
    \qquad
    \varepsilon\sim\mathcal{N}(0,I)
    \]

    \State Starting from \(x_\tau^{\mathrm{anc},(k)}\), run the reverse stochastic differential equation from $t=\tau$ to $t=0$:
    \[
    d x_t
    =
    \left[
    f_t(x_t)
    -
    g_t^2\nabla_{x_t}\log p_t(x_t\mid\Tgt)
    \right]dt
    +
    g_t\,d\bar{W}_t
    \]

    \State Set the refined ligand as the next ligand \(x^{(k+1)}\)
    \State Locally optimize the refined ligand \(x^{(k+1)}\) with AutoDock Vina (Vina Min)
    \If{\(k > 0\) and affinity drop is observed}
        \State \textbf{break}
    \EndIf
    \State \(x^\mathrm{final} = x^{(k + 1)}\)

\EndFor

\State \Return the last refined ligand as \(x^{\mathrm{final}}\)

\end{algorithmic}
\end{algorithm}

\paragraph{Refinement baselines.}
For Table~\ref{tab:refinement}, all refinement baselines are evaluated
under a shared one-step pseudo-label recovery setting.
They consume the same DiffLinker pseudo-labels and use the same DiffSBDD
pocket-conditioned ligand prior, and each method generates one refined
sample for each CrossDocked2020 test complex.
Thus, the comparison isolates the effect of the recovery procedure rather
than the pseudo-label generator, diffusion prior, or outer growth policy.
All refinement baselines use the same linker size, closest-patch selection rule, and Vina-based early stopping criterion described above.
The baseline-specific refinement hyperparameters are summarized in Table~\ref{tab:refinement_hyperparams}.

\begin{table*}[t]
\centering
\small
\setlength{\tabcolsep}{5pt}
\renewcommand{\arraystretch}{1.15}
\begin{tabularx}{\textwidth}{
    l
    >{\raggedright\arraybackslash}p{0.30\textwidth}
    >{\raggedright\arraybackslash}X
}
\toprule
Method & Key steps & Other hyperparameters \\
\midrule
SDEdit
& 50 reverse steps
& \(\tau=0.15\) \\

DPS
& 1000 reverse steps from \(t=1\)
& Anchor weight \(=0.001\) \\
DCDP
& 10 purification iterations; 100 inner consistency steps; 50 reverse steps per purification
& Consistency LR \(=0.01\); consistency weight \(=1.0\); purification \(\tau:0.15\rightarrow0.05\); raw anchor gradients \\

RED-Diff
& 50 optimization steps
& LR \(=0.001\); anchor weight \(=0.001\); prior weight \(=0.5\); raw anchor gradients \\

RSD
& 50 optimization steps; 4 particles
& LR \(=0.001\); anchor weight \(=0.001\); prior weight \(=0.5\); repulsion weight \(=0.001\); kernel bandwidth \(=1.0\); raw anchor gradients \\
\bottomrule
\end{tabularx}
\caption{
Hyperparameters for refinement baselines.
All methods share the same outer ligand-growth loop, DiffLinker proposal, and Vina-based early stopping rule.
}
\label{tab:refinement_hyperparams}
\end{table*}

\subsection{Surface Patch Selection Policy}
\label{app:surface_patch_selection}

At each ligand-growth iteration, we select a protein surface patch as the target pseudo-fragment for linker generation. As the selection policy, we first construct visibility-filtered surface candidates and then selects the candidate closest to the current ligand.
The selected candidate is used as the surface-side fragment for DiffLinker with fixed linker size 10.

\paragraph{Candidate construction.}
Let \(\mathcal L\) be the set of ligand heavy-atom coordinates and let \(\mathcal R\) be the receptor heavy atoms.
We first define the local pocket atom set
\[
\mathcal P
=
\left\{
x\in\mathcal R:
\min_{\ell\in\mathcal L}
\|x-\ell\|_2
\le
10\,\text{\AA}
\right\}.
\]
From \(\mathcal P\), we extract an outer surface shell \(\mathcal S\) by retaining the outer \(20\%\) of atoms with respect to distance from the pocket centroid.
Specifically, for \(c_{\mathcal P}=|\mathcal P|^{-1}\sum_{x\in\mathcal P}x\) and \(\rho(x)=\|x-c_{\mathcal P}\|_2\), we set
\[
\mathcal S
=
\left\{
x\in\mathcal P:
\rho(x)
\ge
Q_{0.80}
\left(
\{\rho(y):y\in\mathcal P\}
\right)
\right\}.
\]

We first construct a residue pool by selecting protein residues whose closest atom lies 4–10 Å from the current ligand. Each candidate fragment \(r\) is represented by the connected carbon atoms forming the backbone of the corresponding residue.

We compute its distance to the extracted surface shell as
\[
d_{\mathrm{surf}}(r)
=
\min_{x\in r,\ s\in\mathcal S}
\|x-s\|_2 .
\]
If at least one candidate satisfies \(d_{\mathrm{surf}}(r)\le 2.5\,\text{\AA}\), we keep only those surface-proximal candidates.

\paragraph{Visibility filtering.}
We further require each candidate to have line-of-sight to at least one ligand atom.
For a candidate fragment \(r\), a ligand atom \(\ell_i\), and a representative candidate atom \(x\in r\), we sample the segment
\[
p(t)=(1-t)\ell_i+tx,
\qquad
t\in\{0.1,0.2,\ldots,0.9\}.
\]
Blockers are receptor heavy atoms excluding atoms from the same  as \(r\).
With clash cutoff \(\delta=2.5\,\text{\AA}\), we define
\[
\begin{aligned}
f_{\mathrm{block}}(r,\ell_i)
&=
\frac{1}{|\mathcal T|}
\sum_{t\in\mathcal T}
\mathbf 1
\left[
\min_b
\|p(t)-b\|_2
<
\delta
\right], \\
\mathcal T
&=
\{0.1,0.2,\ldots,0.9\}.
\end{aligned}
\]
A candidate is retained if there exists a ligand atom \(\ell_i\) such that
\[
f_{\mathrm{block}}(r,\ell_i)\le 0.05.
\]
Let \(\mathcal C\) denote the resulting visibility-filtered candidate set.

\paragraph{Closest surface policy.}
Among the final candidates, we select the fragment with minimum ligand--fragment distance:
\[
r^\star
=
\arg\min_{r\in\mathcal C}
d_{\mathrm{lig}}(r),
\qquad
d_{\mathrm{lig}}(r)
=
\min_{\ell\in\mathcal L,\ x\in r}
\|x-\ell\|_2 .
\]
If \(\mathcal C\) is empty, the growth iteration stops.
Otherwise, \(r^\star\) is used as the surface patch, and DiffLinker grows a linker of size 10 toward \(r^\star\).

\subsection{Evaluation Metrics}
\label{app:eval_metrics}

\paragraph{Off-target agnostic flag.}
For ligand-growth experiments, we indicate whether each method is \emph{off-target agnostic}.
A method is marked as off-target agnostic if it does not use off-target structures, scores, or activity labels during generation.
This is a method-level property rather than a per-sample metric.
Off-target pockets are used only for evaluation unless explicitly stated otherwise.

\paragraph{Clash rate.}
We measure protein--ligand clashes using a distance-based criterion.
A generated ligand is considered clashing if any ligand atom is within \(2\,\text{\AA}\) of any protein atom.
The reported clash rate is the fraction of generated ligands that satisfy this condition:
\[
\mathrm{ClashRate}
=
\frac{1}{N}
\sum_{i=1}^{N}
\mathbf{1}
\left[
\min_{a\in L_i,\; b\in P_i}
\|a-b\|_2
<
2\,\text{\AA}
\right].
\]
Lower clash rate indicates fewer physically implausible protein--ligand overlaps.

\paragraph{Ligand size.}
We report the mean number of heavy atoms in the generated ligand.
Atom counts are computed using RDKit from the generated SDF files, with hydrogens excluded.

\paragraph{Empirical specificity.}
We evaluate cross-pocket specificity using AutoDock Vina docking scores.
For each generated ligand, we dock it to the target pocket and to a finite evaluated off-target set \(\Off_{\mathrm{eval}}\).
Let \(S_{\mathrm{tgt}}(L)\) denote the Vina docking score on the target pocket, and let \(S_O(L)\) denote the score on off-target pocket \(O\).
Because lower Vina scores indicate stronger predicted binding, we define the empirical specificity estimate as
\[
\widehat{\Spec}(L)
=
-\left(
S_{\mathrm{tgt}}(L)
-
\min_{O\in\Off_{\mathrm{eval}}} S_O(L)
\right).
\]
This metric compares the target-pocket score against the strongest predicted off-target binder among the evaluated off-targets. In addition to the sample average, we report thresholded success rates for
\(\widehat{\Spec}>0.2\), \(\widehat{\Spec}>0.4\), and
\(\widehat{\Spec}_{\mathrm{avg}}>0.6\).

\paragraph{Geometric mismatch.}
Geometric mismatch measures ligand--pocket surface complementarity using
the Jensen--Shannon distance defined in
Eq.~\eqref{eq:js_geometric_mismatch}.
For each ligand--pocket pair, we construct the ligand surface and the
ligand-oriented target-pocket surface restricted to points within
\(8\,\text{\AA}\) of any ligand heavy atom.
These surfaces induce probability measures on the common ligand-local
domain through the signed-distance-based Boltzmann construction in
Appendix~\ref{app:surface_probability_field}.
Lower geometric mismatch indicates better agreement between the ligand
surface-induced measure and the ligand-oriented target-pocket measure.

\paragraph{Pocket occupancy.}
Pocket occupancy measures how much of the target pocket is covered by the generated ligand.
We define occupancy as the fraction of target-pocket grid points that are occupied by the ligand, where a grid point is counted as occupied if it lies within the predefined ligand fill radius of any ligand atom.
We report the average, first quartile, and third quartile across samples.
Higher occupancy indicates broader coverage of the target-pocket region.

\paragraph{Target-pocket Vina docking.}
We report the AutoDock Vina docking score of each generated ligand on
the target pocket.
This metric is distinguished from the Vina-derived score used during
optimization, as it is computed by re-docking the final generated ligand
with AutoDock Vina under the evaluation protocol.
Because lower Vina scores indicate stronger predicted binding, lower
values correspond to better target-pocket docking performance.
We report the average, first quartile, and third quartile across
generated samples.

\paragraph{Refinement prior consistency.}
For refinement experiments, we evaluate whether a refined ligand remains consistent with the pocket-conditioned ligand prior.
We report protein clash rate using the same \(2\,\text{\AA}\) criterion as above, RDKit valence validity, Vina score-only repulsion diagnostics, and bond-length MMD.
The repulsion diagnostics use the unweighted repulsion term from Vina \texttt{--score\_only}.
We compute this diagnostic globally on the full receptor and locally on receptor atoms within \(4\,\text{\AA}\) of the fragment neighborhood.
Larger repulsion indicates stronger steric conflict and is not interpreted as docking affinity.

\paragraph{Valence validity.}
Valence validity is the fraction of expected samples that pass RDKit valence and sanitization checks.
Missing ligands and invalid molecules are counted as failures.
Higher validity indicates better chemical consistency.

\paragraph{Bond-length MMD.}
Bond-length MMD measures the discrepancy between the bond-length distribution of generated ligands and the reference distribution from the dataset.
We compute maximum mean discrepancy separately for C--C, C--N, and C--O bonds.
Lower values indicate that the generated bond geometry is closer to the reference molecular geometry.

\paragraph{Pseudo-label faithfulness.}
We evaluate whether refinement preserves the intended pseudo-label using RMSD and topological similarity.
RMSD is computed against the pseudo-label SDF using the same atom indexing, without MCS remapping:
\[
\mathrm{RMSD}
=
\sqrt{
\frac{1}{n}
\sum_{j=1}^{n}
\|x_j-\tilde{x}_j\|_2^2
},
\]
where \(x_j\) and \(\tilde{x}_j\) denote corresponding atoms in the refined ligand and pseudo-label, respectively.
Lower RMSD indicates better preservation of pseudo-label geometry.
Topological similarity is computed as the Tanimoto similarity between RDKit path fingerprints of the refined ligand and pseudo-label using \texttt{RDKFingerprint} with \texttt{maxPath=7}.
Higher topological similarity indicates better preservation of pseudo-label topology.

\section{Additional Results}
\label{app:additional_results}

\subsection{Empirical Activation of the Geometric Specificity Certificate}
\label{app:certificate_nonvacuity}

The specificity certificate in Theorem~\ref{thm:spec_bound} is based on a
geometric condition: the target--off-target pocket separation should
exceed the target--ligand mismatch.
We empirically evaluate how often this condition is activated on randomly
sampled target--off-target pocket pairs from the CrossDocked2020 test
complexes.
For each pair, we use the reference ligand \(L\) associated with the
target pocket to construct the ligand-oriented pocket measures
\(\nu_{\Tgt}^{L}\) and \(\nu_O^{L}\), and compute the target--ligand
mismatch \(\dgm(L,\Tgt)\) and the ligand-oriented target--off-target
separation
\[
m_L(O)
=
d_{\mathrm{JS}}(\nu_{\Tgt}^{L},\nu_O^{L}).
\]
We define the geometric certificate margin as
\begin{align}
M_L(O)
&=
m_L(O)-\dgm(L,\Tgt)
\\
&=
d_{\mathrm{JS}}(\nu_{\Tgt}^{L},\nu_O^{L})
-
\dgm(L,\Tgt).
\end{align}
When \(M_L(O)>0\), the triangle-inequality lower bound
\[
\dgm(L,O)
\ge
[M_L(O)]_+
\]
is nonzero.
Thus, \(M_L(O)>0\) measures activation of the geometric lower-bound term.

Figure~\ref{fig:certificate_margin_ecdf} reports the empirical CDF of
\(M_L(O)\) over the filtered usable sampled pairs.
We find that \(65.96\%\) of pairs have positive margin, showing that the
geometric certificate is activated for a substantial fraction of evaluated
target--off-target pairs.
This indicates that the triangle-inequality certificate yields a nonzero
ligand--off-target mismatch lower bound in many evaluated cases.

Table~\ref{tab:certificate_nonvacuity_stats} summarizes the resulting
geometric lower-bound statistics.
The positive activation rate shows that the conservative geometric
lower-bound term is not purely vacuous on this benchmark.

\begin{figure}[t]
\centering
\includegraphics[
    width=0.7\columnwidth,
    clip,
    keepaspectratio
]{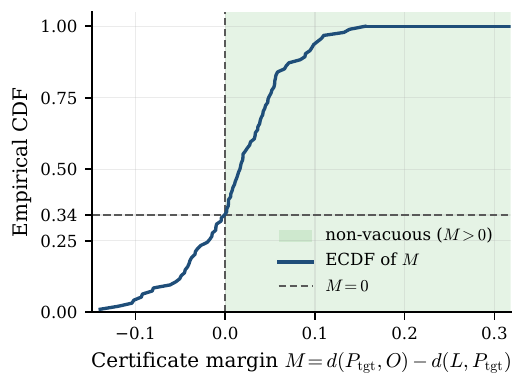}
\caption{
Empirical activation of the geometric specificity certificate.
The plot shows the empirical CDF of the certificate margin
\(M_L(O)=m_L(O)-\dgm(L,\Tgt)\), where
\(m_L(O)=d_{\mathrm{JS}}(\nu_{\Tgt}^{L},\nu_O^{L})\) is the
ligand-oriented target--off-target pocket separation.
The shaded region indicates \(M_L(O)>0\), where the triangle-inequality
lower bound on ligand--off-target mismatch is nonzero.
Overall, \(65.96\%\) of filtered usable pairs fall in this activated
region.
}
\label{fig:certificate_margin_ecdf}
\end{figure}

\begin{table}[t]
\centering
\caption{
Geometric activation statistics of the specificity certificate on the
filtered usable target--off-target pairs.
Here \(M=M_L(O)\) and \(\mathrm{LB}=[M]_+\) denotes the
triangle-inequality lower bound on ligand--off-target geometric mismatch.
}
\label{tab:certificate_nonvacuity_stats}
\scriptsize
\setlength{\tabcolsep}{4.0pt}
\renewcommand{\arraystretch}{1.12}
\begin{tabular}{lccccc}
\toprule
\textbf{Metric} & \textbf{Avg.} & \textbf{Std.} & \textbf{Q1} & \textbf{Med.} & \textbf{Q3} \\
\midrule
\(\Pr[M>0]\) & \(65.96\%\) & -- & -- & -- & -- \\
Margin \(M\) & 0.014 & 0.057 & -0.017 & 0.017 & 0.049 \\
Lower bound \(\mathrm{LB}\) & 0.031 & 0.037 & 0.000 & 0.017 & 0.049 \\
\bottomrule
\end{tabular}
\end{table}

\subsection{Additional Validation with Resampled Off-Targets}
\label{app:additional_offtarget_results}

To further validate the robustness of the main lead-optimization results,
we repeat the off-target evaluation under a different evaluation setting.
Specifically, for each target pocket, we construct a new
\(\Off_{\mathrm{eval}}\) by randomly sampling a different set of 10
off-target pockets from the remaining CrossDocked2020 test pockets,
distinct from those used in Table~\ref{tab:main}.
We also relax the AutoDock Vina timeout from 300 seconds to 1800 seconds
per docking run.
This setting is conservative in the sense that it reduces the influence
of docking timeouts and provides more complete target/off-target docking
evaluations.
We additionally report the docking success rate, defined as the fraction
of attempted Vina dockings that return a finite score under the timeout
policy.

Table~\ref{tab:supple_main} reports the full evaluation under this
resampled off-target set and relaxed docking budget.
The \(\widehat{\mathrm{Spec}}\) block is the primary metric in this
analysis, and its second line reports the relative improvement over the
initial lead.
SurfSpec achieves the largest empirical specificity improvement, including
the best average score and the best specificity success rates at all
evaluated thresholds.
It also obtains the best average pocket occupancy, zero clash rate, and
full docking success, while maintaining competitive target-pocket docking
performance.
Although the DiffSBDD size-extension baseline achieves the lowest
geometric mismatch, it produces substantially higher clash rates, lower
docking success, and weaker empirical specificity.
These results support that SurfSpec improves target-over-off-target
preference through controlled surface-directed growth rather than through
naive ligand expansion or docking-failure artifacts.

Table~\ref{tab:spec_seed42} further evaluates specificity under a new
random seed for ligand generation and off-target sampling.
Even under this independent seed, SurfSpec achieves the best empirical
specificity across the average score and all specificity success rates.
This additional specificity-focused comparison supports that the observed
improvement is robust to the sampled ligands and off-target set, rather
than being driven by a particular random seed.


\begin{table*}[t]
    \centering
    \small
    \setlength{\tabcolsep}{2.0pt}
    \renewcommand{\arraystretch}{1.08}

    \definecolor{impblue}{RGB}{0,82,155}
    \definecolor{impred}{RGB}{180,35,35}
    \definecolor{impgray}{RGB}{100,100,100}

    \newcommand{\posimp}[1]{{\scriptsize\textcolor{impblue}{#1}}}
    \newcommand{\negimp}[1]{{\scriptsize\textcolor{impred}{#1}}}
    \newcommand{\zeroimp}[1]{{\scriptsize\textcolor{impgray}{#1}}}
    \newcommand{\bestimp}[1]{{\scriptsize\bfseries\textcolor{impblue}{#1}}}

    \newcommand{\specval}[2]{%
    \begin{tabular}[c]{@{}c@{}}
    #1\\[-1.5pt]
    #2
    \end{tabular}%
    }

    \resizebox{\textwidth}{!}{%
    \begin{tabular}{
        l
        c
        @{\hspace{3pt}}
        >{\columncolor{blue!4}}c
        >{\columncolor{blue!4}}c
        >{\columncolor{blue!4}}c
        >{\columncolor{blue!4}}c
        >{\columncolor{blue!4}}c
        @{\hspace{5pt}}
        ccc
        @{\hspace{5pt}}
        ccc
        @{\hspace{5pt}}
        ccc
        @{\hspace{3pt}}
        c
        c
        c
    }
    \toprule
    \multirow{2}{*}{Method}
    & \multicolumn{1}{c}{Off-Target}
    & \multicolumn{5}{c}{\textbf{Primary Metric: }$\widehat{\mathrm{Spec}}$ $\uparrow$}
    & \multicolumn{3}{c}{Geometric Mismatch $\downarrow$}
    & \multicolumn{3}{c}{Occupancy $\uparrow$}
    & \multicolumn{3}{c}{Vina Dock $\downarrow$}
    & \multirow{2}{*}{\shortstack{Clash\\Rate $\downarrow$}}
    & \multirow{2}{*}{\#\,Atoms}
    & \multirow{2}{*}{\shortstack{Dock\\Success $\uparrow$}} \\
    \cmidrule(lr){3-7}
    \cmidrule(lr){8-10}
    \cmidrule(lr){11-13}
    \cmidrule(lr){14-16}
    & Agnostic
    & Avg. & Std. & $>{0.2}$ & $>{0.4}$ & $>{0.6}$
    & Avg. & Std. & Med.
    & Avg. & Q1 & Q3
    & Avg. & Q1 & Q3
    &  &  &  \\
    \midrule

    Initial Lead
    & O
    & -0.81 & 1.43 & 0.15 & 0.14 & 0.11
    & 0.540 & 0.051 & 0.536
    & 0.26 & 0.22 & 0.30
    & -7.17 & -8.67 & -5.73
    & 0.00 & 22.75 & 1.00 \\
    \hline

    Delete
    & O
    & \specval{\underline{-0.81}}{\zeroimp{\(+0\%\)}} & 1.54
    & \specval{\underline{0.19}}{\posimp{\(+27\%\)}}
    & \specval{0.14}{\zeroimp{\(+0\%\)}}
    & \specval{0.13}{\posimp{\(+18\%\)}}
    & 0.539 & 0.052 & 0.535
    & 0.26 & 0.22 & 0.30
    & -7.23 & -8.65 & -5.73
    & \textbf{0.00} & 22.85 & 1.00 \\

    PMDM
    & O
    & \specval{-1.10}{\negimp{\(-36\%\)}} & 2.03
    & \specval{0.15}{\zeroimp{\(+0\%\)}}
    & \specval{0.12}{\negimp{\(-14\%\)}}
    & \specval{0.12}{\posimp{\(+9\%\)}}
    & 0.501 & 0.056 & 0.489
    & \underline{0.33} & 0.27 & 0.38
    & -6.47 & -8.87 & -5.10
    & 0.16 & 31.00 & 0.92 \\

    DecompOpt
    & O
    & \specval{-0.93}{\negimp{\(-14\%\)}} & 1.43
    & \specval{0.13}{\negimp{\(-13\%\)}}
    & \specval{0.12}{\negimp{\(-14\%\)}}
    & \specval{0.11}{\zeroimp{\(+0\%\)}}
    & 0.542 & 0.053 & 0.536
    & 0.25 & 0.22 & 0.29
    & -6.08 & -8.26 & -4.49
    & 0.02 & \underline{22.75} & 0.91 \\

    DiffShape
    & O
    & \specval{-0.88}{\negimp{\(-9\%\)}} & 1.58
    & \specval{0.13}{\negimp{\(-13\%\)}}
    & \specval{0.13}{\negimp{\(-7\%\)}}
    & \specval{0.12}{\posimp{\(+9\%\)}}
    & 0.585 & 0.115 & 0.555
    & 0.14 & 0.00 & 0.25
    & -6.90 & -8.33 & -5.20
    & 0.11 & \textbf{21.91} & 0.99 \\

    ActivityDiff
    & X
    & \specval{-0.96}{\negimp{\(-18\%\)}} & 2.15
    & \specval{0.13}{\negimp{\(-13\%\)}}
    & \specval{0.11}{\negimp{\(-21\%\)}}
    & \specval{0.10}{\negimp{\(-9\%\)}}
    & 0.538 & 0.054 & 0.529
    & 0.27 & 0.23 & 0.31
    & -7.07 & -8.35 & -5.39
    & 0.11 & 22.89 & 1.00 \\
    \hline

    PMDM--MS2
    & O
    & \specval{-1.38}{\negimp{\(-70\%\)}} & 2.19
    & \specval{0.13}{\negimp{\(-13\%\)}}
    & \specval{0.13}{\negimp{\(-7\%\)}}
    & \specval{0.12}{\posimp{\(+9\%\)}}
    & 0.499 & 0.063 & 0.482
    & 0.31 & 0.26 & 0.37
    & \textbf{-8.95} & -10.04 & \textbf{-7.01}
    & 0.25 & 33.42 & 0.91 \\

    PMDM--MS3
    & O
    & \specval{-2.05}{\negimp{\(-152\%\)}} & 2.93
    & \specval{0.07}{\negimp{\(-53\%\)}}
    & \specval{0.07}{\negimp{\(-50\%\)}}
    & \specval{0.06}{\negimp{\(-45\%\)}}
    & 0.496 & 0.069 & 0.480
    & 0.31 & 0.26 & 0.37
    & -8.38 & \textbf{-10.51} & -6.43
    & 0.31 & 35.83 & 0.82 \\

    PMDM--MS+ES
    & O
    & \specval{-1.02}{\negimp{\(-26\%\)}} & 2.08
    & \specval{\underline{0.19}}{\posimp{\(+27\%\)}}
    & \specval{\underline{0.17}}{\posimp{\(+21\%\)}}
    & \specval{\underline{0.14}}{\posimp{\(+27\%\)}}
    & 0.500 & 0.057 & 0.489
    & \underline{0.33} & \underline{0.28} & 0.38
    & -8.61 & -10.14 & \underline{-6.56}
    & 0.16 & 31.37 & 0.87 \\

    DecompOpt--MS2
    & O
    & \specval{-1.03}{\negimp{\(-27\%\)}} & 1.33
    & \specval{0.11}{\negimp{\(-27\%\)}}
    & \specval{0.11}{\negimp{\(-21\%\)}}
    & \specval{0.08}{\negimp{\(-27\%\)}}
    & 0.542 & 0.053 & 0.538
    & 0.25 & 0.22 & 0.28
    & -7.08 & -8.79 & -5.62
    & 0.03 & \underline{22.75} & 0.96 \\

    DecompOpt--MS3
    & O
    & \specval{-0.99}{\negimp{\(-22\%\)}} & 1.39
    & \specval{0.12}{\negimp{\(-20\%\)}}
    & \specval{0.09}{\negimp{\(-36\%\)}}
    & \specval{0.08}{\negimp{\(-27\%\)}}
    & 0.541 & 0.053 & 0.537
    & 0.25 & 0.21 & 0.28
    & -7.00 & -8.57 & -5.72
    & 0.04 & 22.76 & 0.97 \\

    DecompOpt--MS+ES
    & O
    & \specval{-0.83}{\negimp{\(-2\%\)}} & 1.59
    & \specval{0.17}{\posimp{\(+13\%\)}}
    & \specval{0.14}{\zeroimp{\(+0\%\)}}
    & \specval{\underline{0.14}}{\posimp{\(+27\%\)}}
    & 0.541 & 0.052 & 0.534
    & 0.25 & 0.22 & 0.28
    & -7.40 & -8.91 & -5.95
    & 0.03 & \underline{22.75} & 0.95 \\
    \hline

    DiffSBDD--SizeExt(+20)
    & O
    & \specval{-1.51}{\negimp{\(-86\%\)}} & 2.24
    & \specval{0.14}{\negimp{\(-7\%\)}}
    & \specval{0.11}{\negimp{\(-21\%\)}}
    & \specval{0.10}{\negimp{\(-9\%\)}}
    & \underline{0.480} & 0.037 & 0.476
    & 0.28 & 0.19 & 0.38
    & -7.57 & -8.25 & -6.37
    & 0.45 & 35.88 & 0.93 \\

    DiffSBDD--SizeExt(+30)
    & O
    & \specval{-2.20}{\negimp{\(-172\%\)}} & 3.76
    & \specval{0.08}{\negimp{\(-47\%\)}}
    & \specval{0.07}{\negimp{\(-50\%\)}}
    & \specval{0.07}{\negimp{\(-36\%\)}}
    & \textbf{0.468} & 0.046 & \textbf{0.461}
    & 0.32 & 0.23 & \textbf{0.42}
    & -7.14 & -8.40 & -5.44
    & 0.62 & 41.58 & 0.81 \\
    \hline

    \rowcolor{blue!10}
    Ours
    & O
    & \specval{\textbf{-0.37}}{\bestimp{\(+54\%\)}} & 2.54
    & \specval{\textbf{0.31}}{\bestimp{\(+107\%\)}}
    & \specval{\textbf{0.29}}{\bestimp{\(+107\%\)}}
    & \specval{\textbf{0.25}}{\bestimp{\(+127\%\)}}
    & \underline{0.480} & 0.050 & \underline{0.467}
    & \textbf{0.36} & \textbf{0.30} & \underline{0.41}
    & \underline{-8.63} & \underline{-10.43} & -6.31
    & \textbf{0.00} & 34.99 & 1.00 \\
    \bottomrule
    \end{tabular}%
    }

    \caption{
    Off-target-agnostic lead optimization analysis.
    The \(\widehat{\mathrm{Spec}}\) block is the primary metric for
    specificity enhancement; its second line reports the relative
    improvement over the initial lead,
    \(100\times(x-x_0)/|x_0|\), except for the standard deviation.
    Blue and red percentages indicate positive and negative improvement,
    respectively, and bold percentages indicate the largest improvement
    in each specificity column.
    SurfSpec achieves the largest empirical specificity improvement,
    demonstrating strong target-over-off-target preference gains without
    off-target access during optimization.
    The geometric mismatch results provide supporting evidence that
    SurfSpec attains this gain through controlled target-pocket geometric
    fitting rather than naive ligand expansion.
    Bold indicates the unique best value; underline marks all
    second-best values when the best value is unique. Docking Success is the fraction of attempted Vina docks that returned a finite score under the 1800\,s hard-timeout policy (timeouts and missing scores count as failures).
    }
    \label{tab:supple_main}
\end{table*}

\begin{table}[t]
    \centering
    \small
    \setlength{\tabcolsep}{5.0pt}
    \renewcommand{\arraystretch}{1.1}
    \begin{tabular}{lccccc}
    \toprule
    \multirow{2}{*}{Method}
    & \multicolumn{5}{c}{\textbf{Primary Metric: }$\widehat{\mathrm{Spec}}$ $\uparrow$} \\
    \cmidrule(lr){2-6}
    & Avg. & Std. & $>{0.2}$ & $>{0.4}$ & $>{0.6}$ \\
    \midrule
    Delete
    & \underline{-0.82} & 1.59 & \underline{0.19} & 0.13 & 0.12 \\
    PMDM--MS+ES
    & -1.02 & 1.86 & \underline{0.19} & 0.14 & \underline{0.14} \\
    DecompOpt--MS+ES
    & -0.93 & 1.36 & 0.17 & \underline{0.15} & \underline{0.14} \\
    \hline
    \rowcolor{blue!8}
    Ours
    & \textbf{-0.52} & 2.21 & \textbf{0.29} & \textbf{0.27} & \textbf{0.25} \\
    \bottomrule
    \end{tabular}
    \caption{
    Specificity-focused comparison under a second random seed for ligand
    generation and off-target sampling, complementing
    Tables~\ref{tab:main} and~\ref{tab:supple_main}.
    Docking is evaluated with a relaxed 1800\,s timeout.
    SurfSpec achieves the best empirical specificity across the average
    score and all thresholded success rates, indicating that its
    specificity improvement is not driven by a particular random seed.
    Bold indicates the unique best value; underline marks all second-best
    values when the best value is unique.
    }
    \label{tab:spec_seed42}
\end{table}

\subsection{Fine-Grained Empirical Specificity Distribution}
\begin{figure}[t]
\centering
\includegraphics[
    width=\columnwidth,
    clip,
    keepaspectratio
]{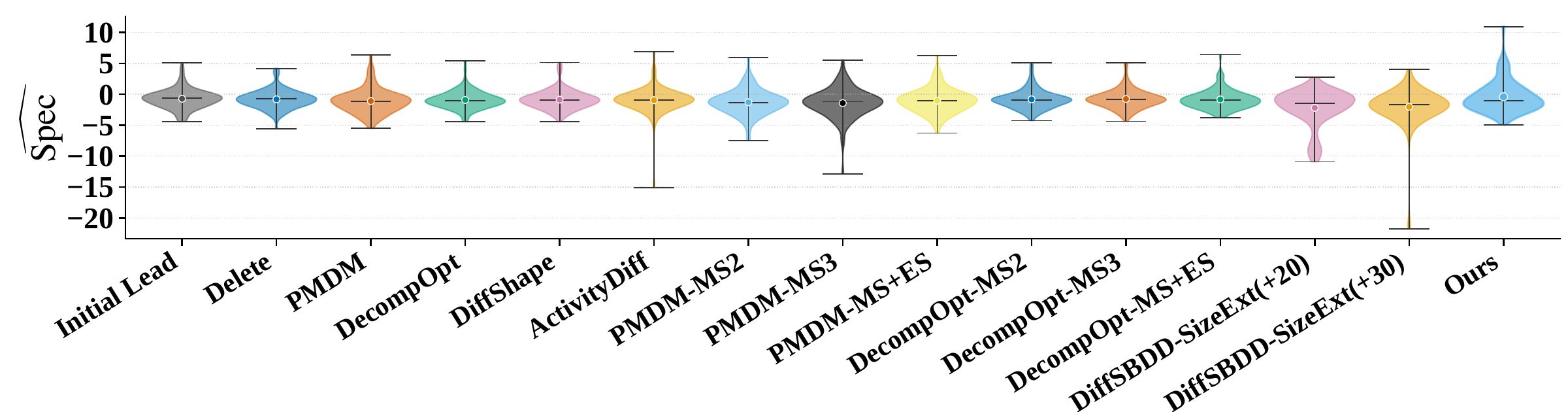}
\caption{
Fine-grained distribution of empirical specificity scores on CrossDocked2020.
Each violin plot shows the per-sample distribution of
\(\widehat{\Spec}\) for one lead optimization method, where
\(\widehat{\Spec}\) measures the target-pocket preference over the
strongest evaluated off-target pocket.
The inner marker indicates the central tendency.
SurfSpec shifts the specificity distribution upward compared with
target-only baselines, indicating more consistent improvement in
target-over-off-target preference.
}
\label{fig:specificity_distribution}
\end{figure}

In addition to the aggregate statistics reported in
Table~\ref{tab:main}, we visualize the full per-sample distribution of
empirical specificity scores.
Figure~\ref{fig:specificity_distribution} shows the distribution of
\(\widehat{\Spec}\) across generated ligands for each method.
SurfSpec produces an upward-shifted distribution relative to the initial
lead and target-only baselines, indicating that the improvement in
specificity is not driven only by a small number of outlier samples.
This supports that target-surface-directed growth provides more
consistent target-over-off-target preference across the benchmark.

\subsection{Qualitative Results}
\label{app:qualitative_comparison}

Figure~\ref{fig:qual_process_comparison} qualitatively compares optimized
ligands on three representative target pockets.
SurfSpec expands ligands toward under-filled pocket regions while
maintaining plausible pocket-bound conformations.
The inset values report the target-pocket geometric mismatch
\(d_{\mathrm{gm}}\), where lower values indicate better surface
complementarity.
These examples support that SurfSpec achieves controlled surface-directed
growth with low geometric mismatch.

\begin{figure*}[t]  
\centering
\begingroup

\setlength{\tabcolsep}{0pt}
\renewcommand{\arraystretch}{1.0}
\setlength{\fboxsep}{0.5pt}
\setlength{\fboxrule}{0.2pt}

\newcommand{\imgw}{0.18\textwidth}
\newcommand{\rowstyle}{\scriptsize\bfseries}
\newcommand{\colstyle}{\scriptsize\bfseries}
\newcommand{\valstyle}{\tiny\bfseries}

\newcommand{\figimg}[2]{%
  \adjustbox{valign=c}{%
    \begin{overpic}[
      width=\imgw,
      trim=2pt 2pt 2pt 2pt,
      clip
    ]{#1}
      \put(74,4){%
        \fcolorbox{black}{white}{\valstyle #2}%
      }
    \end{overpic}%
  }%
}

\newcommand{\rowlab}[1]{%
  \adjustbox{valign=c}{%
    \makebox[0.07\textwidth][l]{\rowstyle #1}%
  }%
}

\newcommand{\collab}[1]{%
  {\colstyle #1}%
}

\makebox[\textwidth][c]{%
\begin{tabular}{
@{}l
c@{\hspace{2pt}}
c@{\hspace{2pt}}
c@{\hspace{2pt}}
c@{\hspace{2pt}}
c@{}
}

&
\collab{Lead}
&
\collab{PMDM}
&
\collab{DecompOpt}
&
\collab{ActivityDiff}
&
\collab{Ours}
\\[1pt]

\rowlab{ACE}
&
\figimg{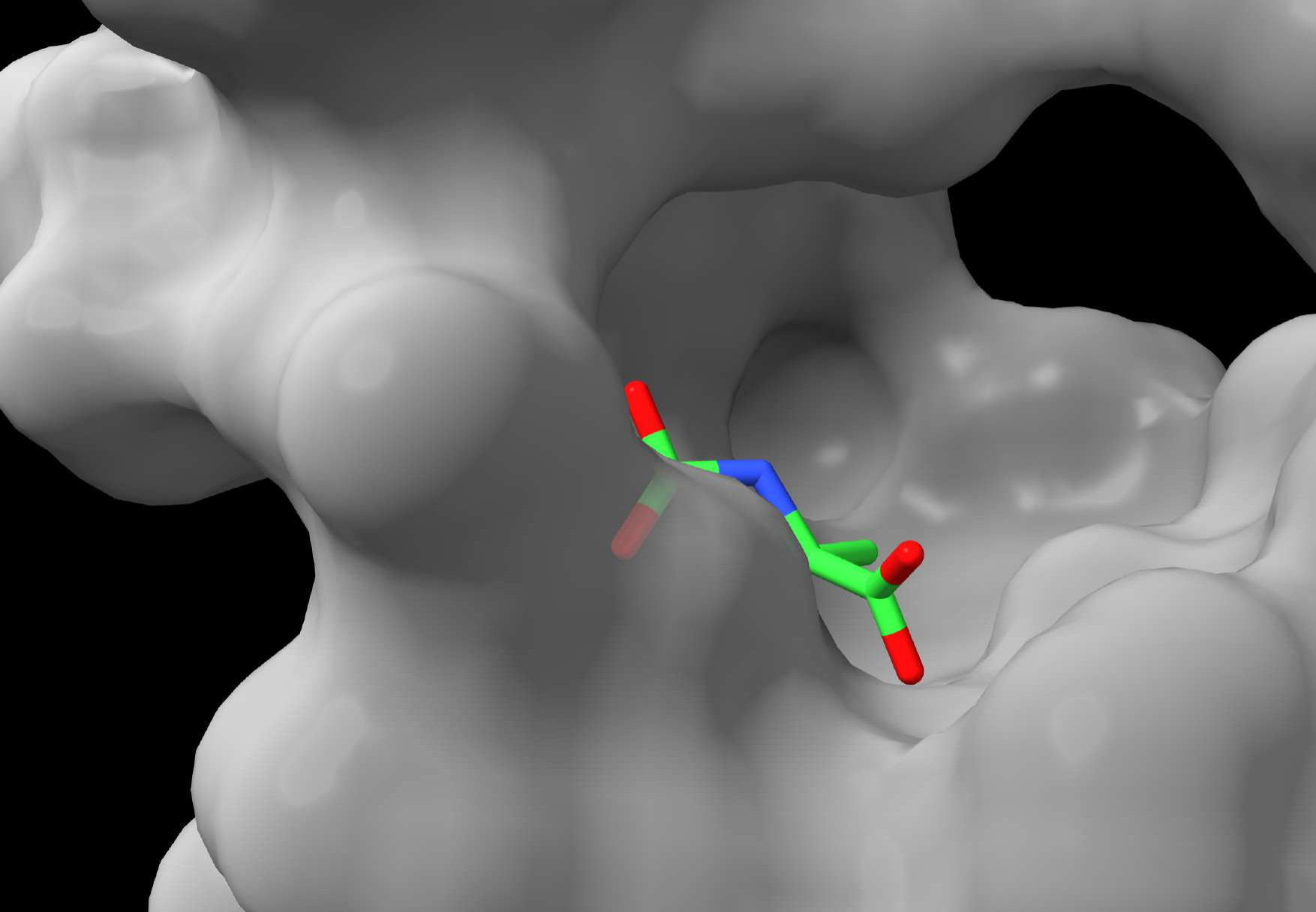}{0.631}
&
\figimg{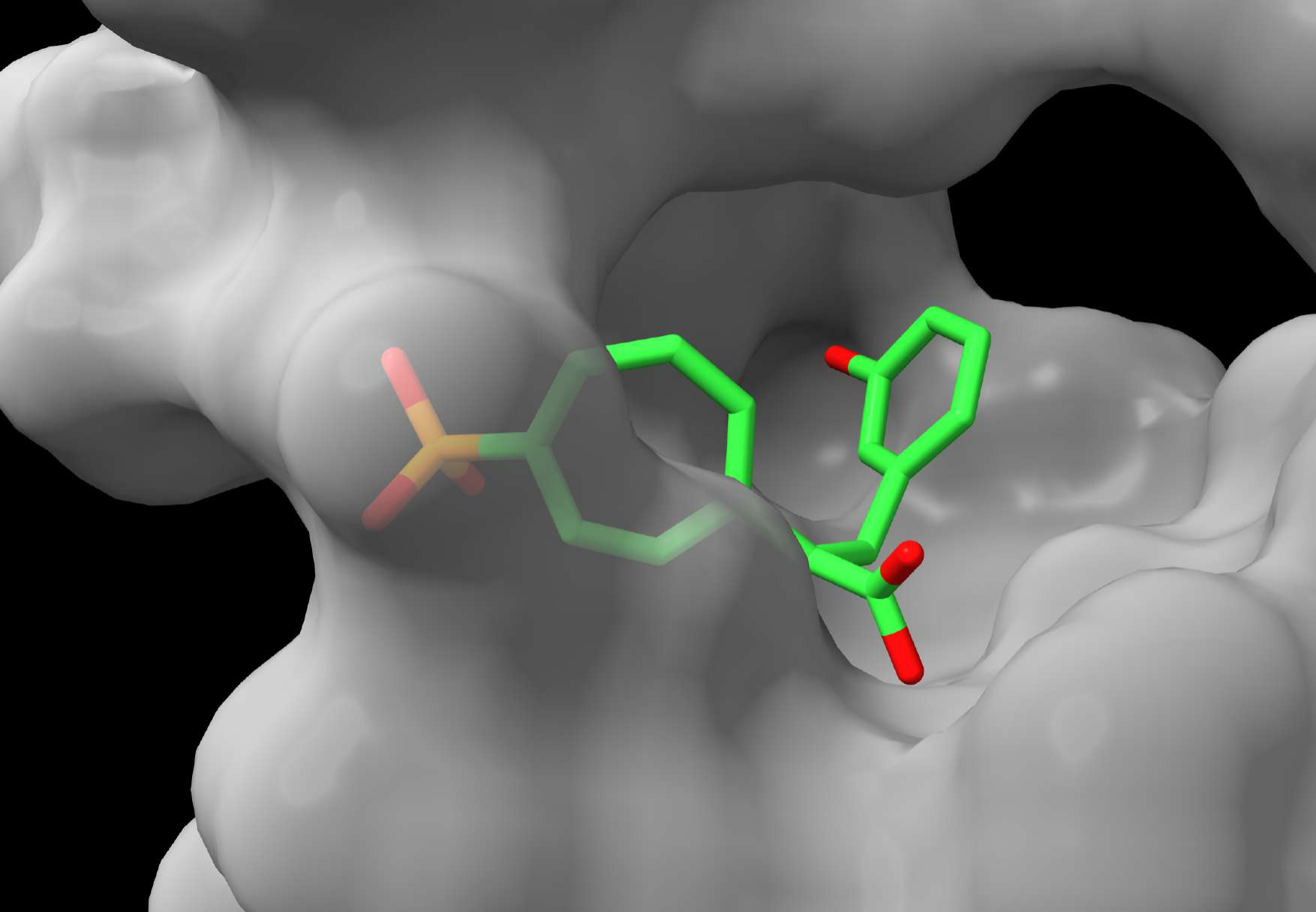}{0.542}
&
\figimg{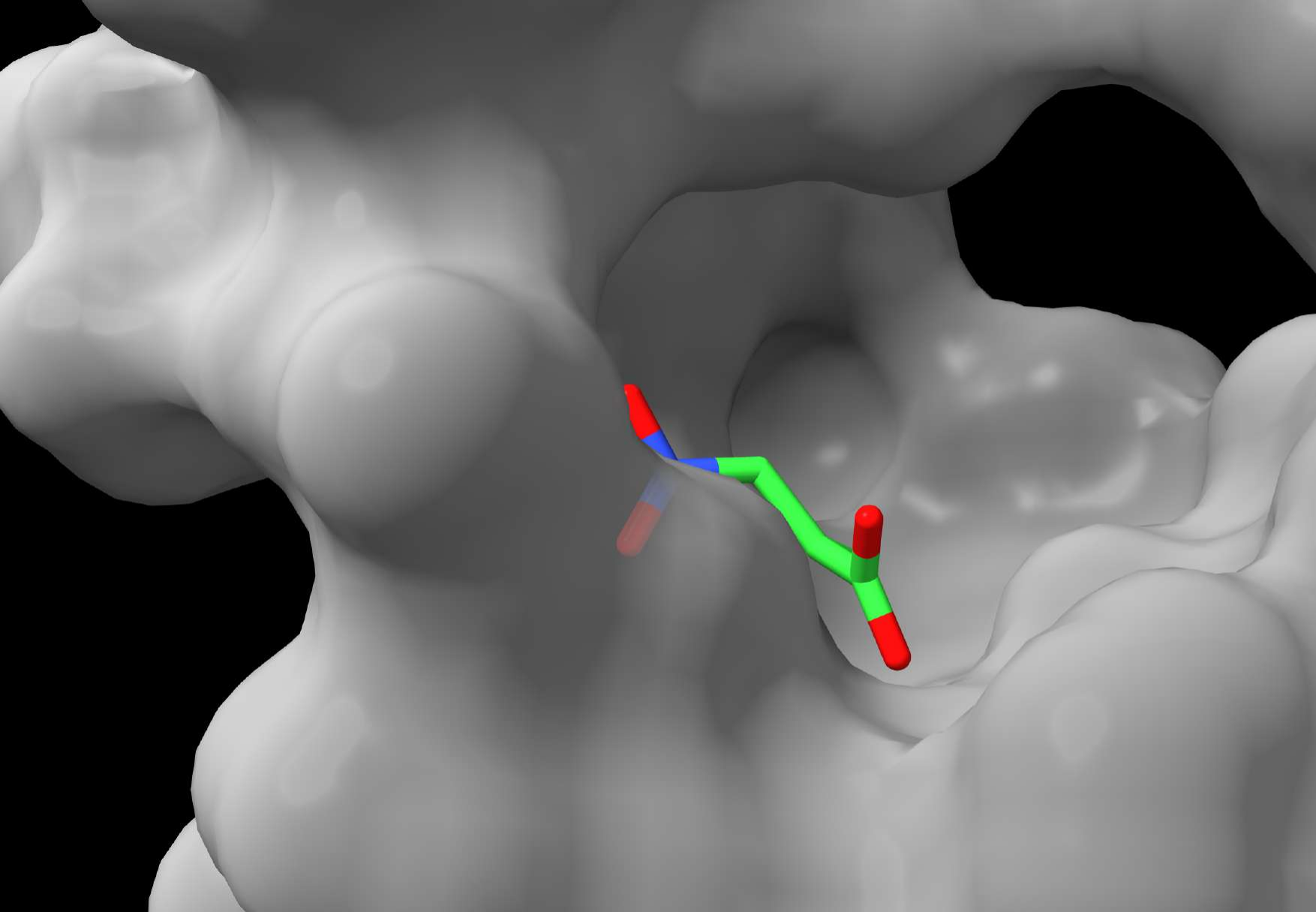}{0.634}
&
\figimg{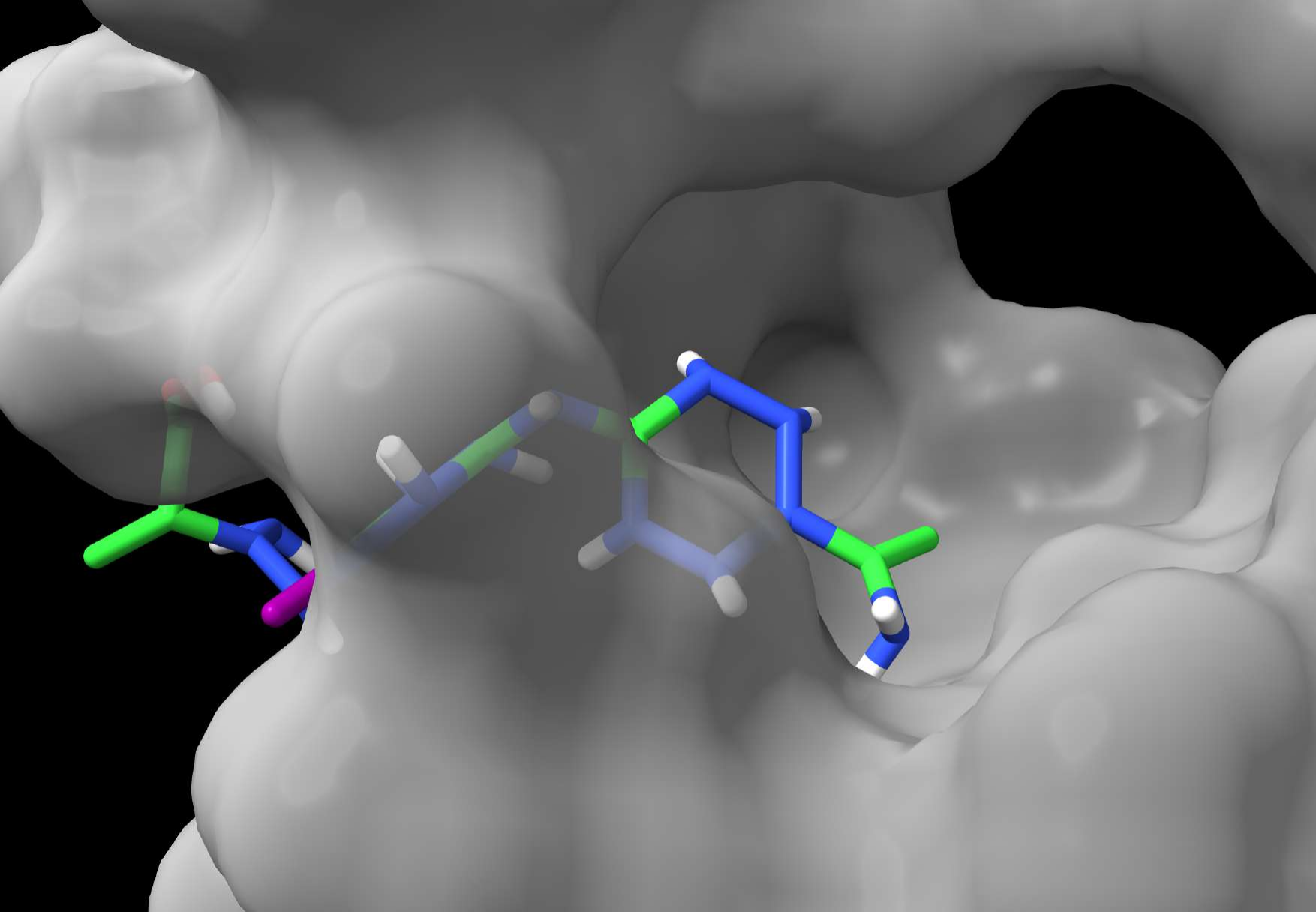}{0.626}
&
\figimg{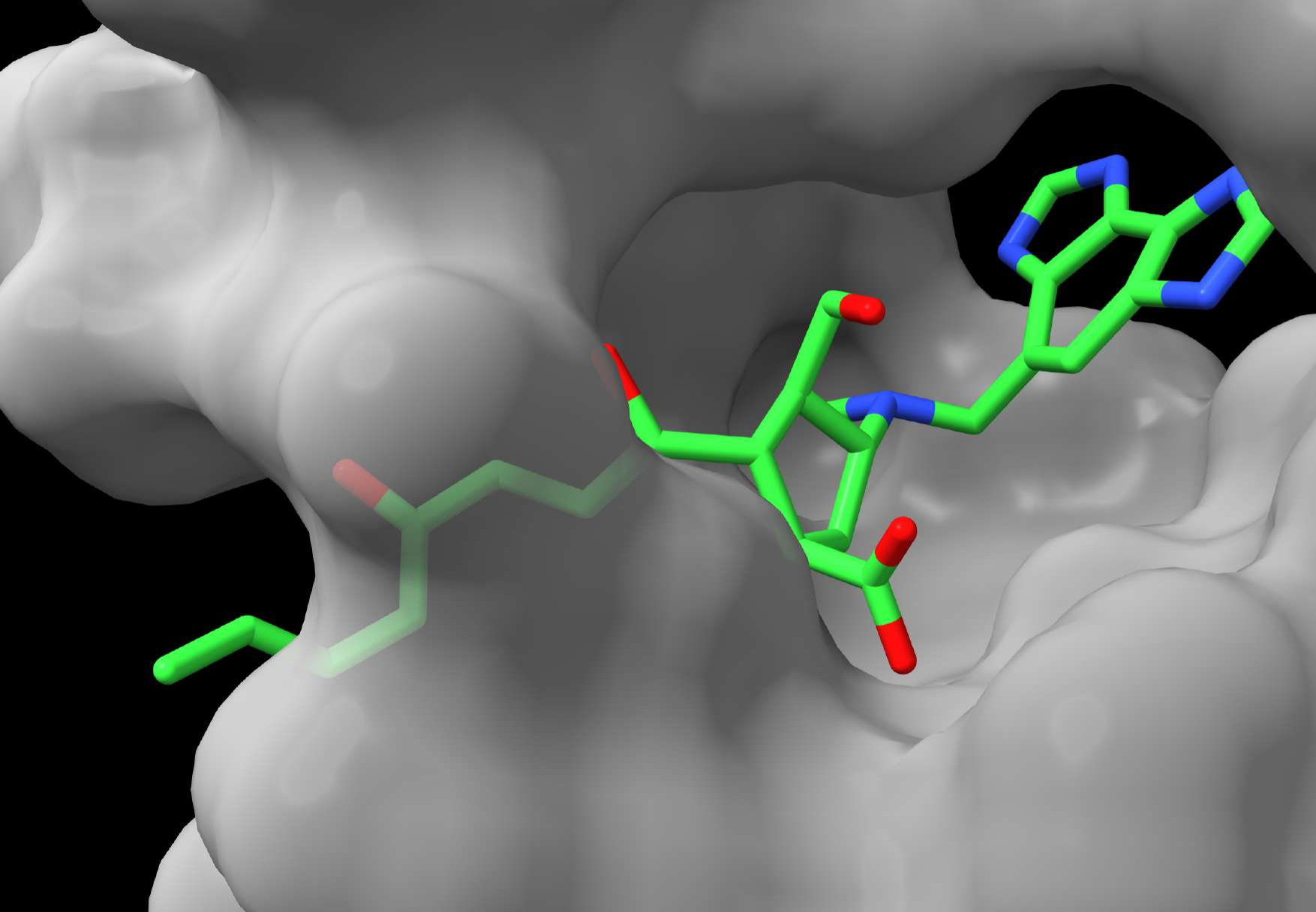}{0.474}
\\[2pt]

\rowlab{ATS5}
&
\figimg{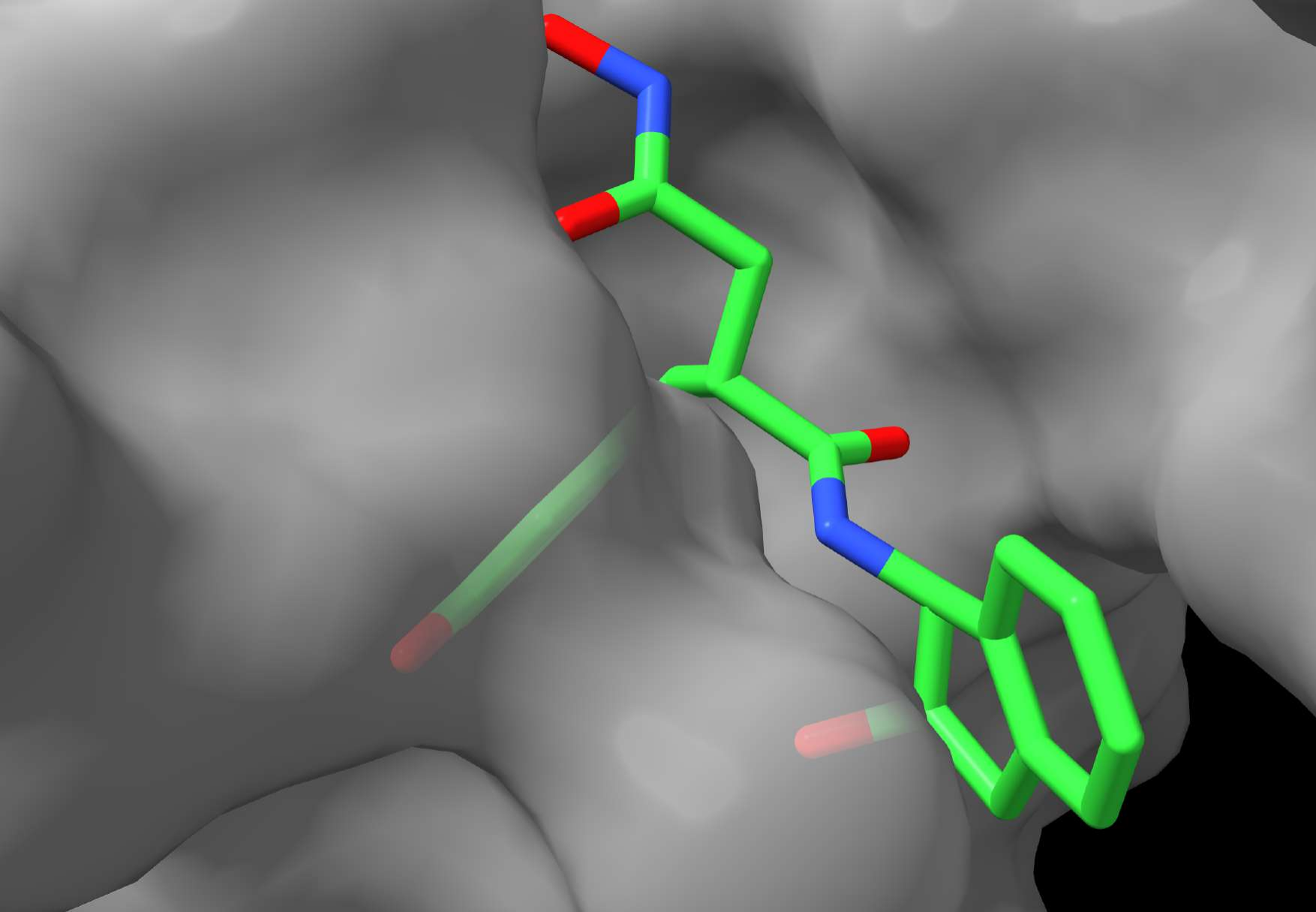}{0.485}
&
\figimg{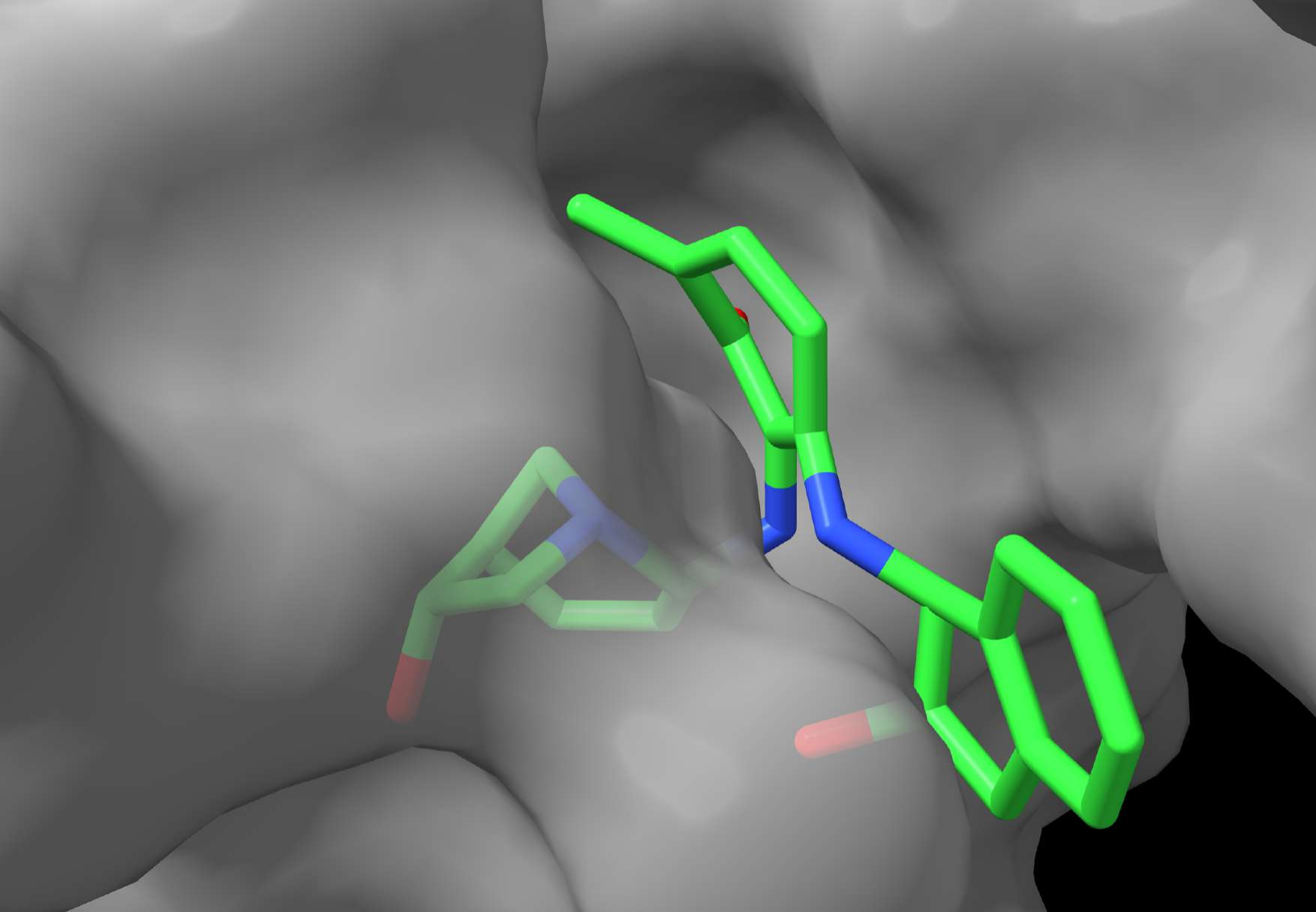}{0.447}
&
\figimg{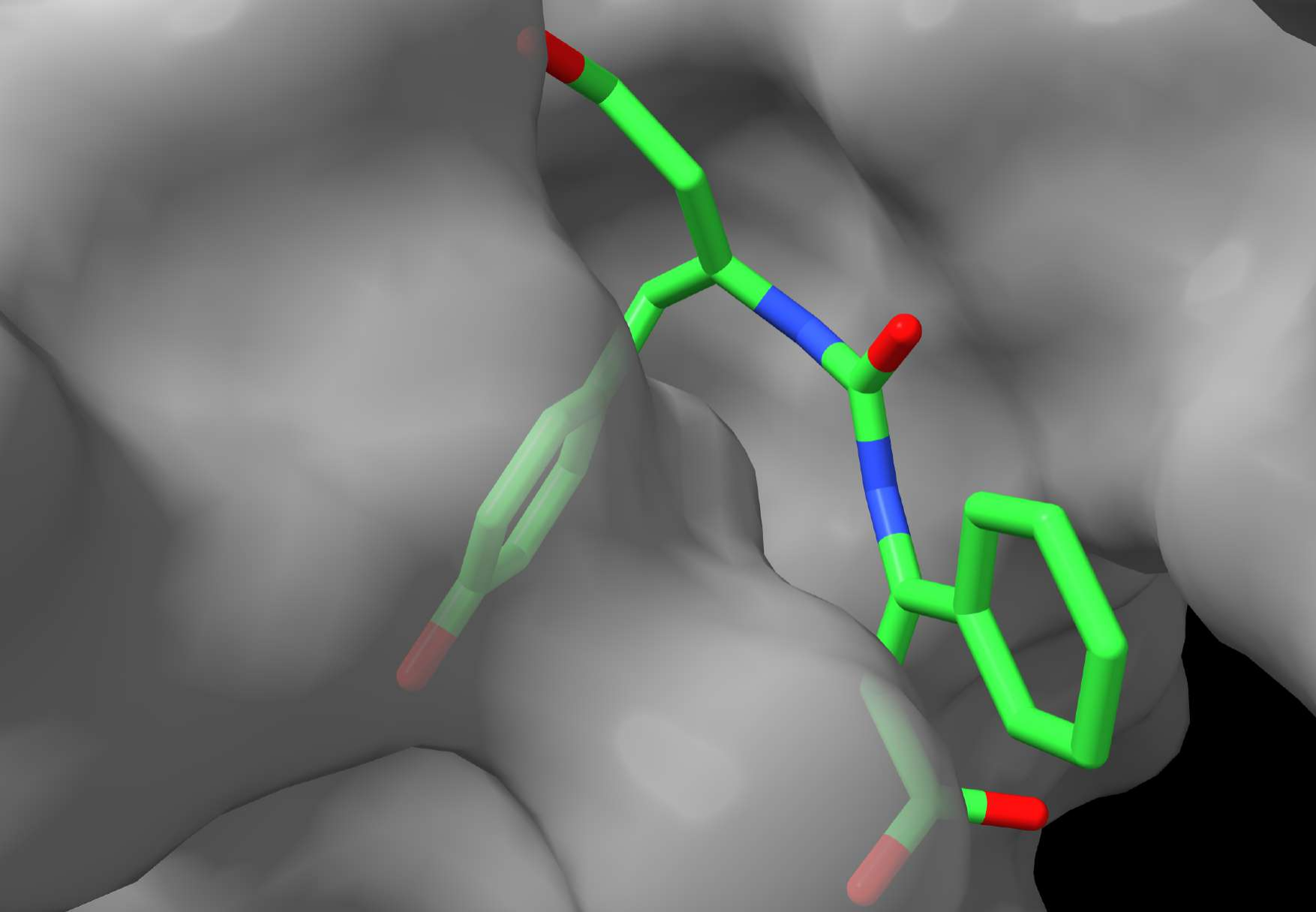}{0.474}
&
\figimg{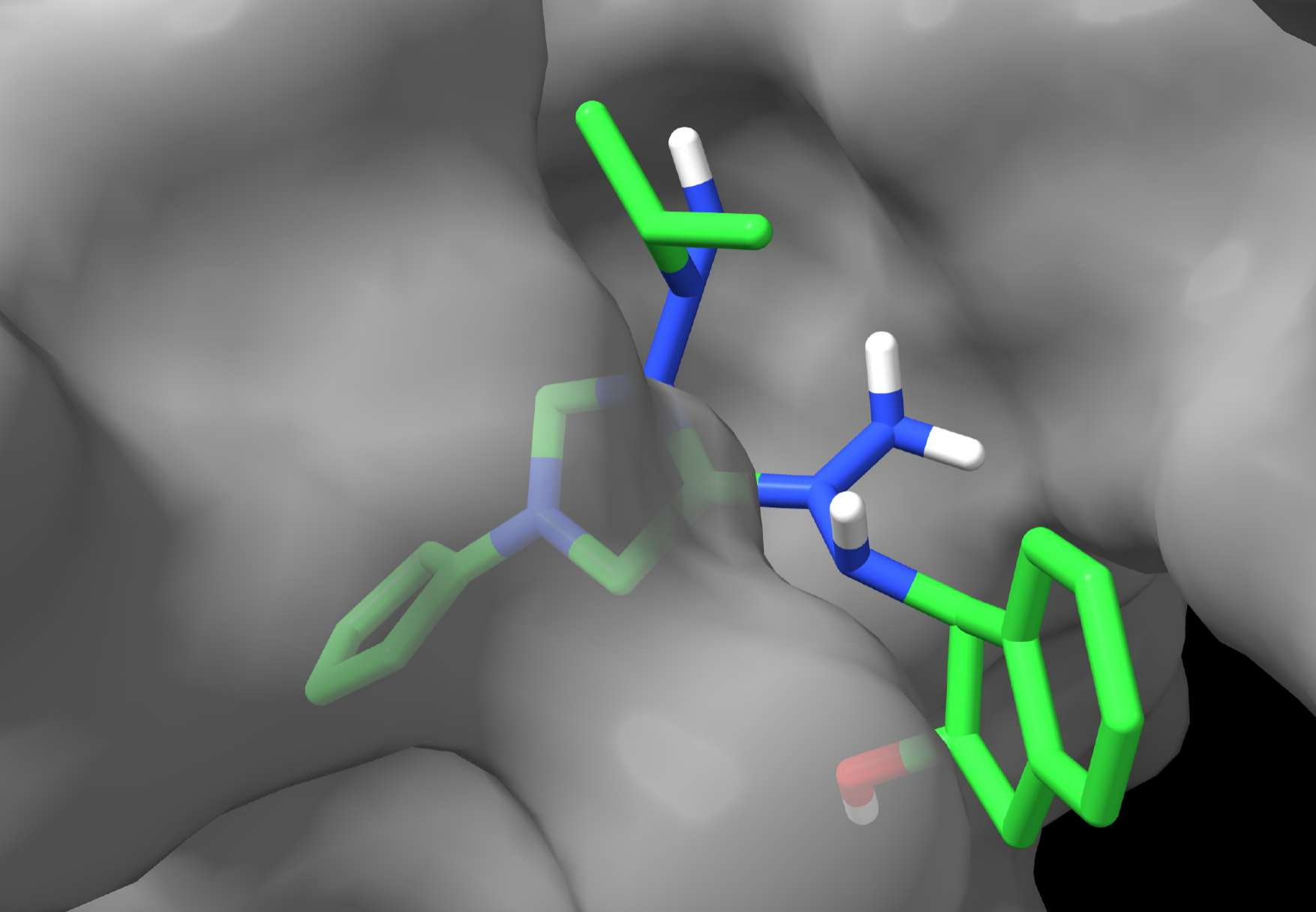}{0.484}
&
\figimg{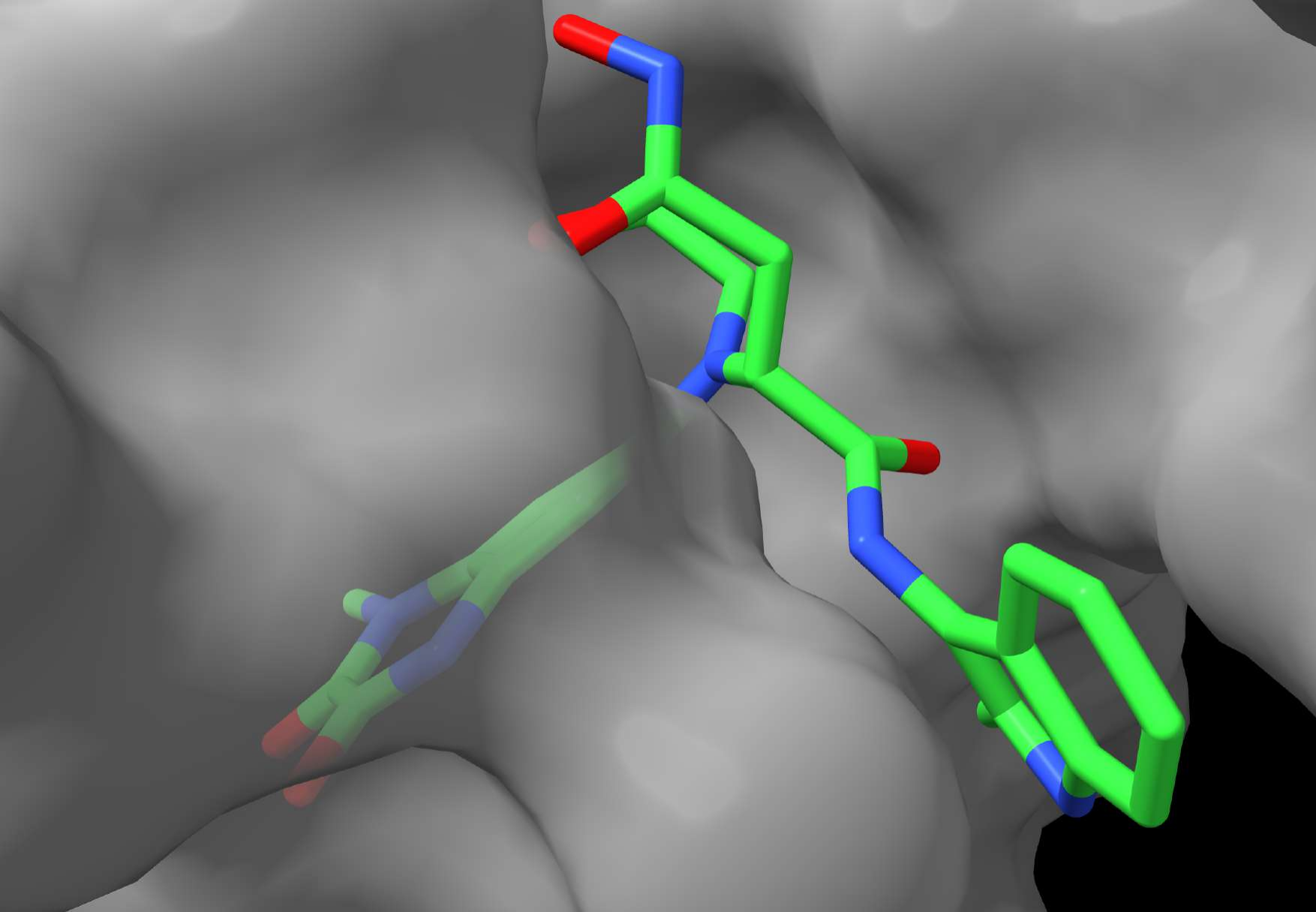}{0.452}
\\[2pt]

\rowlab{BAPA}
&
\figimg{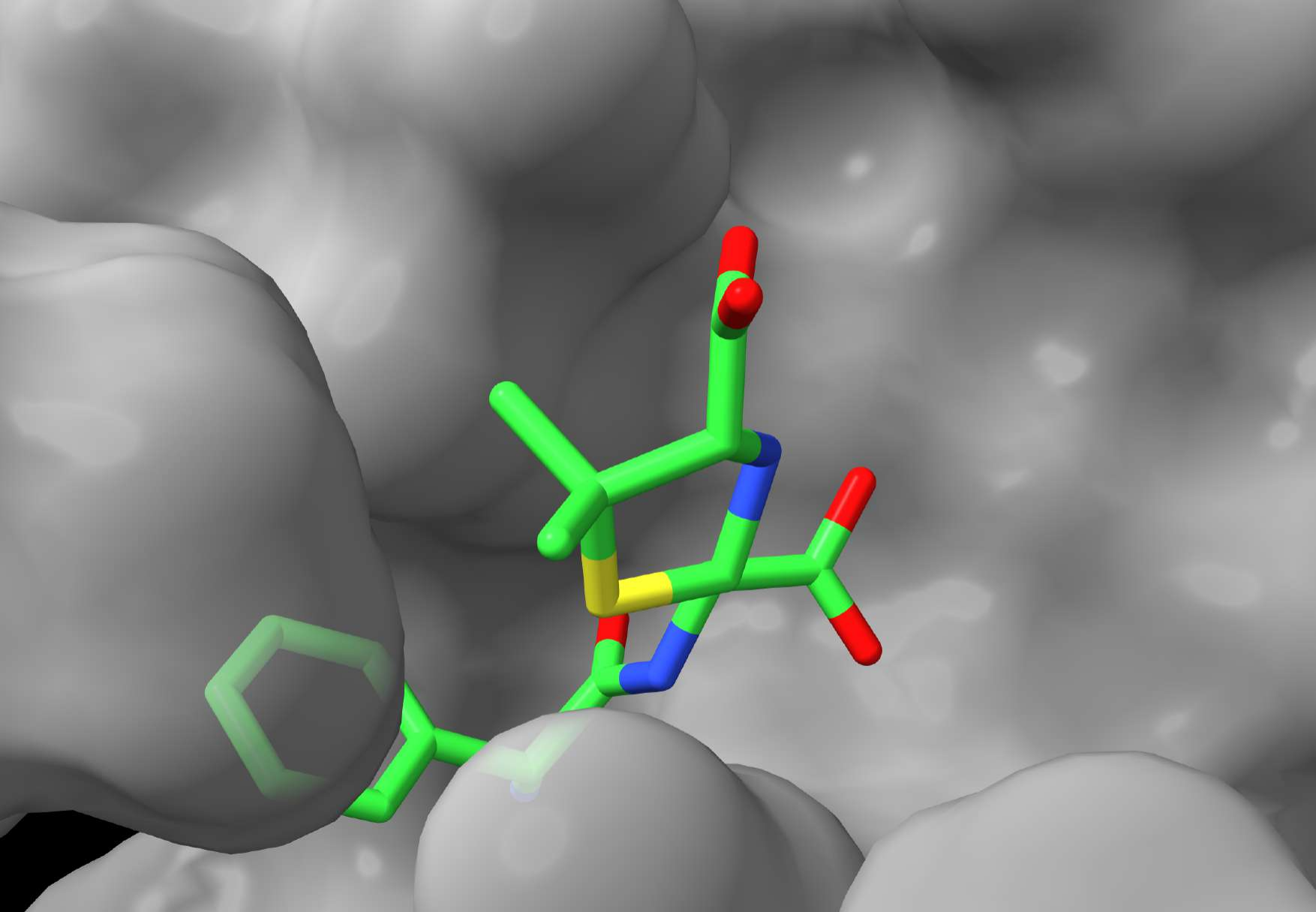}{0.500}
&
\figimg{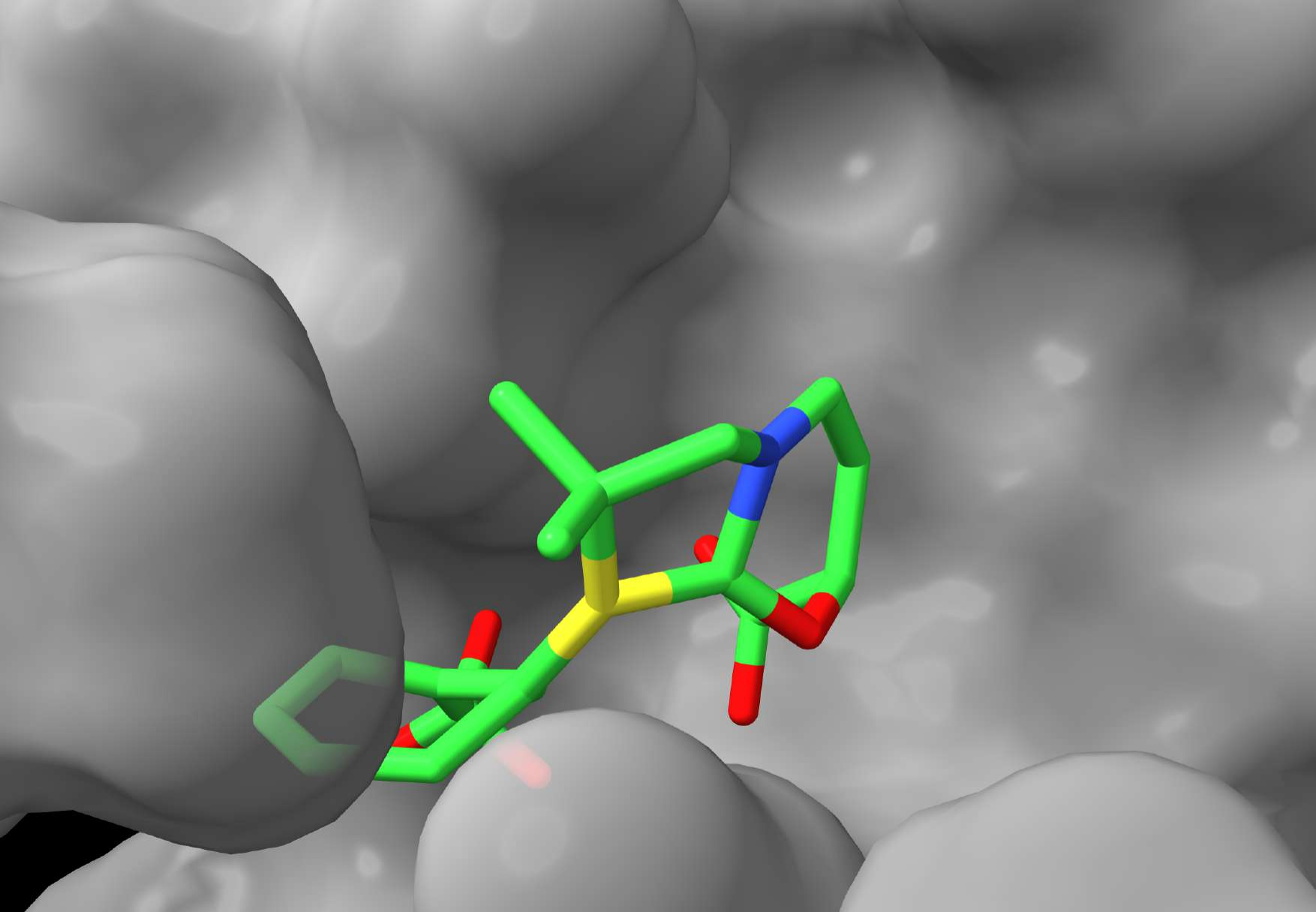}{0.464}
&
\figimg{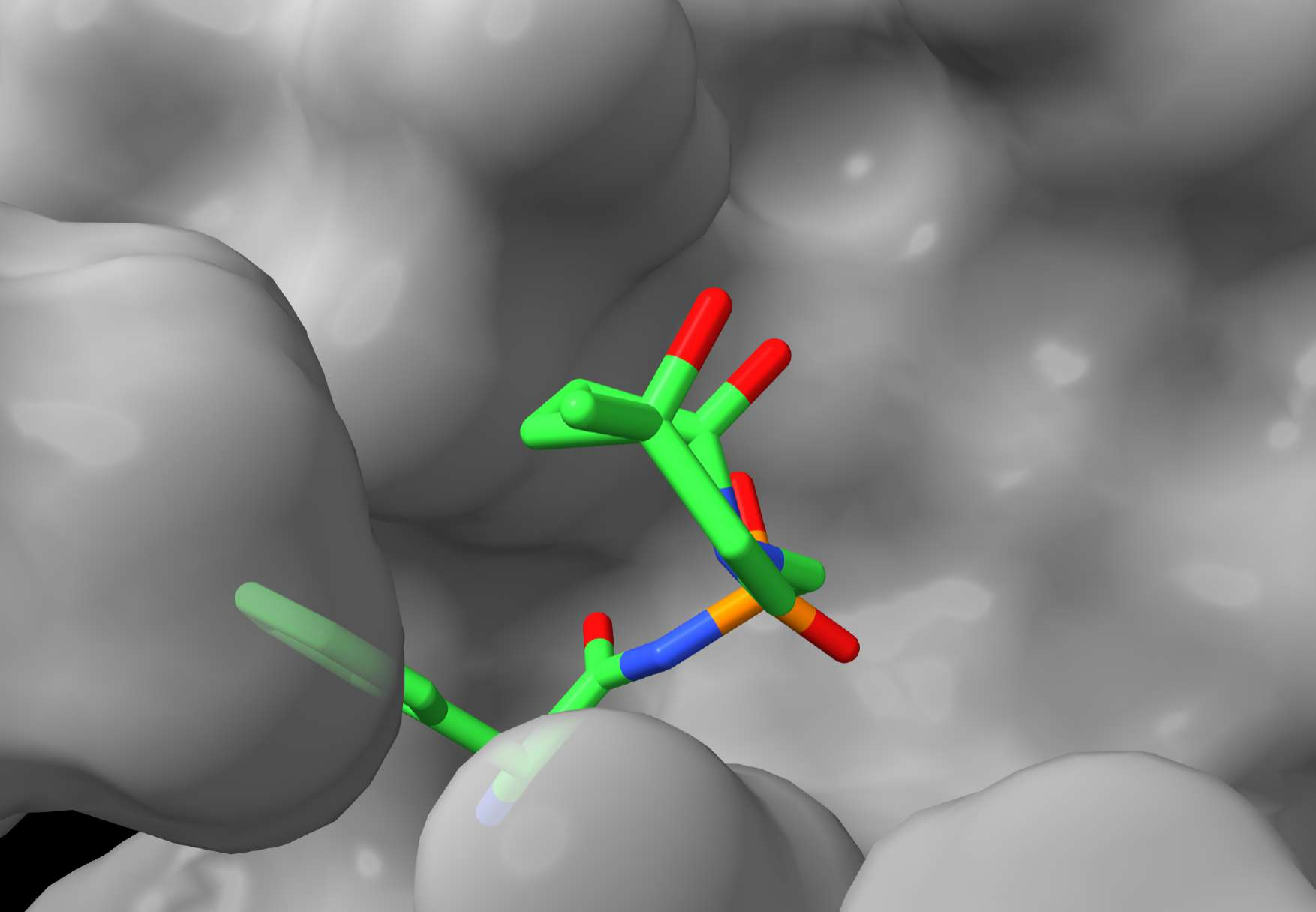}{0.508}
&
\figimg{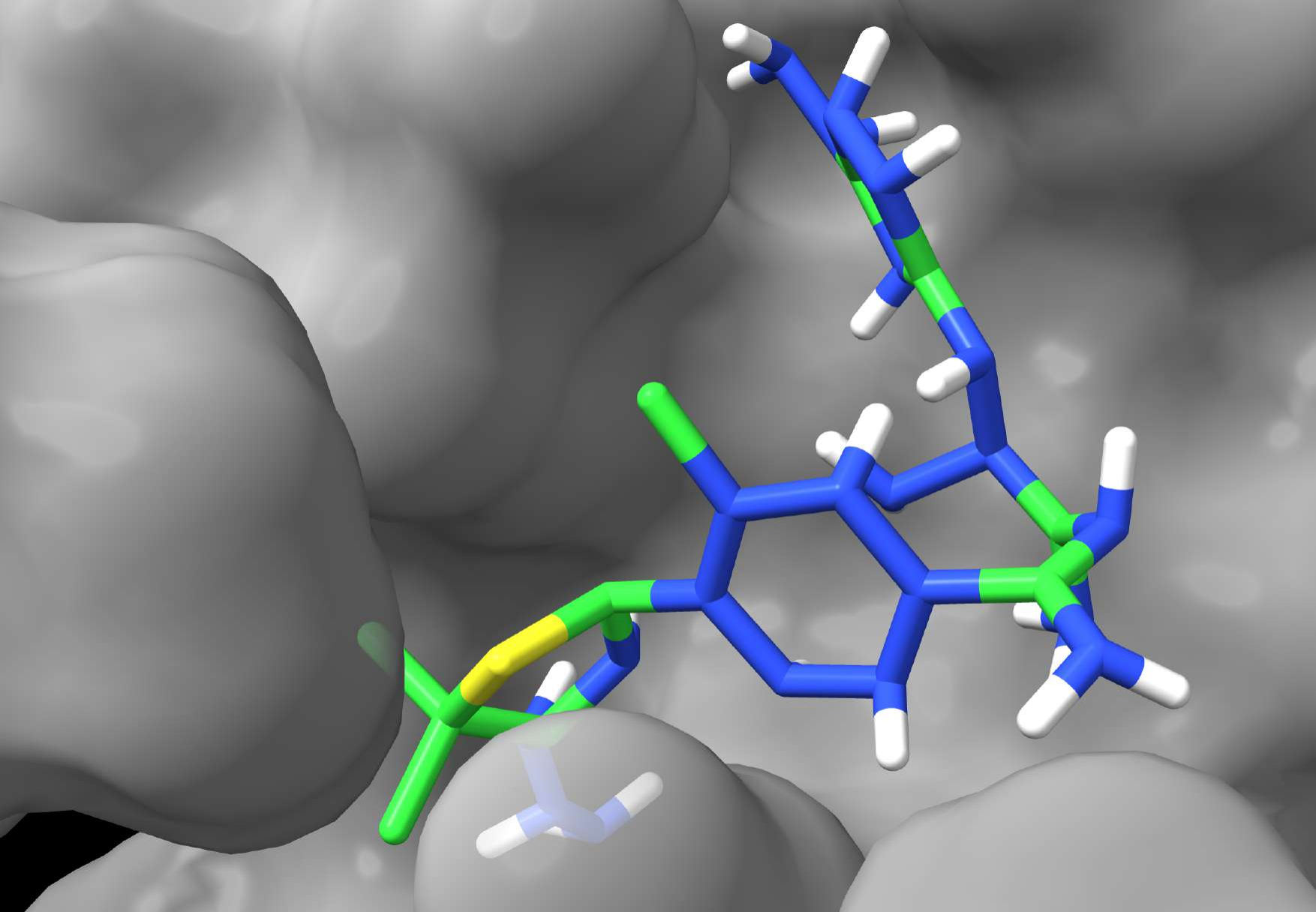}{0.514}
&
\figimg{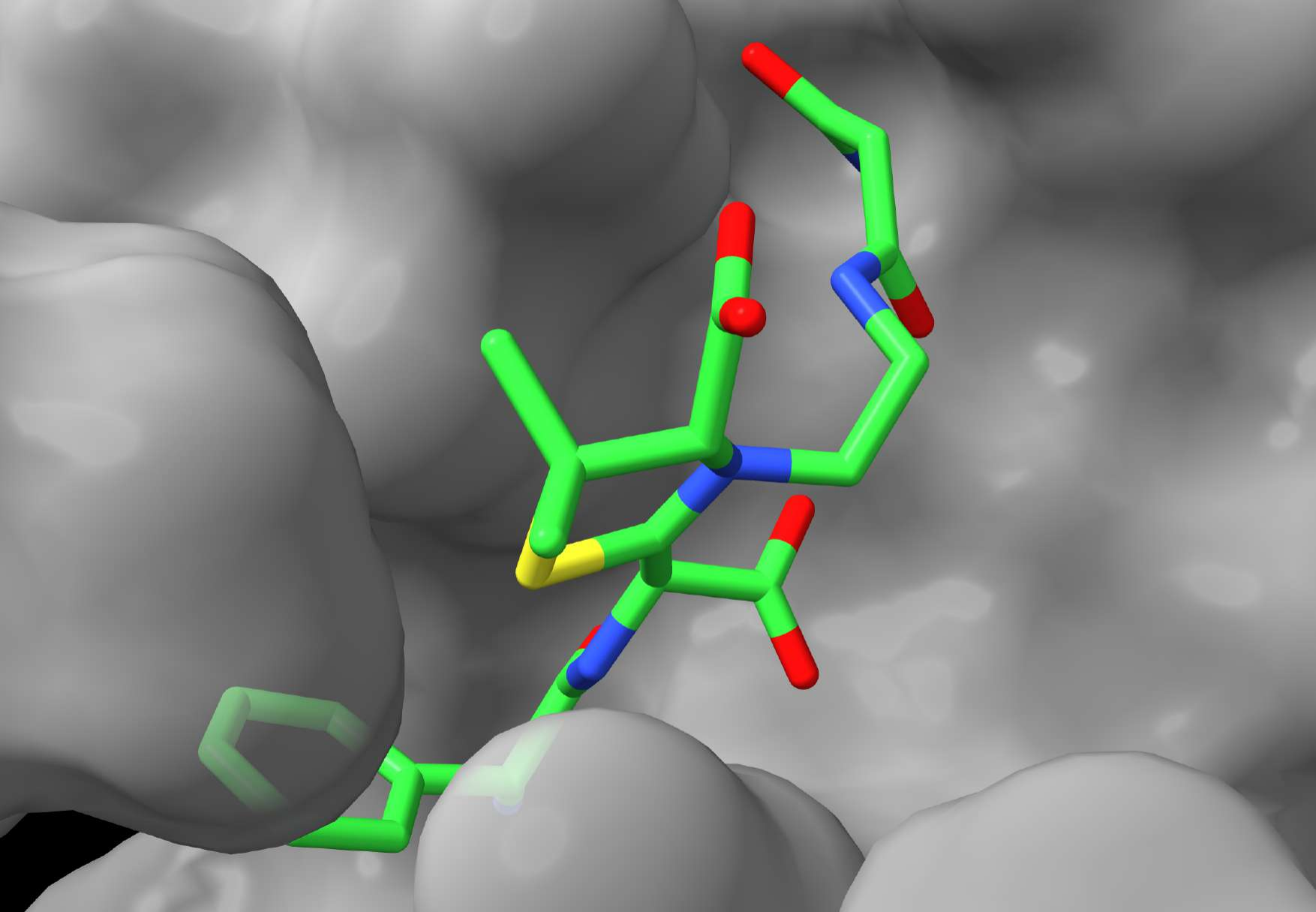}{0.461}

\end{tabular}
}

\caption{
Qualitative comparison of lead optimization results on three representative
target pockets.
Rows correspond to target pockets, and columns correspond to the initial
lead and optimized ligands generated by different methods.
The gray surface denotes the target pocket, and the inset value reports
the target-pocket geometric mismatch \(d_{\mathrm{gm}}\), where lower
values indicate better surface complementarity.
SurfSpec achieves low geometric mismatch across the shown examples while
maintaining plausible pocket-bound conformations.
}
\label{fig:qual_process_comparison}

\endgroup
\end{figure*}


\begin{figure}[t]
\centering
\includegraphics[
    width=0.5\columnwidth,
    clip,
    keepaspectratio
]{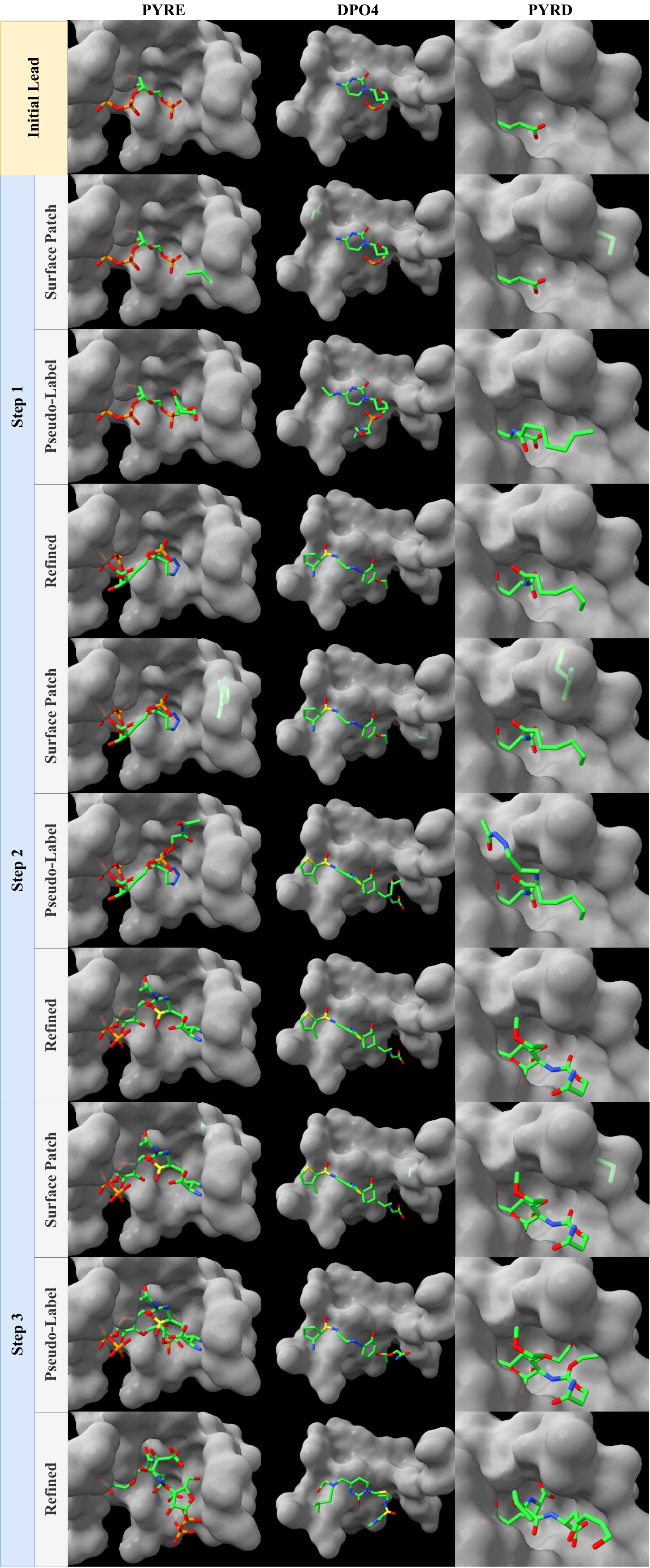}
\caption{
Qualitative visualization of SurfSpec ligand growth on representative
CrossDocked2020 complexes.
Columns show different target pockets, and rows show the initial lead and
the intermediate states across three growth iterations.
At each step, SurfSpec first selects an under-occupied target-surface
patch, constructs a surface-directed pseudo-label toward that patch, and
then refines the pseudo-label under the pocket-conditioned ligand prior.
The gray surface denotes the target pocket, the green sticks denote the
current ligand, and the pale green sticks indicate the selected
surface-directed growth region.
SurfSpec progressively expands the ligand toward uncovered pocket regions
while maintaining plausible pocket-bound conformations after refinement.
}
\label{fig:ligand_growth.pdf}
\end{figure}

\end{document}